\documentclass{daj}

\usepackage[utf8]{inputenc}

\usepackage{color,amsmath,amsfonts,fullpage,amsthm,mathrsfs,fancyhdr,verbatim,framed,algorithmic,algorithm,hyperref}
\usepackage[capitalize]{cleveref}
\usepackage{thm-restate}
\usepackage{hyperref}
\usepackage{graphicx, wrapfig}
\usepackage{IEEEtrantools}

\definecolor{DarkBlue}{rgb}{0,0,0.8}  
\definecolor{DarkOrange}{rgb}{0.8,0.4,0}  

\allowdisplaybreaks

\newcommand{\calO}{\mathcal{O}}
\newcommand{\mH}{\mathcal{H}}

\newcommand{\wG}{\widehat{G}}
\newcommand{\wH}{\widehat{H}}
\newcommand{\indicator}{\mathbb{I}}

\newcommand{\outc}{\mathsf{Outcomes}}

\newcommand {\C} {\mathbb{C}}
\newcommand {\D} {\mathcal{D}}

\newcommand {\F} {\mathbb{F}}
\newcommand {\N} {\mathbb{N}}
\newcommand {\R} {\mathbb{R}}
\newcommand{\eps}{\varepsilon}
\newcommand{\poly}{\mathrm{poly}}

\newcommand{\Paren}[1]{\left ( #1 \right) }

\newcommand{\ket}[1]{{|#1\rangle}}
\newcommand{\bra}[1]{{\langle#1|}}

\newcommand{\ketbra}[2]{|#1\rangle\! \langle #2|}
\newcommand{\tr}{\mbox{\rm tr}}
\newcommand{\Id}{\mathbf{1}}
\newcommand{\id}{\mathbf{1}}
\newcommand{\hilb}{\mathcal{H}}

\newcommand{\norm}[1]{\left \| #1 \right \|}

\newcommand{\abs}[1]{\left\vert {#1} \right\vert}
\newcommand{\ot}{\otimes}

\renewcommand{\cal}[1]{\mathcal{#1}}
\DeclareMathOperator*{\E}{\mathbb{E}}

\newcommand{\Bounded}{\mathrm{B}}

\newcommand{\trace}[1]{\tau \left ( #1 \right)}

\newcommand{\code}{\cal{C}}
\newcommand{\distinct}{\mathsf{distinct}}
\newcommand{\codemt}{{\left(\code^{\otimes{m}}\right)}^t}
\newcommand{\interpol}{\phi}

\newcommand{\MIP}{\mathsf{MIP}}
\newcommand{\NEXP}{\mathsf{NEXP}}
\newcommand{\RE}{\mathsf{RE}}

\newcommand{\strategy}{\mathscr{S}}
\newcommand{\algebra}{\mathscr{A}}
\newcommand{\mA}{\mathcal{A}}
\newcommand{\proj}[1]{\ket{#1}\!\bra{#1}}

\newcommand{\hnote}[1]{}
\newcommand{\tnote}[1]{}
\newcommand{\jnote}[1]{}
\newcommand{\znote}[1]{}
\newcommand{\anote}[1]{}

\newtheorem{theorem}{Theorem}[section]
\newtheorem{proposition}[theorem]{Proposition}

\newtheorem{lemma}[theorem]{Lemma}
\newtheorem{claim}[theorem]{Claim}

\newtheorem{corollary}[theorem]{Corollary}

\newtheorem{remark}[theorem]{Remark}

\theoremstyle{definition}
\newtheorem{definition}[theorem]{Definition}

\dajAUTHORdetails{%
  title = {Quantum Soundness of Testing Tensor Codes}, 
  author = {Zhengfeng Ji, Anand Natarajan, Thomas Vidick, John Wright and Henry Yuen},
  plaintextauthor = {Zhengfeng Ji, Anand Natarajan, Thomas Vidick, John Wright and Henry Yuen},
  plaintexttitle = {Quantum Soundness of Testing Tensor Codes}, 
  keywords = {quantum multiplayer games, tensor codes, entanglement tests},
}   

\dajEDITORdetails{%
   year={2022},
   number={17},
   received={8 February 2022},   
   published={6 December 2022},  
   doi={10.19086/da.55554},       
}   

\begin{document}

\begin{frontmatter}[classification=text]

\title{Quantum Soundness of Testing Tensor Codes} 


\author[zji]{Zhengfeng Ji\thanks{Supported by Australian Research Council (DP200100950). Conducted this work while at University of Technology, Sydney.}}
\author[anand]{Anand Natarajan}
\author[thomas]{Thomas Vidick\thanks{Supported by MURI Grant FA9550-18-1-0161, NSF QLCI Grant OMA-2016245, and the IQIM, an NSF Physics Frontiers Center (NSF Grant PHY-1125565) with support of the Gordon and Betty Moore Foundation (GBMF-12500028).}}
\author[john]{John Wright}
\author[henry]{Henry Yuen\thanks{Supported by AFOSR
award FA9550-21-1-0040, NSF CAREER award CCF-2144219, and the Sloan Foundation.}}

\begin{abstract}
A locally testable code is an error-correcting code that admits very efficient probabilistic tests of membership. Tensor codes provide a simple family of combinatorial constructions of locally testable codes that generalize the family of Reed-Muller codes. The natural test for tensor codes, the axis-parallel line vs.\ point test, plays an essential role in constructions of probabilistically checkable proofs. 

We analyze the axis-parallel line vs.\ point test as a two-prover game and show that the test is sound against quantum provers sharing entanglement. Our result implies the quantum-soundness of the low individual degree test, which is an essential component of the $\MIP^* = \RE$ theorem. Our proof generalizes to the infinite-dimensional \emph{commuting-operator model} of quantum provers.\end{abstract}
\end{frontmatter}


\newcommand{\RM}{\mathcal{RM}}

\section{Introduction}

Let $\Sigma$ denote a finite field and let $\code$ denote a subspace of the vector space $\Sigma^n$ of dimension $k \leq n$. 
We say that $\code$ is a \emph{(linear) code over alphabet $\Sigma$ with blocklength $n$ and dimension $k$}. Viewing the elements $c \in \code$ (called \emph{codewords}) as  functions mapping $[n] = \{1,2,\ldots,n\}$ to $\Sigma$, if it further holds that for any two distinct $c,c' \in \code$ the number of $i \in [n]$ such that $c(i) \neq c'(i)$ is at least $d$ then we say that $\code$ has \emph{distance $d$}. 

A natural method to build large codes out of smaller codes is to take their
\emph{tensor product}.
Let $\code$ denote a code (called the \emph{base code}) over alphabet $\Sigma$
with blocklength $n$, dimension $k$, distance $d$.
The \emph{tensor code} $\code^{\otimes m}$ is defined as the collection of all
functions $c:[n]^m \to \Sigma$ such that for all
$(u_1,\ldots,u_m) \in [n]^m$, for all $j \in [m]$, the function
$g(s) = c(u_1,\ldots,u_{j-1},s,u_{j+1},\ldots,u_m)$ is a codeword of $\code$
(see Definition~\ref{def:axis-line} and Definition~\ref{def:tensor-code}).
The blocklength of $\code^{\otimes m}$ is $n^m$, its dimension is $k^m$, and its
distance is at least $d^m$.
One can also consider taking tensor products of \emph{different} codes
$\code_1 \otimes \code_2 \otimes \cdots$, but we focus our attention on tensor
codes of the form $\code^{\otimes m}$.

Tensor codes have been extensively studied in theoretical computer science because they provide natural constructions of \emph{locally testable codes}~\cite{goldreich2006locally,ben2006robust,dinur2006robust,viderman2015combination}. A code $\code$ is locally testable if it has a \emph{tester} algorithm that makes a small number of queries to a given $w: [n] \to \Sigma$, accepts if $w$ is a codeword of $\code$, and rejects with noticeable probability if $w$ is far from every codeword. In other words the code has  efficient probabilistic tests for membership. Locally testable codes are central components of probabilistically checkable proofs (PCPs); in particular improvements in PCPs have been continually driven by the development of improved \emph{low-degree tests}, which are tests for codes based on low-degree polynomials~\cite{arora1998probabilistic,arora1998proof,raz1997sub,moshkovitz2008two}. 

In the context of PCPs it is common to present a testing procedure for a code as a game between a \emph{referee} and two non-communicating \emph{provers}. The goal of the referee is to determine, by making a few queries to the provers, whether the provers' responses are globally consistent with a codeword from the code. Let $\code^{\otimes m}$ denote a tensor code; it has a natural two-prover test (called the \emph{tensor code test}) associated to it:
\begin{enumerate}
	\item The referee samples a uniformly random $u \sim [n]^m$
	and sends $u$ to prover A. They respond with a value $a \in \Sigma$.
	\item The referee samples a uniformly random $j \sim [m]$ and sends the axis-parallel line \\ $\ell = \{ (u_1,\ldots,u_{j-1}, s, u_{j+1},\ldots,u_m) : s \in [n] \}$ to prover B. They respond with a codeword $g \in \code$.
	\item The referee accepts iff $g(u_j) = a$. 
\end{enumerate}
If $c \in \code^{\otimes m}$ is a codeword and the provers respond with $c(u)$ and $c|_\ell$ (i.e.\ the restriction of $c$ to the line $\ell$) then they will pass with probability $1$. \emph{Soundness} of of the tensor code test refers to the converse statement, namely that provers who succeed with high probability must be (approximately) responding according to some codeword $c \in \code^{\otimes m}$. This is formalized as follows.

\begin{theorem}
\label{thm:ldt-classical}
Let $\code$ be an interpolable\footnote{The \emph{interpolable} qualifier on the
  base code $\code$ is a technical condition stating that for every $t=n-d+1$
  coordinates $i_1,\ldots,i_t \in [n]$ and values $a_1,\ldots,a_t \in \Sigma$
  there exists a unique codeword $c \in \code$ such that $c(i_j) = a_j$ for
  $j \in [t]$.
  We discuss the interpolability condition in more detail in
  \Cref{sec:proof-overview}.}
code.
Suppose provers A and B are deterministic and pass the tensor code test
corresponding to $\code^{\otimes m}$ with probability $1 - \eps$.
Let $f: [n]^m \to \Sigma$ denote the responses of prover A.
Then there exists a codeword $c \in \code^{\otimes m}$ such that
\[
	\Pr_{u \sim [n]^m} [f(u) = c(u)] \geq 1 - \poly(m, t) \cdot (\poly(\eps) + \poly(1/n))\;,
\]
where $t = n - d + 1$. 
\end{theorem}

The proof of \Cref{thm:ldt-classical} follows almost directly from the analysis of the low-degree test of Babai, Fortnow and Lund~\cite{babai1991non}, a central ingredient of their characterization $\MIP = \mathsf{NEXP}$ of (classical) multiprover interactive proof systems. They analyzed the tensor code test for the case when the base code $\code$ is the set of univariate polynomials over a finite field $\F$ of degree at most $s$, and the tensor code $\code^{\otimes m}$ is the set of $m$-variate polynomials over $\F$ with \emph{individual} degree at most $s$ (i.e.\ each variable has degree at most $s$). \footnote{Technically speaking, Babai, Fortnow and Lund analyzed the low-degree test for \emph{multilinear} polynomials (i.e.\ polynomials with individual degree at most $1$), but as noted in Remark 5.15 of~\cite{babai1991non}, their analysis generalizes to polynomials with larger individual degree.} By noting that their proof only uses the tensor code structure of low individual degree polynomials and the fact that univariate polynomials with degree at most $s$ form a code with distance $n - s$ we see that their proof implies \Cref{thm:ldt-classical}.

The analysis of Babai, Fortnow, and Lund's low-degree test in the \emph{quantum} setting has played a major role in the study of the class $\MIP^*$ of languages having quantum multi-prover interactive proofs~\cite{ito2012multi,vidick2016three,natarajan2018two,natarajan2018low,natarajan2019neexp}.
In this setting, the provers are no longer modeled as deterministic but instead are allowed to use a \emph{quantum strategy}, in which their responses are generated by performing measurements on a shared entangled state. While entanglement does not allow the provers to communicate with each other, it gives rise to stronger correlations that cannot be achieved classically. Thus the classical analysis does not extend in any direct way. 

Quantum soundness of a variant of the multilinearity test (i.e.\ the case $s=1$) played with three entangled provers is at the heart of the proof of $\NEXP\subseteq\MIP^*$~\cite{ito2012multi}. This analysis is extended to general degree $s$ in~\cite{vidick2016three}, and to the case of two provers in~\cite{natarajan2018two}. However, it was later discovered~\cite{vidick2020erratum} that the argument contained in~\cite{vidick2016three} contains a mistake. The case of two provers and general $s$ is essential in 
 the proof of the recent characterization $\MIP^* = \RE$, which shows that every recursively enumerable language  (including the Halting problem) has an interactive proof in the entangled provers model~\cite{ji2020mip}. Our main contribution is a proof of soundness of the low individual degree test that is used in the proof of $\MIP^*=\RE$.\footnote{The present paper is an updated version of an earlier arXiv posting by the same authors~\cite{ji2020quantum}, which was posted after the mistake in the earlier works had been discovered. We discuss additional improvements compared to that earlier draft below. 
 } We show that quantum provers who succeed with high probability in the test must still -- in a certain sense -- respond according to a low individual degree polynomial.
Informally, our main result is the following:

\begin{theorem}[Quantum soundness of the tensor code test, informal]
\label{thm:main-informal}
Let $\code$ be an interpolable code. 
Suppose provers A and B are quantum and pass the (augmented \footnote{The qualifier ``augmented'' refers to the addition of a subtest which is used to enforce approximate commutation relations between the measurement operators used by prover A. This is a technical point that we discuss in more detail in \Cref{sec:proof-overview}. }) tensor code test corresponding to $\code^{\otimes m}$ with probability $1 - \eps$. Let $\{A^u \}_{u \in [n]^m}$ denote the ``points'' measurements of prover A, who upon receiving question $u \in [n]^m$ performs measurement $A^u$ to obtain outcome $a \in \Sigma$. Then there exists a measurement $G$ with outcome set $\code^{\otimes m}$ satisfying the following properties:
\begin{itemize}
	\item (Self-consistency) If prover A and prover B both perform the measurement $G$, they both obtain the same outcome $c \in \code^{\otimes m}$ with probability at least $1 - \delta$.
	\item (Consistency with points) If prover A measures $A^u$ for a uniformly random point $u \in [n]^m$ to obtain an outcome $a \in \Sigma$ and prover B measures $G$ to obtain an outcome $c \in \code^{\otimes m}$, then $a = c(u)$ with probability at least $1 - \delta$.
\end{itemize}
Here, $\delta =\poly(m,t) \cdot \poly(\eps,1/n)$  with $t = n - d + 1$ .

\end{theorem}
The formal statement of this theorem is presented as \Cref{thm:main}. Intuitively the theorem states that even though the provers are allowed to employ quantum strategies, their strategy is essentially equivalent to a classical one: the provers would have roughly the same winning probability if they both measured $G$ (which does not depend on any question) to obtain a codeword $c \in \code^{\otimes m}$ and then responded according to $c$. In other words, the ``points'' measurements $\{A^u\}$ of prover A are -- despite not necessarily commuting with each other -- approximately consistent with a single ``global'' codeword measurement $G$.

This paper generalizes our previous manuscript~\cite{ji2020quantum}  which established \Cref{thm:main-informal} for the low individual degree test. In addition we significantly simplify the earlier analysis of~\cite{ji2020quantum} by working in the framework of \emph{synchronous strategies} for the quantum provers. Synchronous strategies are a subclass of quantum strategies where the provers are assumed to use the same measurement operators whenever they receive the same question, and furthermore the entangled state that they share is a maximally entangled state. These assumptions allow us to shortcut many technical calculations that appear in~\cite{ji2020quantum}. Furthermore, it is without loss of generality since recent work by one of us establishes that soundness with synchronous strategies for a large class of two-prover games (including the tensor code test) can be translated back into soundness for general quantum strategies~\cite{vidick2021almost}. As a result, Theorem~\ref{thm:main-informal} extends in a straightforward manner to general strategies; see Theorem~\ref{thm:main-bipartite} for the precise statement. 

An additional benefit of our streamlined analysis is that synchronous strategies can also be defined in infinite dimensions. In this case the condition that the entangled state is maximally entangled is replaced by the condition that the state is \emph{tracial}, a notion which we introduce formally below. Our proof naturally extends to this situation, and we are thus able to prove soundness of the tensor code test against \emph{infinite-dimensional} (synchronous) quantum strategies. We explain the synchronous strategies framework in more detail in \Cref{sec:synchronous}. We motivate the consideration of infinite-dimensional strategies from a complexity-theoretic viewpoint in Section~\ref{sec:open}.

\subsection{Nonlocal games and synchronous strategies}
\label{sec:synchronous}

A \emph{nonlocal game} $G$ is specified by a quadruple $(\cal{X},\cal{A},\mu,D)$ which corresponds to the following scenario: a referee first samples a pair of questions $(x,y)$ from a distribution $\mu$ over a finite product set $\cal{X} \times \cal{X}$ and sends $x$ to prover A and $y$ to prover B. The provers respond with answers $a,b \in \cal{A}$ respectively and the referee accepts if and only if $D(x,y,a,b) = 1$, where $D: \cal{X} \times \cal{X} \times \cal{A} \times \cal{A} \to \{0,1\}$ is a decision predicate. Different kinds of restrictions on the provers' allowed actions to determine their answers lead to different optimal success probabilities for them in the game. 
The tensor code test described earlier is an example of a nonlocal game, and the analysis of it with deterministic non-communicating provers is a key component in the aforementioned complexity-theoretic results, e.g.~\cite{babai1991non}. 

A quantum strategy for a nonlocal game describes the provers' behavior during the game. We focus on the class of \emph{synchronous strategies}. A synchronous strategy $\strategy$ for a game $G$ 
is specified by a separable Hilbert space $\hilb$ (which could be infinite-dimensional), a von Neumann algebra $\algebra$ on $\hilb$, a tracial state $\tau$ on the algebra $\algebra$,\footnote{A \emph{von Neumann algebra} $\algebra$ on a Hilbert space $\hilb$ is a $*$-subalgebra of $\Bounded(\hilb)$ (the set of bounded operators on $\hilb$) that contains the identity operator and is closed under the weak operator topology. A \emph{tracial state} $\tau$ on the algebra $\algebra$ is a positive, unital linear functional that satisfies the \emph{trace property}: $\trace{AB} = \trace{BA}$ for all $A,B \in \algebra$. 
}  and a collection of 
projective measurements $\{ M^x \}_{x \in \cal{X}}$ in $\algebra$ (each $M^x$ is a set of projections $\{M^x_a\}_{a \in \cal{A}}$ summing to the identity).\footnote{Unlike in the finite-dimensional setting, we cannot without loss of generality take $\algebra$ to be all of $\Bounded(\hilb)$; this is because a tracial state on $\algebra$ does not extend to a tracial state on $\Bounded(\hilb)$, which is not finite.} Given questions $(x,y)$, the probability of obtaining answers $(a,b)$ is given by $\tau(M^x_a \, M^y_b)$. Thus the probability that the strategy $\strategy$ succeeds in the game $G$ is given by
\[
	\sum_{x,y \in \cal{X}} \mu(x,y) \, \sum_{a,b \in \cal{A}} D(x,y,a,b) \, \tau \big( M^x_a \, M^y_b \big)~.
\]
Readers who are not familiar with von Neumann algebras and tracial states may find the finite-dimensional setting easier to understand. (No additional difficulty is posed by the infinite-dimensional setting; we include it because it is more general and, once one gets used to it, simpler.) When $\hilb= \C^r$ for some dimension $r$, then we can without loss of generality take the algebra $\algebra$ to be the set $\Bounded(\hilb)$ of all bounded operators on $\hilb$ (which in finite dimensions is simply the set of all linear operators). In this case there is a \emph{unique} tracial state, the normalized trace $\tau(X) = \frac{1}{r} \tr(X)$. In terms of strategies for nonlocal games, this corresponds to the provers using the projective measurements that are transposed from each other and sharing the maximally entangled state $\ket{\Phi} = \frac{1}{\sqrt{r}} \sum_{e = 1}^r \ket{e} \ket{e}$. (The transposition is because in general we have $\bra{\Phi} A \otimes B \ket{\Phi} = \frac{1}{r}\tr(AB^T)$.) Such a strategy has the property that if both provers receive the same question $x \in \cal{X}$ then they always output the same answer $a \in \cal{A}$ (this is why these strategies are called ``synchronous'').

In the infinite-dimensional setting synchronous strategies give rise to
\emph{commuting operator} strategies: as shown in \cite[Theorem
5.5]{paulsen2016estimating}, for every synchronous strategy
$\strategy = (\tau,\{M^x\})$ with Hilbert space $\hilb$ there exist another
Hilbert space $\hilb'$, a state $\ket{\psi} \in \hilb'$, and measurements
$\{A^x\}, \{B^x\}$ on $\hilb'$ for the provers respectively such that for all
$x,y \in \cal{X}$ and $a,b\in \cal{A}$, the operators $A^x_a$ and $B^y_b$
commute and we have
\[
	\tau(M^x_a \, M^y_b) = \bra{\psi} A^x_a \, B^y_b \ket{\psi}~.
\]
A consequence of the characterization $\MIP^*=\RE$ (specifically, of the fact
that $\MIP^*$ contains undecidable languages) is that 
there exists a game that can be won with probability $1$ using a
(synchronous) commuting-operator strategy but tensor product strategies (even
infinite-dimensional ones) cannot succeed with probability larger than
$\frac{1}{2}$.

Synchronous strategies arise naturally when considering \emph{synchronous games}: these are games where the provers must output the same answers whenever they receive the same question (i.e.\ for all $x$, $D(x,x,a,b) = 1$ if and only if $a = b$). Many games studied in quantum information theory and theoretical computer science are synchronous games; for example the games constructed in the proof of $\MIP^* = \RE$ are all synchronous. The augmented tensor code test (presented formally in \Cref{sec:test}) studied in this paper is also a synchronous game. 

It was shown by~\cite{paulsen2016estimating,kim2018synchronous} that if a synchronous game has a perfect strategy (i.e.\ a strategy that wins with probability $1$) then it also has a perfect synchronous strategy. Furthermore, if the former strategy is finite-dimensional, then so is the latter synchronous strategy. This characterization of perfect (finite-dimensional) strategies for synchronous games was recently extended to the case of near-perfect strategies in~\cite{vidick2021almost}: any finite-dimensional strategy for a synchronous game $G$ that succeeds with high probability can be well-approximated by a convex combination of finite-dimensional synchronous strategies. In particular, results such as Theorem~\ref{thm:main-informal} for synchronous strategies extend in a natural way to general strategies. 

\medskip

Besides the greater generality there are several advantages to focusing on synchronous strategies when analyzing the tensor code test. We list the three most important, in our opinion.  The first advantage is that it simplifies notation. In the standard formulation of strategies for nonlocal games, separate measurement operators are specified for both provers, and in the finite-dimensional setting they are associated to separate Hilbert spaces. With synchronous strategies there is  a single Hilbert space and a single set of measurement operators. Furthermore, in the finite-dimensional setting there is no need to specify the state used by the provers; it is assumed that they use \emph{the} maximally entangled state. (In infinite dimensions there can be multiple non-unitarily equivalent tracial states.)

The second advantage is that it simplifies certain routine steps that typically occur in analyses of quantum strategies. For example, an often-repeated argument is ``prover switching'': this is to argue that prover A's measurement operator $A^x_a$ corresponding to a question $x \in \cal{X}$ and answer $a \in \cal{A}$ can be approximately mapped to an analogous measurement operator on prover B's side. This allows for ``cycling'' of measurement operators in the following manner:
\[
	\bra{\psi} A^z_c \, A^y_b \, A^x_a \otimes \id \ket{\psi} \approx \bra{\psi} A^z_c \, A^y_b \otimes B^x_a \ket{\psi} \approx \bra{\psi} A^x_a \, A^z_c \, A^y_b \otimes \id \ket{\psi}~,
\]
where here the first approximation sees $A^x_a$ as acting on the state $\ket{\psi}$ on the right-hand side, while the second approximation sees $B^x_a$ as acting on the state $\bra{\psi}$ on the left-hand side. 
Though fundamentally elementary, such steps can obfuscate the argument in a proof and are tedious to verify. With synchronous strategies prover switching comes ``for free'' because of the tracial property $\tau( A^z_c \, A^y_b \, A^x_a) = \tau( A^x_a \, A^z_c \, A^y_b)$.

The third advantage is the operator-algebraic implications of $\MIP^* = \RE$ can be still be obtained by focusing on synchronous strategies. In fact, as shown by Dykema and Paulsen~\cite{dykema2016synchronous} the connection with Connes' Embedding Problem (CEP) arises arguably more naturally in the synchronous setting, e.g.\ a negative answer to CEP is equivalent to whether there is a nonlocal game $G$ such that the optimal success probability with finite-dimensional synchronous strategies is strictly smaller than the optimal success probability with infinite-dimensional synchronous strategies.

\subsection{Proof overview}
\label{sec:proof-overview}

At a high level our proof follows the approach of~\cite{babai1991non}, and
it is useful to start by summarizing their analysis, rephrased in terms of tensor codes rather than low-degree polynomials.

In the setting of~\cite{babai1991non}, we consider classical (i.e.\ deterministic) strategies: the provers' strategy is described
by a ``points function,'' assigning a value in $\Sigma$ to each point 
$u \in [n]^m$, and a
``lines function,'' assigning a codeword of $\code$ to each
axis-parallel line $\ell \subset [n]^{m}$ queried in the test. From the assumption that the
points and lines functions agree at a randomly chosen point with high
probability, we would like to construct a ``global'' codeword of $\code^{\otimes m}$ that has high agreement with
the ``local'' points function, on average over a uniformly random $u$ at which both are evaluated. This is done by inductively constructing
``subspace functions'' defined on axis-aligned ``affine subspaces'': subsets of $[n]^m$ where $m-k$ coordinates have been fixed, for some $1 \leq k \leq m$. We refer to the parameter $k$ as the \emph{dimension} of the subspace. As special cases, a subspace of dimension $k = 0$ is a single point in $[n]^{m}$, a subspace of dimension $k =1$ is an axis-aligned line, and a subspace with dimension $k = m$ is the entire space $[n]^m$.

The induction proceeds on the dimension parameter $k$. The base case is $k=1$, and the ``global'' codeword is automatically supplied by the lines
function. At each step, to construct the subspace function for a
subspace $S$ of dimension $k+1$, we pick $t$ parallel subspaces of
dimension $k$ that lie within $S$, for some large enough $t$, and compute the unique function in $\code^{\otimes(k+1)}$ that interpolates them. The analysis shows that at each
stage the function constructed through interpolation has high
agreement on average with the points function and lines function. In
the end, when we reach the ambient dimension $k=m$, the subspace
function we construct is the desired global codeword.

\paragraph{Interpolation}

Before delving into more detail about the analysis, we first discuss a required property of the codes considered in this paper, which is the ability to uniquely \emph{interpolate} a codeword from a few positions. Let $\code$ be a code with blocklength $n$ and distance $d$. 
Let $t = n - d + 1$. The distance property of the code implies that given $t$ coordinates $i_1,\ldots,i_{t} \in [n]$ and values $a_1,\ldots,a_{t} \in \Sigma$, there exists \emph{at most} one codeword $c \in \code$ such that $c(i_j) = a_j$ for $j \in [t]$. We say that the code $\code$ is \emph{interpolable} if for all $i_1,\ldots,i_t \in [n]$ and $a_1,\ldots,a_t \in \Sigma$ there \emph{always} exists a codeword $c$ with the prescribed values $c(i_j) = a_j$. 

Interpolability is a direct generalization of the same property of polynomials. The code $\code$ consisting of the evaluations of all degree-$s$ polynomials over a finite field $\F$ is a code with blocklength $|\F|$ and distance $|\F| - s$. Given $t = s + 1$ points and values, there always exists a degree $s$ polynomial interpolating through them. 

Not all codes are interpolable; for example consider a code where some of the coordinates are always fixed to a predetermined symbol. However, a wide class of error-correcting codes are interpolable, including all polynomial codes. 

In this paper, we assume all base codes $\code$ are linear (meaning that they form a subspace of a vector space over a finite field) and interpolable. These two properties ensure that the tensor codes $\code^{\otimes m}$ are also linear interpolable.


\paragraph{The classical zero-error case}
We return to giving an overview of the proof. To start building intuition, it is useful to think about how to carry
out the above program in a highly simplified setting: the classical
zero-error case, for $m=3$ and $d = n - 1$. In this case, we assume that we have access to a
points function $f: [n]^m \to \Sigma$ and lines function $g: L_{n,m} \to \code$ that \emph{perfectly} pass the tensor code test, where $L_{n,m}$ denotes the set of axis-parallel lines in $[n]^m$. Moreover, we will focus on the \emph{final} step of the
induction: thus, we assume that we have already constructed a set of
planes functions defined for every axis-parallel plane that are
perfectly consistent with the line and point functions. In the final
step of the induction our goal is to combine these planes functions
to create a single global codeword $h \in \code^{\otimes 3}$ that is consistent with the points and lines functions. 

To do this we  interpolate the planes as follows. Let us label
the 3 coordinates in the space $x$, $y$, and $z$, and consider planes
parallel to the $(x,y)$-plane. Each plane $S_z$ is specified by a
value of the $z$ coordinate:
\[ S_z = \{(x,y,z): (x, y) \in [n]^2\}. \]
By the induction hypothesis, for every such plane $S_z$ there exists a
function $g_{z}: [n]^{2} \to \Sigma$ belonging to the code $\code^{\ot 2}$ that agrees with the
points function. To construct a global function
$h: [n]^3 \to \Sigma$ belonging to $\code^{\ot 3}$, we pick two distinct values $z_1 \neq z_2$ and ``paste''
the two plane functions $g_{z_1} $ and
$g_{z_2}$ together using the interpolability of the code. Specifically, we define $h$ to be the unique
function in $\code^{\ot 3}$ that interpolates between $g_{z_1}$ on the
plane $S_{z_1}$ and $g_{z_2}$ on the plane $S_{z_2}$. This interpolation is guaranteed to exist because the code $\code$ (and thus $\code^{\otimes 3}$) is linear interpolable. 

Why is $h$ consistent with the points function? To show this we need to consider the lines function on
lines parallel to the $z$ axis. Given a point $(x,y,z)$, let $\ell$ be
the line parallel to the $z$ axis through this point, and let $g_\ell$
be the associated lines function. By construction, $h$ agrees with
$g_\ell$ at the two points $z_1$ and $z_2$. But $g_\ell$ and $h|_{\ell}$ (the restriction of $h$ to $\ell$) are both codewords in $\code$, and hence if
they agree at two points, they must agree \emph{everywhere}. (This is because the distance of the code $\code$ is $d = n-1$.) Thus, $h$
agrees with $g_\ell$ at the original point $z$ as well. By success in
the test, $g_\ell$ in turn agrees with the points function $f$ at
$(x,y,z)$, and thus, $h(x,y,z) = f(x,y,z)$. Thus, we have shown that
the global function $h$ is a codeword in $\code^{\otimes 3}$ and agrees with the points
function $f$ exactly.

\paragraph{Dealing with errors}
To extend the sketch above to the general case, with nonzero error,
requires some modifications. At the most basic level, we may consider
what happens when we allow for deterministic classical strategies that
succeed with probability less than $1$ in the test. Such strategies
may have ``mislabeling'' error: the points function
$f$ may be imagined to be a global codeword of $\code^{\otimes 3}$ that has been
corrupted at a small fraction of the points. This type of error is
handled by the analysis in~\cite{babai1991non}. The main modification to the zero-error
sketch above is a careful analysis of the probability that the pasting
step produces a ``good'' interpolated codeword for a randomly chosen
pair of planes $S_{z_1}, S_{z_2}$. This analysis makes use of the
distance property of the code together with combinatorial properties of the
point-line test itself (namely, the expansion of the hypercube graph associated with $[n]^m$).

At the next level of generality, we could consider classical
\emph{randomized} strategies. Suppose we are given a randomized
strategy that succeeds in the test with probability $1 - \eps$. Any
randomized strategy can be modeled by first sampling a random seed,
and then playing a deterministic strategy conditioned on the value of
the seed. A success probability of $1 - \eps$ could have two
qualitatively different underlying causes: (1) on $O(\eps)$ fraction
of the seeds, the strategy uses a function which is totally corrupted,
and (2) on a large fraction of the seeds, the strategy uses functions
which are only $\eps$-corrupted. An analysis of randomized
strategies could naturally proceed in a ``seed-by-seed'' fashion, applying the
deterministic analysis of~\cite{babai1991non} to the large fraction of ``good'' seeds (which
are each only $\eps$-corrupted), while giving up entirely on the
``bad'' seeds. 

Here we consider quantum strategies, which
have much richer possibilities for error. Nevertheless, we are able to
preserve some intuition from the randomized case by working with
\emph{submeasurements}, which are quantum measurements that do not always yield
an outcome. Working with a submeasurement allows us to distinguish two
kinds of error: \emph{consistency error} (the probability that a
submeasurement returns a wrong outcome) and \emph{completeness error}
(the probability that the submeasurement fails to return an outcome at
all). Roughly speaking, the completeness error corresponds to the
probability of obtaining a ``bad'' seed in the randomized case, while
the consistency error corresponds to how well the strategies do on
``good'' seeds. 

The technique of using submeasurements and managing the two types of
error separately goes back to~\cite{ito2012multi}. That work developed a
crucial tool to convert between these two types of error called the
\emph{self-improvement lemma}, and in our analysis we make extensive
use of a refined version of this lemma (\Cref{lem:self-improvement}), that in particular applies to the two-prover setting as opposed to three provers in~\cite{ito2012multi}. 

Essentially, the self-improvement lemma says the
following: suppose that (a) the provers pass the test with probability
$1 - \eps$, and (b) there is a complete measurement $G_{g}$ whose
outcomes are codewords $g \in \code^{\otimes m}$ that has
consistency error $\nu$: that is, $G$ always returns an outcome, but
has probability $\nu$ of producing an outcome $g$ that disagrees with the points measurement $A^u_a$ at a
random point $u$. Then there exists an
``improved'' submeasurement  $H_{h}$ with consistency error $\zeta$
depending \emph{only} on $\eps$, and with completeness error
(i.e. probability of not producing an outcome at all) of $\nu
+\zeta$. Thus the lemma says that we can always ``reset'' the
consistency error of any measurement we construct at intermediate
points in the analysis to a universal
function $\zeta$ depending only on the provers' success in the test,
at the cost of introducing some amount completeness error. Intuitively, one may
think of the action of the lemma as correcting $G$ on the portions of
Hilbert space where it is only mildly corrupted, while ``cutting out'' the
portions of Hilbert space where $G$ is too corrupted to be
correctable. In some sense, this lemma is the quantum analogue of the
idea of identifying ``good'' and ``bad'' random seeds in the classical
randomized case. The proof of the lemma uses a version of semidefinite programming duality
together with the combinatorial facts used in~\cite{babai1991non}.

Armed with the self-improvement lemma, we set up the following
variant of the induction loop sued in~\cite{babai1991non}. We say that an affine subspace $S$ is ``$k$-aligned'' if $S$ has dimension $k$ and for all $u,v \in S$, the projection of $u,v$ to the first $m-k$ coordinates are the same. For $k$ running from $1$ to $m$, 
for all $k$-aligned subspaces $S$, we construct a ``subspace measurement'' $\{G^{S}_{g}\}$ that returns a codeword $g \in \code^{\otimes k}$, as follows. 
\begin{enumerate}
  \item By the induction hypothesis, we know that there exists a
    measurement  $G^{S}_{g}$ for every $k$-aligned subspace $S$, which
    has (on average over the the subspaces $S$) consistency error $\delta(k)$ with the points
    measurement. (For the base case $k = 1$, this is the lines
    measurement from the provers' strategy). 
  \item We apply the self-improvement lemma to these measurements, yielding
    submeasurements $\hat{G}^{S}_{g}$ that have (on average) consistency error
    $\zeta$ independent of $\delta(k)$, and completeness error $\kappa(k) = \delta(k) + \zeta$.
  \item For each $(k+1)$-aligned subspace $S'$, we construct a
    pasted submeasurement $G^{S'}_{g'}$, by performing a quantum version of the
    classical interpolation argument: we define the pasted
    submeasurement by sequentially measuring several parallel $k$-aligned subspaces $S \subseteq S'$   
    and interpolate the resulting outcomes. This pasted submeasurement has (on average) consistency error slightly worse than $\zeta$, and completeness error which
    is slightly worse than $\kappa(k)$. It is at this step that it is crucial to treat the two types
    of error separately: in particular, we need the consistency error
    to be low to ensure that the interpolation produces a good result.
  \item We convert the resulting submeasurement into a full
    measurement, by assigning a random outcome whenever the
    submeasurement fails to yield an outcome. This measurement will
    have (on average) consistency error $\delta(k+1)$ which is larger than $\delta(k)$
    by some additive factor.
  \end{enumerate}
  At the end of the loop, when $k = m$, we obtain a single measurement that returns
  a global codeword as desired.
  
\paragraph{The subcube commutation test}
The ``augmented'' qualifier of the tensor code test in \Cref{thm:main-informal} refers to an additional component called the ``subcube commutation test'', which is described in detail in \Cref{sec:test}. This additional
test is not needed in the classical setting, but appears necessary in the noncommutative
setting. The purpose
of this test is to certify that the points measurements used by the provers
approximately commute on average over all pairs of points $x,y \in [n]^m$---something which is automatically true in the classical case. This is done by asking one prover to report values for a pair randomly chosen points $(u,v)$ from a randomly chosen \emph{subcube}, asking the other prover to report values for either $u$ or $v$ (chosen randomly), and checking consistency between their reported values. A subcube is a subset of points in $[n]^m$ that all share the same values in the last $k$ coordinates for some $1 \leq k \leq n-1$. The reason for choosing a pair of points uniformly at random in a subcube of this form, as opposed to a pair of points uniformly at random in the entire cube, biases their distribution in a way that is appropriate for the induction.


The (approximate) commutation guarantee plays an important role in our analysis of
the tensor code test, and it is an interesting question whether it is truly
necessary to test it directly with the subcube commutation test; might it
not automatically follow from success in the axis-parallel lines test? An interesting contrast can be drawn to
the Magic Square
game~\cite{mermin1990simple,peres1990incompatible,aravind2002simple},
in which questions are either cells or axis-parallel lines in a $3 \times 3$ square grid.
This has the same question distribution as
the axis-parallel line-point test over $[3]^2$. 
For the Magic Square game,
``points'' measurements along the same axis-parallel ``line'' commute,
but points that are not axis-aligned do \emph{not} commute: indeed,
for the perfect strategy, they anticommute.



\subsection{Future directions}
\label{sec:open}

This paper has several motivations. The first is to provide a complete proof of quantum soundness of the tensor code test (and thus the low individual degree test), which as already mentioned plays a central role in the result $\MIP^* = \RE$. 

Second, although $\MIP^* = \RE$ is ostensibly a complexity-theoretic result about the power of quantum multiprover interactive proof systems, it also provides negative answers to seemingly-unrelated questions in mathematical physics and operator algebras, namely Tsirelson's Problem and Connes' Embedding Problem. 
Unfortunately there is still a large gap between the methods used to prove $\MIP^* = \RE$ (which combine techniques from theoretical computer science, complexity theory, and quantum information theory) and the traditional methods to study the problems of Tsirelson and Connes in functional analysis. We believe that the simpler analysis of the tensor code test and its presentation in the synchronous strategies framework will help researchers in pure mathematics as well as theoretical computer science  gain a better understanding of $\MIP^* = \RE$. A simpler, more direct proof may lead to the resolution of other long-standing open problems such as the existence of a non-hyperlinear group~\cite{capraro2015introduction}. 

Third, we believe that the study of locally testable codes in the quantum setting is a subject that is interesting in its own right. At its core, the quantum soundness of tensor codes is about robustly deducing a global ``consistency'' property of a set of operators on a Hilbert space that do not necessarily commute. This hints at a more general study of \emph{noncommutative property testing}, which could combine questions from property testing in theoretical computer science and questions from functional analysis and operator algebras. The field of property testing, which studies how global properties of combinatorial objects such as Boolean functions and graphs can be robustly detected by only examining local views of the object, was initially motivated by the $\MIP = \mathsf{NEXP}$ and PCP theorems. Similarly, $\MIP^* = \RE$ and its mathematical consequences suggest the occurrence of ``local-to-global'' phenomena  in the noncommutative setting that seem worth studying independently of interactive proofs and complexity theory. 

We conclude with an open question. 
The complexity of quantum multiprover interactive proofs with \emph{infinite-dimensional commuting operator} strategies, i.e.\ the class $\MIP^{co}$ (where $co$ stands for ``commuting operator''),  remains unknown. It is natural to conjecture that $\MIP^{co} = \mathsf{coRE}$, as this would nicely complement the $\MIP^* = \RE$ result.\footnote{Note that the $co$ on either side of $\MIP^{co} = \mathsf{coRE}$ refer to different things!} The analysis of the soundness of the tensor code test for infinite-dimensional strategies establishes a first step towards proving this conjecture.

\subsubsection*{Acknowledgments} We thank the anonymous referees for their feedback. We thank William Slofstra and Vern Paulsen for their explanations regarding synchronous strategies.

\section{Preliminaries}
\label{sec:prelims}

We write $[n]$ to denote the set $\{1,2,\ldots,n\}$. For positive integers $a,b,c,\ldots$, that we treat as growing to infinity, we write $\poly(a,b,c,\ldots)$ to denote a function $C (a + b + c \cdots)^C$ for some universal constant $C \geq 1$. For nonnegative real numbers $0 \leq \alpha,\beta,\gamma,\ldots < 1$ that we treat as going to $0$, we write $\poly(\alpha,\beta,\gamma,\ldots)$ to denote a function $C(\alpha^{1/C} + \beta^{1/C} + \gamma^{1/C})$ for some universal constant $C \geq 1$. In either case, the universal constant $C$ can vary each time the $\poly(\cdot)$ notation is used.

For a finite set $S$, we write $x \sim S$ to indicate that $x$ is sampled from $S$ according to some distribution specified by context (usually the uniform distribution). For an event $E$ in a probability space, we write $\indicator[E]$ to denote the indicator  for the event $E$. 

\subsection{Algebras, tracial states, and norms}
Let $\hilb$ be a (separable) Hilbert space and let $\Bounded(\hilb)$ denote the set of bounded linear operators on $\hilb$. We write $\id_\hilb$ to denote the identity operator on $\hilb$ (and simply write $\id$ when the Hilbert space is clear from context). 

A von Neumann algebra on a Hilbert space $\hilb$ is a unital $*$-subalgebra of bounded operators $\Bounded(\hilb)$ that is closed in the \emph{weak operator topology}. Let $\algebra \subseteq \Bounded(\hilb)$ denote a von Neumann algebra on $\hilb$. We say that a positive linear functional $\tau: \algebra \to \C$ is 
\begin{itemize}
	\item \emph{Unital} if $\tau(\id) = 1$~;
	\item \emph{Normal} if there exists a positive trace class\footnote{An operator $A \in \Bounded(\hilb)$ is \emph{trace class} if the trace of $A$ is well-defined; i.e. for all orthonormal bases $\{\ket{e_k}\}$ for $\hilb$, the quantity $\sum_k \bra{e_k} A \ket{e_k}$ is well-defined and independent of the choice of basis.} operator $A$ such that for all $X \in \algebra$, $\tau(X) = \tr(XA)$~; 
	\item \emph{Tracial} if for all $A,B \in \algebra$, we have $\tau(AB) = \tau(BA)$~;
\end{itemize}
In this paper, $\tau$ will always represent a positive linear functional that is tracial, normal, and unital. We call such functionals a \emph{normal tracial state}. For brevity we often drop the ``normal'' qualifier. 

 If $\hilb$ is finite dimensional (i.e. isomorphic to $\C^r$), then we will without loss of generality take the algebra to be $\Bounded(\hilb)$. There is a unique tracial state on $\Bounded(\hilb)$ which is the \emph{normalized trace}: 
\[
    \tau(A) = \frac{1}{r} \tr(A).
\]
For a more comprehensive reference on von Neumann algebras, we refer the reader to~\cite{blackadar2006operator}.

We record some basic properties of tracial states. First, the map $(A,B)\mapsto \tau(A^* B)$ is a semi-inner product and in particular satisfies the Cauchy-Schwarz inequality, i.e.
\[
    |\tau(A^* B)|^2 \leq \tau(A^* A) \, \tau(B^* B).
\]
Second, tracial states give rise to a seminorm on $\algebra$: we define the \emph{$\tau$-norm} of an operator $A \in \algebra$ to be
\[
    \| A \|_\tau = \sqrt{\tau(A^* A)} = \sqrt{\tau(AA^*)}.
\]
The $\| \cdot \|_\tau$ norm satisfies the triangle inequality: i.e., $\| A + B \|_\tau \leq \|A\|_\tau + \|B \|_\tau$. If $\cal{H}$ is finite-dimensional, then $\tau$-norm is the normalized Frobenius norm. 

More generally, tracial states give rise to $p$-seminorms for $1 \leq p < \infty$ via $\| A \|_p = \tau(|A|^p)^{1/p}$, where for all $A \in \algebra$, the operator absolute value $|A|$ denotes $(A^* A)^{1/2}$. These $p$-seminorms satisfy triangle and H\"{o}lder inequalities (see~\cite{pisier2003non} for an in-depth discussion of these norms and noncommutative $L_p$ spaces); we record some special cases here.

\begin{proposition}[Triangle and H\"{o}lder inequalities]\label{prop:holder}
Let $\tau$ be a tracial state on a von Neumann algebra $\algebra$, and let $A,B \in \algebra$. Then 
\begin{enumerate}
	\item \label{item:holder-0} $|\tau(A)| \leq \tau(|A|)$~.
	\item \label{item:holder-1} (H\"{o}lder inequality 1) $\trace{\abs{AB}} \le \norm{A}_{\tau} \cdot \norm{B}_{\tau}$~.
  \item \label{item:holder-2} (H\"{o}lder inequality 2) $\trace{\abs{AB}} \le \norm{A} \cdot \trace{\abs{B}}$~.
	\item \label{item:triangle} (Triangle inequality for $1$-norm) $\tau(|A + B|) \leq \tau(|A|) + \tau(|B|)$.  
\end{enumerate}
\end{proposition}
\begin{proof}
The proof of the first item is as follows. By the polar decomposition, there exists a partial isometry $U \in \algebra$ such that $A = U|A|$. Then by Cauchy-Schwarz, 
\begin{gather*}
	|\tau(A)| = |\tau(U|A|)| = |\tau(|A|^{1/2} U |A|^{1/2})| \leq \sqrt{ \tau(|A|^{1/2} U^* U |A|^{1/2}) \tau(|A|)} = \tau(|A|).
\end{gather*}
The proof of the first H\"{o}lder inequality is as follows. By the polar decomposition, there exists a partial isometry $U \in \algebra$ such that $|AB| = UAB$. Then
\begin{align*}
	\trace{\abs{AB}} = |\tau(UAB)| \leq \sqrt{ \trace{A^* U^* UA} \cdot \trace{B^* B}} = \norm{A}_\tau \cdot \norm{B}_\tau
\end{align*}
where we used Cauchy-Schwarz and that $U$ is a partial isometry.

The proof of the second H\"{o}lder inequality is as follows. Let $|AB| = UAB$ and $B = V|B|$ for partial isometries $U,V \in \algebra$. Then
\begin{gather*}
	\trace{\abs{AB}} = |\tau(UAV|B|)| = |\tau(\abs{B}^{1/2} UAV \abs{B}^{1/2})| \leq \sqrt{ \trace{\abs{B}^{1/2} UAA^* U^* \abs{B}^{1/2}} \cdot \trace{|B|}}
\end{gather*}
where we used the cyclicity of the trace and Cauchy-Schwarz. Since $AA^* \leq \norm{AA^*} \id$, we have that this is at most
\begin{gather*}
	\sqrt{ \trace{\abs{B}^{1/2} (\norm{AA^*} \cdot \id) \abs{B}^{1/2}} \cdot \trace{|B|}} = \norm{A} \cdot \trace{|B|}
\end{gather*}
where we use that $\sqrt{\norm{AA^*}} = \norm{A}$. 

The proof of the triangle inequality is as follows. By the polar decomposition, there exists a partial isometry $U \in \algebra$ such that $|A + B| = U(A+B)$. Therefore
\begin{align*}
	\tau(|A + B|) = |\tau(U(A+B))| = |\tau(UA) + \tau(UB)| \leq |\tau(UA)| + |\tau(UB)| \leq \tau(|A|) + \tau(|B|)
\end{align*}
where the first equality follows from the positivity of $|A+B|$ and the last inequality follows from the second H\"{o}lder inequality.
\end{proof}



\subsection{Measurements and distance measures on them} 
\label{sec:measurements}

Let $M = \{M_a\}_{a \in \cal{A}}$ and $N = \{N_a\}_{a \in \cal{A}}$ denote sets of operators, indexed by a finite set $\cal{A}$, in a von Neumann algebra $\algebra$ with trace $\tau$. We measure the distance between $M$ and $N$, denoted by $\| M - N \|_\tau$, as
\[
    \| M - N\|_\tau = \sqrt{\sum_{a \in \cal{A}} \left \| M_a - N_a \right \|^2_\tau }\;.
\]
We say that $M$ is \emph{$\delta$-close} to $N$, denoted by $M_a \approx_\delta N_a$, if $\|M - N \|_\tau \leq \delta$. 


A \emph{submeasurement on $\hilb$ with answer/outcome set $\cal{A}$} is a set of positive operators $\{M_a\}_{a \in \cal{A}}$ such that $\sum_{a \in \cal{A}} M_a \leq \id$. We say that $\{M_a\}_{a \in \cal{A}}$ is a \emph{measurement} if $\sum_a M_a = \id$. A \emph{projective (sub)measurement} is a (sub)measurement such that each element $M_a$ is a projection. To denote ``data processed'' measurements, i.e., apply a function $f: \cal{A} \to \cal{B}$ to the outcome of a measurement, we use the following notation: $M_{[f]}$ denotes the (sub)measurement with elements
\[
    M_{[f \mid b]} = \sum_{a : f(a) = b} M_a
\]
for all $b \in \cal{B}$. As an example, suppose $\cal{A} = \{0,1\}^n$ and $\cal{B} = \{0,1\}$. Then we write $M_{[a \mapsto a_i]}$ to denote the processed measurement that measures a string $a$, and then returns the $i$-th bit of $a$. To refer to the element of $M_{[a \mapsto a_i]}$ corresponding to outcome $b \in \{0,1\}$, we write $M_{[a \mapsto a_i \mid b]}$. For a predicate $P: \cal{A} \to \{0,1\}$, we also use the notation
\[
	M_{[a : P(a)]} = \sum_{a : P(a) = 1} M_a ~.
\]
For example, the operator $M_{[a : f(a) \neq b]}$ denotes the sum over all $M_a$ such that $f(a) \neq b$. 

Throughout this paper (sub)measurements are often indexed by a finite set $\cal{X}$, and the elements in the set $\cal{X}$ are generally drawn from a distribution $\mu$. The set $\cal{X}$ is typically called a \emph{question set} and $\mu$ is typically called a \emph{question distribution}, because these will correspond to questions that are sampled in a nonlocal game such as the tensor codes test. Questions appear as superscripts and answers appear as subscripts of a measurement operator; for example, $M^x_a$ denotes the measurement operator corresponding to question $x$ and answer $a$.

We introduce two important distance measures between (sub)measurements that will be used throughout this paper. 
We fix a question set $\cal{X}$, a question distribution $\mu$, and an answer set $\cal{A}$. Let $M = \{M^x\}_{x}$ and $N = \{N^x\}_x$ denote two sets of submeasurements, where each $M^x = \{M^x_a\}_a$ and $N^x = \{N^x_a\}_a$. 

The first distance measure we define is called \emph{consistency}. 
We say that $M$ and $N$ are \emph{$\delta$-consistent} if 
\[
 	\E_{x \sim \mu} \sum_{\substack{a,b \in \cal{A}: \\ a \neq b}} \tau(M_a^x \, N_b^x) \leq \delta\;.
\]
We also refer to the quantity on the left-hand side as the \emph{inconsistency} between $M$ and $N$. So, saying that $M$ and $N$ are $\delta$-consistent is equivalent to saying that the inconsistency of $M$ and $N$ is at most $\delta$. 
When the question set $\cal{X}$, question distribution $\mu$, and answer set $\cal{A}$ are clear from context, we write $M_a \simeq_\delta N_a$ to denote that $M$ and $N$ are $\delta$-consistent. Observe that in case $M$ and $N$ are measurements then $M$ and $N$ are $\delta$-consistent if and only if
\[
 	\E_{x \sim \mu} \sum_{\substack{a\in \cal{A}}} \tau(M_a^x \, N_a^x) \geq 1-\delta\;.
\]

The second distance measure we introduce is called \emph{closeness}. We say that sets $M,N$ of (sub)measurements are \emph{$\delta$-close} if 
\[
  \sqrt{ \E_{x \sim \mu} \sum_{a} \| M^x_{a} - N^x_{a} \|_\tau^2} \leq \delta.
\]
Similarly, when $\cal{X},\cal{A},\mu$ are clear from context, we write $M_a^x \approx_\delta N_a^x$ to denote that $M^x$ and $N^x$ are $\delta$-close on average over $x \sim \mu$. Observe that this notion of closeness is also well-defined when the operators $M_a^x$, $N_a^x$ are not necessarily positive. Thus we will also write $M_a^x \approx_\delta N_a^x$ to denote closeness of arbitrary operator sets that are indexed by question and answer sets $\cal{X},\cal{A}$.

\subsubsection{Utility lemmas about consistency and closeness of measurements}

We now establish several utility lemmas concerning consistency, closeness, and measurements. In what follows, we let $\cal{X}$ denote a finite question set, $\mu$ a distribution over $\cal{X}$, and $\cal{A}$ a finite answer set. All expectations are over $x$ sampled from $\mu$.

\begin{proposition}[Cauchy-Schwarz for operator sets]
	Let $M = \{M_a\}_{a \in \cal{A}}$ and $N = \{N_a\}_{a \in \cal{A}}$ denote sets of operators (not necessarily submeasurements). Then 
	\[
		\Big | \sum_{a \in \cal{A}} \tau( M_a \cdot N_a) \Big |^2 \leq \Big( \sum_{a \in \cal{A}} \| M_a \|_\tau^2 \Big) \cdot  \Big( \sum_{a \in \cal{A}} \| N_a \|_\tau^2 \Big)  \;.
	\]
\end{proposition}
\begin{proof}
Applying Cauchy-Schwarz twice,
\begin{align*}
    \Big | \sum_{a \in \cal{A}} \tau( M_a \cdot N_a) \Big |
    \leq  \sum_{a \in \cal{A}} \Big |\tau( M_a \cdot N_a) \Big |
    &\leq \sum_{a \in \cal{A}} \sqrt{\tau(M_a M_a^*) \tau(N_a^* N_a)}\\
    &= \sum_{a \in \cal{A}} \| M_a \|_\tau \| N_a \|_\tau
    \leq \sqrt{\sum_{a \in \cal{A}} \| M_a \|_\tau^2}\cdot \sqrt{\sum_{a \in \cal{A}} \| N_a \|_\tau^2}.
\end{align*}
The proposition follows by squaring both sides.
\end{proof}

\begin{corollary}
\label{cor:cauchy-schwarz}
Let $\{M^x_a\}$ and $\{N^x_a\}$ be sets of operators indexed by both $\cal{X}$ and $\cal{A}$. Then
	\[
		\Big | \E_{x} \sum_{a \in \cal{A}} \tau( M_a^x \cdot N_a^x) \Big |^2 \leq \Big( \E_{x} \sum_{a \in \cal{A}} \| M_a^x \|_\tau^2 \Big) \cdot  \Big( \E_{x} \sum_{a \in \cal{A}} \| N_a^x \|_\tau^2 \Big)  
	\].
\end{corollary}
\begin{proof}
	We apply Cauchy-Schwarz twice:
	\begin{align*}
	\Big | \E_{x} \sum_{a \in \cal{A}} \tau( M_a^x \cdot N_a^x) \Big |^2 &\leq \Big ( \E_{x} \Big | \sum_{a \in \cal{A}} \tau( M_a^x \cdot N_a^x) \Big | \Big)^2 \\
	&\leq \Big ( \E_{x} \sqrt{ \sum_{a \in \cal{A}} \|M^x_a \|_\tau^2 } \cdot \sqrt{ \sum_{a \in \cal{A}} \|N^x_a \|_\tau^2 } \Big)^2 \\
	&\leq \Big( \E_{x} \sum_{a \in \cal{A}} \| M_a^x \|_\tau^2 \Big) \cdot  \Big( \E_{x} \sum_{a \in \cal{A}} \| N_a^x \|_\tau^2 \Big)~.
	\end{align*}
\end{proof}
%

\begin{proposition}[Data processing inequality for consistency]
\label{lem:data-processing}
Let $M^x = \{M_a^x\}$ and $N^x = \{N_a^x\}$ be submeasurements with outcomes in $\cal{A}$ such that $M_a^x \simeq_\delta N_a^x$ on average over $x$. Let $f: \cal{A} \to \cal{B}$. Then 
\[
	M^x_{[f \mid b]} \simeq_\delta N^x_{[f \mid b]}
\]
on average over $x$ where the answer summation is over $b \in \cal{B}$.
\end{proposition}
\begin{proof}
	\begin{equation*}
		\E_x \sum_{b \neq b' \in \cal{B}} \tau (M^x_{[f \mid b]} N^x_{[f \mid b']}) = \E_x \sum_{\substack{b \neq b' \in \cal{B} \\ a,a' \in \cal{A} \\ f(a) = b \\ f(a') = b'}} \tau(M_a^x N^x_{a'}) \leq \E_x \sum_{a \neq a' \in \cal{A}} \tau(M_a^x N^x_{a'}) \leq \delta.\qedhere
	\end{equation*}
\end{proof}

\begin{proposition}[Consistency to closeness]
\label{lem:consistency-consequences}
Let $M^x = \{M_a^x\}$ and $N^x = \{N_a^x\}$ be measurements with outcomes in $\cal{A}$ such that $M_a^x \simeq_\delta N_a^x$ on average over $x$. Then $M_a^x \approx_{\sqrt{2\delta}} N_a^x$ on average over $x$.
\end{proposition}
\begin{proof}
	\begin{align*}
	    \sqrt{ \E_x \sum_a \|M_a^x - N_a^x\|_\tau^2} &=
		\sqrt{\E_x  \sum_a \tau ((M_a^x - N_a^x)^2)} \\
		&\leq \sqrt{ \E_x\sum_a \tau (M_a^x + N_a^x - 2 M_a^x N_a^x)} \\
		&= \sqrt{ \E_x  2 - 2 \sum_a \tau(M_a^x N_a^x)} \\
		&= \sqrt{2 \E_x \sum_a \tau(M_a^x (\Id - N_a^x))} \\
		&\leq \sqrt{2 \delta}.
	\end{align*}
	The second line follows because $M_a^x - (M_a^x)^2 \geq 0$ as $\{M_a^x\}$ are measurements.
\end{proof}

\begin{proposition}[Closeness to consistency]
\label{lem:closeness-to-consistency}
Let $M^x = \{M_a^x\}$ be projective submeasurements and let $N^x = \{N_a^x\}_{a \in \cal{A}}$ be submeasurements with outcomes in $\cal{A}$. Suppose that $M_a^x \approx_\delta N_a^x$ on average over $x$. Then $M_a^x \simeq_{\delta} N_a^x$ on average over $x$. 
\end{proposition}
\begin{proof}
	\begin{align*}
		\E_x \sum_{a \neq b} \tau(M^x_a N^x_b) &\leq \E_x \sum_a \tau( M_a^x (\Id - N_a^x)) \\
		&= \E_x \sum_a \tau( M_a^x (M_a^x - N_a^x))\\
		&\leq \sqrt{\E_x  \sum_a \tau((M_a^x)^2)} \cdot \sqrt{\E_x \sum_a \tau( (M_a^x - N_a^x)(M_a^x - N_a^x)^*) } & \text{(by Cauchy-Schwarz)} \\ 
		&\leq \delta 
	\end{align*}
	where we used that $\sum_a \tau((M_a^x)^2) \leq 1$. 
\end{proof}


\begin{proposition}\label{lem:closeness-to-close-ips}
  Let $A^x = \{A_a^x\}$ and $B^x = \{B_a^x\}$ be two sets of operators indexed
  by subscripts in $\cal{A}$ such that $A^x_a \approx_\delta B^x_a$.
  Let $M^{x}_{a}$ be operators satisfying
  $\sum_{a} M^{x}_{a} {(M^{x}_{a})}^{*} \le \Id$ for all $x$.
  Then
  \[
    \E_{x} \sum_{a} \trace{ M^{x}_{a} A^{x}_{a} } \approx_{\delta}
    \E_{x} \sum_{a} \trace{ M^{x}_{a} B^{x}_{a}}.
  \]
  In particular, if $C^{x} = \{C^{x}_{a}\}$ are submeasurements, then
  \[
	\E_x \sum_a \tau(C^{x}_{a} A^{x}_{a}) \approx_\delta
    \E_x \sum_a \tau(C^{x}_{a} B^{x}_{a}).
  \]
\end{proposition}
\begin{proof}
Via Cauchy-Schwarz we have:
\[
  \Big | \E_x \sum_a \tau(M^{x}_{a}(A^x_a - B^x_a)) \Big |
  \leq \sqrt{ \E_x \sum_a \tau(M^{x}_{a}{(M^{x}_{a})}^{*})} \cdot
  \sqrt{ \E_x \sum_a \| A^x_a - B^x_a \|_\tau^2 }
  \leq \delta.\qedhere
\]
\end{proof}

\begin{proposition}
\label{lem:add-a-proj}
Let $M = \{M_a\}_{a \in \cal{A}},N = \{N_a\}_{a \in \cal{A}}$ be sets of operators (not necessarily measurements), and let $R = \{R_b\}_{b \in \cal{B}}$ be a set of operators such that $\sum_b R_b^* R_b \leq \Id$. Suppose that $M_a \approx_\delta N_a$. Then $R_b M_a \approx_\delta R_b N_a$  where the answer summation is over $(a,b) \in \cal{A} \times \cal{B}$. Similarly, if $\sum_b R_b R_b^* \leq \Id$, we have $M_a R_b  \approx_\delta N_a R_b$. 
\end{proposition}
\begin{proof}
We prove the approximation $R_b M_a \approx_\delta R_b N_a$:
	\begin{align*}
		\sum_{a \in \cal{A}, b \in \cal{B}} \| R_b (M_a - N_a) \|_\tau^2 &= \sum_{a \in \cal{A}, b \in \cal{B}} \tau \Big( (M_a - N_a)^* R_b^* R_b (M_a - N_a) \Big) \\
		&= \sum_{a} \tau \Big( (M_a - N_a)^* \Big( \sum_b R_b^* R_b \Big) (M_a - N_a) \Big) \\
		&\leq \sum_{a} \tau \Big( (M_a - N_a)^* (M_a - N_a) \Big) \\
		&= \sum_a \| M_a - N_a \|_\tau^2
		&\leq \delta^2.
	\end{align*}
	where in the first inequality we used the assumption that $\sum_b R_b^* R_b \leq \Id$ and positivity of the trace. The proof for the approximation $M_a R_b  \approx_\delta N_a R_b$ follows similarly.
\end{proof}

\begin{proposition}[Transfering ``$\simeq$'' using ``$\approx$'']\label{lem:transfer-cons}
  Let $\{A^x_a\}$ and $\{B^x_a\}$ be measurements, and let $\{C^x_a\}$ be a
  submeasurement.
  Suppose that $A^x_a \simeq_{\delta} C^x_a$ and
  $A^x_a \approx_{\eps} B^x_a$.
  Then $B^x_a \simeq_{\delta + \eps} C^x_a$.
\end{proposition}

\begin{proof}
  Let $C^x = \sum_a C^x_a$ and $C = \E_x C^x$. First, we can rewrite the inconsistency between $\{A^x_a\}$ and $\{C^x_a\}$ as
  \begin{alignat*}{3}
    & \E_{x} \sum_{a \neq b} \trace {A^{x}_a C^{x}_b}
    && = \E_{x} \sum_{a} \trace{ A^{x}_a (C^{x} - C^{x}_a) } \\
    & && =  \E_{x} \trace{ C^{x} } - \E_{x} \sum_{a} \trace{ A^{x}_a  C^{x}_a }
      & \qquad\qquad \text{(because $\{A^x_a\}$ is a measurement)}\\
    & && = \trace{ C } - \E_{x} \sum_a \trace{ A^{x}_a  C^{x}_a}.
  \end{alignat*}
  Likewise, we can rewrite the
  inconsistency between $\{B^x_a\}$ and $\{C^x_a\}$ as
  \begin{align*}
    \E_{x} \sum_{a \neq b} \trace{B^{x}_a  C^{x}_b}
    &= \trace{ C } - \E_{x} \sum_a
                 \trace{ B^{x}_a  C^{x}_a }.
  \end{align*}
  We want to show that the inconsistency between~$B$ and~$C$ is close to the inconsistency between $A$
  and $B$. In particular, we claim that
  \begin{equation*}
  \E_{x} \sum_{a \neq b} \trace{ A^{x}_a C^{x}_b}
    \approx_{\eps} \E_{x} \sum_{a \neq b} \trace{ B^{x}_a C^{x}_b }.
  \end{equation*}
  To show this, we bound the magnitude of the difference using Cauchy-Schwarz.
  \begin{alignat*}{3}
    & && \Big| \E_{x} \sum_{a } \trace{ (A^{x}_a - B^{x}_a) C^{x}_a}
    \Big| \nonumber \\
    & \leq \;\; && \sqrt{ \E_{x} \sum_{a} \trace{ {(A^{x}_a - B^{x}_a)}^2 } }
    \cdot \sqrt{\E_{x} \sum_{a} \trace { {(C^{x}_a)}^2 }} \\
    & \leq && \eps \cdot 1. & \qquad \text{(because~$\{C^x_a\}$ is a submeasurement)}
  \end{alignat*}
  This completes the proof.
\end{proof}

The following fact is useful for translating between statements about
consistency and closeness between submeasurements.
\begin{proposition}\label{lem:cons-sub-meas}
  Let $\{A^x_a\}$ be a submeasurement and let $\{B^x_a\}$ be a measurement such
  that on average over $x$,
  \[
    A^x_a \simeq_\gamma B^x_a.
  \]
  Then
  \begin{gather}
    \label{eq:cons-sub-meas}
    A^x_a \approx_{\sqrt{\gamma}} A^x_a B^x_a \approx_{\sqrt{\gamma}} A^x B^x_a,
  \end{gather}
  where $A^x = \sum_a A^x_a$.
  As a result, by the triangle inequality,
  \begin{equation*}
    A^x_a \approx_{2\sqrt{\gamma}} A^x B^x_a
  \end{equation*}
\end{proposition}

\begin{proof}
  We establish the first approximation in \cref{eq:cons-sub-meas}:
  \begin{align*}
    \E_{x} \sum_{a} \trace{\bigl(A^{x}_{a}(\Id-B^{x}_{a})\bigr)^{*}
      (A^{x}_{a}(\Id-B^{x}_{a}))}
    = &\E_{x} \sum_a \trace{{(A^{x}_a)}^{2} {(\Id - B^{x}_a)}^2}  \\
    \leq & \E_{x} \sum_a \trace{A^{x}_a (\Id - B^{x}_a)} \\
    = & \E_{x} \sum_{\substack{a,b : \\ b \neq a}} \trace{A^{x}_a B^{x}_b}\\
    \leq & \gamma.
  \end{align*}
  The first equality uses that positive operators are self-adjoint and follows from cyclicity of the trace. The first inequality folows from the fact that
  ${(A^{x}_{a})}^{2} \le A^{x}_{a}$,
  ${(\Id-B^{x}_{a})}^{2} \le \Id - B^{x}_{a}$, and positicity of the trace.
  The last line follows from the assumption of consistency between the $A$ and
  $B$ (sub)measurements.

  To establish the second approximation in \cref{eq:cons-sub-meas}, we compute
  the difference:
  \begin{align*}
    \E_{x} \sum_{a} \trace{ {(A^{x} - A^{x}_{a})}^{2} {(B^{x}_{a})}^{2} }
    \leq & \E_{x} \sum_{a} \trace{ (A^{x} - A^{x}_{a}) B^{x}_{a} }\\
    = & \E_{x} \sum_{a,b: a \ne b} \trace{A^{x}_{b} B^{x}_{a}}\\
    \leq & \gamma.
  \end{align*}
  The first inequality follows from the fact that
  ${(A^{x} - A^{x}_{a})}^{2} \le A^{x} - A^{x}_{a}$ and
  ${(B^{x}_{a})}^{2} \le B^{x}_{a}$, and the last inequality follows from the
  consistency between the $A$ and $B$ (sub)measurements.
\end{proof}

The following proposition states that, if a collection of measurements approximately pairwise commute, then they can also be approximately commuted within longer products of operators.
\begin{proposition}\label{lem:switcheroo}
  For all $x \in \cal{X}$ let $A^x = \{A^x_a\}$ denote a submeasurement with
  answer set $\cal{A}$.
  Furthermore suppose that on average over $x,y \sim \cal{X}$, it holds that
 \begin{equation}
    \label{eq:switcheroo-0}
    A^x_a A^y_b \approx_\eps A^y_b A^x_a\;.
  \end{equation}
  Then for all $k \in \N$ it holds that on average over $x \sim \cal{X}$ and uniformly random
  $s \sim \cal{X}^{k}$,
  \begin{equation}
    \label{eq:pasting-induction}
    P^s_{\vec{a}}  \,A^{x}_b \approx_{k \eps} A^{x}_b \, P^s_{\vec{a}}
  \end{equation}
where for all integers $k \geq 1$, vectors $\vec{a} \in \cal{A}^k$, and sequences $s \in \cal{X}^k$ we define
\[
	P^s_{\vec{a}} = A^{s_1}_{{a}_1} \cdot A^{s_2}_{{a}_2} \cdots A^{s_k}_{{a}_k}.
\]
\end{proposition}
\begin{proof}

We prove \Cref{eq:pasting-induction} by induction on $k$. 
	The base case for $k = 1$ follows from the assumption of the approximate commutativity of the $A^{x}$ measurements. 
	Assuming the inductive hypothesis holds for some $k \geq 1$, we now prove it for $k+1$: let $s \in \cal{X}^k, t \in \cal{X}$. We can treat $(s,t)$ as a sequence of length $k+1$. Then for all $x \in \cal{X}$, we have
	\begin{align*}
      & \sqrt{\E_{\substack{s \sim \cal{X}^k \\ t,x \sim \cal{X}}} \, \sum_{\vec{a}, b, c} \left \| P^{s,t}_{\vec{a},b} A^{x}_c - A^{x}_c P^{s,t}_{\vec{a},b} \right \|_\tau^2}  \\
      = & \sqrt{\E_{\substack{s \sim \cal{X}^k \\ t,x \sim \cal{X}}} \, \sum_{\vec{a}, b, c} \left \| P^{s}_{\vec{a}} A^{t}_b A^{x}_c - A^{x}_c P^{s}_{\vec{a}} A^t_b \right \|_\tau^2 }\\
      \leq & \sqrt{ \E_{\substack{s \sim \cal{X}^k \\ t,x \sim \cal{X}}} \, \sum_{\vec{a}, b, c} \Big( \left \| P^{s}_{\vec{a}}  \Big( A^t_b A^{x}_c - A^{x}_c  A^t_b \Big) \right \|_\tau + \left \| \Big ( P^{s}_{\vec{a}} A^{x}_c  - A^{x}_c P^{s}_{\vec{a}}  \Big ) A^t_b \right \|_\tau \Big)^2 }\\
      \leq & \sqrt{\E_{\substack{s \sim \cal{X}^k \\ t,x \sim \cal{X}}} \,  \sum_{\vec{a}, b ,c} \left \| P^{s}_{\vec{a}}  \Big( A^t_b A^{x}_c - A^{x}_c  A^t_b \Big) \right \|_\tau^2 } + \sqrt{ \E_{\substack{s \sim \cal{X}^k \\ t,x \sim \cal{X}}} \, \sum_{\vec{a}, b ,c} \left \| \Big ( P^{s}_{\vec{a}} A^{x}_c  - A^{x}_c P^{s}_{\vec{a}}  \Big ) A^t_b \right \|_\tau^2}
	\end{align*}
	where the third line follows from the triangle inequality, and the third line follows from squaring both sides and applying Cauchy-Schwarz. Next, notice that
	\begin{align*}
	\sqrt{\E_{\substack{s \sim \cal{X}^k \\ t,x \sim \cal{X}}} \,  \sum_{\vec{a}, b ,c} \left \| P^{s}_{\vec{a}}  \Big( A^t_b A^{x}_c - A^{x}_c  A^t_b \Big) \right \|_\tau^2 } &= \sqrt{\E_{\substack{s \sim \cal{X}^k \\ t,x \sim \cal{X}}} \,  \sum_{\vec{a}, b ,c} \tau \Big( \Big( A^t_b A^{x}_c - A^{x}_c  A^t_b \Big)^* (P^{s}_{\vec{a}})^* P^{s}_{\vec{a}}  \Big( A^t_b A^{x}_c - A^{x}_c  A^t_b \Big) \Big) }\\
	&\leq \sqrt{\E_{t,x \sim \cal{X}} \,  \sum_{b,c} \left \|A^t_b A^{x}_c - A^{x}_c  A^t_b \right \|_\tau^2 } \leq \eps
	\end{align*}
	where the last inequality follows from the assumption~\eqref{eq:switcheroo-0}. Similarly, we can bound
	\[
	\sqrt{\E_{\substack{s \sim \cal{X}^k \\ t,x \sim \cal{X}}} \,  \sum_{\vec{a}, b ,c}  \left \| \Big ( P^{s}_{\vec{a}} A^{x}_c  - A^{x}_c P^{s}_{\vec{a}}  \Big ) A^t_b \right \|_\tau^2 } \leq \sqrt{  \E_{s \sim \cal{X}^k, x \sim \cal{X}} \, \sum_{\vec{a}, c}  \| P^s_{\vec{a}} A^{x}_c - A^{x}_c P^s_{\vec{a}} \|_\tau^2} \leq k \eps
	\]
	using the fact that $\{A^t_b\}$ is a submeasurement for the first inequality and using the inductive hypothesis for the second inequality.
	Thus we have established that
	\[
	\E_{\substack{s \sim \cal{X}^k \\ t,x \sim \cal{X}}} \, \sqrt{\sum_{\vec{a}, b, c} \left \| P^{s,t}_{\vec{a},b} A^{x}_c - A^{x}_c P^{s,t}_{\vec{a},b} \right \|_\tau^2}  \leq (k+1) \eps
	\]
	which proves the inductive hypothesis for $k+1$. By induction, the hypothesis holds for all $k$. 
\end{proof}

\subsection{Codes}

Recall the definition of a code from the introduction.

\begin{definition}\label{def:code}
  A linear $[n,k,d]_\Sigma$ code over a finite field $\Sigma$ is a
  set $\code$ of functions $c: [n] \to \Sigma$ with size $|\code| =
  |\Sigma|^k$ that is closed under linear combination, such that
  for any two distinct $c \neq c' \in \code$, the number of coordinates $i \in
  [n]$ such that $c(i) \neq c'(i)$ is at least $d$. The parameter $n$
  is called the \emph{blocklength}, $k$ is called the
  \emph{dimension}, and $d$ is called the \emph{distance} of the code.
\end{definition}
Since all codes in this paper are linear, we drop the qualifier ``linear'' for brevity. The distance property of a code implies the following simple fact:
\begin{proposition}\label{prop:distance0}
  Let $\code$ be an $[n,k,d]_\Sigma$ code, and let $t = n -d +1$. Suppose
  that $i_1, \dots, i_t$ are $t$ distinct elements in $[n]$. Then
  given a set of values $a_1, \dots, a_t \in \Sigma$, there exists
  \emph{at most one} codeword $c \in \code$ such that $c(i_j) = a_j$
  for all $1 \leq j \leq t$. 
\end{proposition}
\begin{proof}
  Suppose not. Then there are two distinct codewords $c, c'$ that
  agree on at least $n - d + 1$ locations, and thus disagree on at
  most $d-1$ locations. This contradicts the assumption that the code
  has distance $d$.
\end{proof}


We now formally define the \emph{interpolability} condition on codes, which was discussed in \Cref{sec:proof-overview}. All codes
we consider in the paper are interpolable.

\begin{definition}\label{def:interpolable} 
  Let $\code$ be a linear $[n,k,d]_\Sigma$ code, and let $t = n - d +
  1$. We say $\code$ is \emph{interpolable} if, for all collections of $t$ coordinates $i_1,
  \dots, i_t \in [n]$, there exists a linear \emph{interpolation map}
  $\interpol_{i_1, \dots, i_t}:
  \Sigma^t \to \code$ that maps a collection of values $(a_1, \dots,
  a_t) \in \Sigma^t$ to the unique codeword $c \in \code$ such that $c(i_j) = a_j$ for
  all $j \in [t]$.
\end{definition}
\begin{remark}
The condition of being interpolable is equivalent to saying that any
$t$ rows of the generator matrix $G \in \Sigma^{n \times k}$ of the
code are linearly independent. We observe that the Reed-Solomon code
is interpolable for any degree: for Reed-Solomon with degree $s$, the
distance of the code is $n - s$, and a degree-$s$ polynomial can be
uniquely interpolated from its value at $t = s+1$ points.
\end{remark}

\subsubsection{Tensor codes}


\begin{definition}\label{def:axis-line}
Given $j \in [m]$ and $u \in [n]^{m-1}$, the \emph{axis-parallel line}
$\ell(j, u)$ through axis $j$ and with intercept $u$ is the set
  \[ \ell(j, u) = \{(u_1, u_2, \dots, u_{j-1}, i, u_{j}, \dots,
    u_{m-1}) \in [n]^{m} : i \in [n] \}. \]
\end{definition}

\begin{definition}\label{def:tensor-code}
  Let $\code$ be a $[n,k,d]_\Sigma$ code and let $m \geq 1$ be an integer. The \emph{tensor code} $\code^{\ot  m}$ is the set of all
  functions $c: [n]^m \to \Sigma$ such that the restriction
  $c|_{\ell}$ of $c$ to any axis-parallel line $\ell$ is a codeword of $\code$.
\end{definition}

\begin{proposition}\label{prop:distance}
Let $\gamma_m = 1 - \frac{d^{m}}{n^{m}}$. For any $c\neq c' \in \code^{\otimes m}$ it holds that 
\[ \Pr_{u\in [n]^m} \big( c(u)=c'(u)\big) \,\leq\, \gamma_m \,\leq\, \frac{mt}{n}\;.\]
\end{proposition}

\begin{proof}
The first inequality follows from the fact that the distance of $\code^{\otimes m}$ is $d^m$. (This is easily seen by verifying that codewords of $\code^{\otimes m}$ are exactly the tensor products of (possibly different) codewords from $\code$.) For the second inequality we use Bernoulli's inequality,
\begin{equation*}
	\gamma_m = 1 - \frac{d^{m}}{n^{m}} = 1 - \left( 1-\frac{t-1}{n} \right)^{m}
	\le \frac{m(t-1)}{n} \le \frac{mt}{n}\;.
\end{equation*}
\end{proof}

The following two propositions relate codewords of $\code^{\ot (m+1)}$
to tuples of codewords of $\code^{\ot m}$.
\begin{proposition}\label{prop:tuple-to-code-correspondence}
  Let $\code$ be an $[n,k,d]_\Sigma$ code, and let $t \geq n -d+1$. Let $x_1,
  \dots, x_t \in [n]$, and define the subset $S \subseteq
  \codemt$ as
  \[
    S = \{ (g_1,\ldots,g_t) \in \codemt :
    \text{there exists $h \in \code^{\otimes (m+1)}$ such that } h|_{x_j} = g_j \},
  \]
  where $h|_{x_j}$ denotes the function $h(\cdot, \dots, \cdot, x_j)$. Then there is a one-to-one correspondence between $S$ and
  $\code^{\otimes (m+1)}$.
\end{proposition}
\begin{proof}
  The correspondence maps a tuple $(g_1, \dots, g_t) \in S$ to the
  unique corresponding element $h \in \code^{\otimes (m+1)}$ such that
  $h_{|x_j} = g_j$ for all $j$. By the definition of $S$ such an $h$
  exists; to show that
  this is well defined, suppose there are two distinct $h \neq h'$
  such that $h|_{x_j} = g_j = h'|_{x_j}$ for all $j$. Since $h \neq h'$,
  there must exist some axis-parallel line $\ell(m+1, u)$ along the
  $(m+1)$-st axis such
  that $h|_\ell \neq h'|_\ell$ (this is because such lines partition
  the entire space $[n]^{m+1}$. At the same time, we have that
  \[ (h|_\ell)(x_j) = g_j(u) = (h'|_\ell)(x_j) \]
  for all $1 \leq j \leq t$. Thus, $h|_\ell$ and $h'|_\ell$ are two
  codewords of $\code$ which are not equal, yet agree on at least $n -
  d + 1$
  points. By the distance property of $\code$, this is impossible. Hence, $h$ is
  unique.

  So far, we have established the the correspondence is well
  defined. It remains to show that it is one-to-one. First, we observe
  that it is clearly injective by definition. Next, observe that it is
  also onto: for any $h \in \code^{\otimes (m+1)}$, it holds that the
  tuple $(h|_{x_1}, \dots, h|_{x_t})$ is in $S$ and maps to $h$ under
  the correspondence. Therefore, it is one-to-one.
\end{proof}

\begin{proposition} \label{prop:interpolate-tuple}
  Let $\code$ be a $[n,k,d]_\Sigma$ code that is interpolable, and let $t = n -d+1$. Let $S$ be
  as in \Cref{prop:tuple-to-code-correspondence}. Then $S = \codemt$.
\end{proposition}
\begin{proof}
Fix coordinates $x_1,\ldots,x_t \in [n]$ and let $\phi_{x_1,\ldots,x_t}:\Sigma^t \to \code$ denote the corresponding linear interpolation map for $\code$. It suffices to show that all tuples $(g_1, \dots, g_t) \in \codemt$
  are contained in $S$ (which is a set defined with respect to $x_1,\ldots,x_t$). Indeed, given such a tuple, let $h: [n]^{m+1}
  \to \Sigma$ be the function defined by
  \[ h(z_1, \dots, z_{m}, x) = \interpol_{x_1, \dots, x_t} (g_1(z),
    \dots, g_t(z))(x). \]
  It follows from the definition of $\interpol$ that for any fixed
  $z_1, \dots, z_m \in [n]^m$, the function $h(z_1, \dots, z_m, \cdot)$ is a codeword of
  $\code$, and that for any $j \in [t]$ we have 
  \[ h(\cdot, \dots, \cdot, x_j) = g_j(\cdot, \dots, \cdot). \]
 Moreover, by the linearity of $\interpol$, it holds that
  for any fixed $x$, $h(\cdot, \dots, \cdot, x)$ is a linear
  combination of $g_1, \dots, g_t$ and is thus a codeword of
  $\code^{\ot m}$. This establishes that $h \in \code^{\ot (m+1)}$,
  and therefore that $(g_1, \dots, g_t) \in S$.
\end{proof}

\section{The tensor code test}
\label{sec:test}

Let $\code$ be an interpolable code (called the \emph{base code}) over alphabet $\Sigma$ with blocklength $n$, dimension $k$, and distance $d$, and let $m \geq 2$ be an integer. Informally, the augmented tensor code test proceeds as follows. With probability $1/2$, the referee performs the line-vs-point test as described in the introduction. With probability $1/2$, the referee performs a \emph{subcube commutation test}. 
Here, a \emph{subcube} of $[n]^m$ denotes a subset of points of the form $H_{x_{m-j+2},\ldots,x_{m}} = \{ (x_1,\ldots,x_{m-j+1},x_{m-j+2},\ldots,x_m) \in [n]^m : x_{1},\ldots,x_{m-j+1} \in [n] \}$ where $x_{m-j+2},\ldots,x_{m}$ are fixed elements of $[n]$. If $j = 1$, then we define the subcube $H$ to be all of $[n]^m$.
The purpose of the subcube commutation test is to enforce that for an average subcube $H_{x_{m-j+2},\ldots,x_{m}}$, the measurements corresponding to two randomly chosen points $u,v \in H_{x_{m-j+2},\ldots,x_{m}}$ approximately commute. The motivation for choosing the specific distribution on pairs of points used in the test comes from its use in the analysis, and in particular on how the induction is structured; see the induction step in the proof of Lemma~\ref{lem:induction}. 

The augmented tensor code test for the tensor code $\code^{\otimes m}$ is described precisely in \Cref{fig:test}. From here on we omit the qualifier ``augmented'' and simply refer to the test described in \Cref{fig:test} as the tensor code test.

\begin{definition}\label{def:tracial-strat}
A \emph{tracial strategy} $\strategy$ for the tensor code test is a tuple $(\tau,A,B,P)$ where $\tau: \algebra \to \C$ is a tracial state on a von Neumann algebra $\algebra$, and $A,B,P$ are the following projective measurements:
\begin{enumerate}
	\item \emph{Points measurements}: for every $u \in [n]^m$, let $A^u = \{A^u_a\}_{a \in \Sigma}$ denote the measurement corresponding to the points question $u$.
	\item \emph{Lines measurements}: for every axis-parallel line $\ell \subset [n]^m$, let $B^\ell = \{B^\ell_g\}_{g \in \code}$ denote the measurement corresponding to the lines question $\ell$.
	\item \emph{Pair measurements}: for every subcube $H$ and for every $u,v \in H$, let $P^{u,v} = \{P^{u,v}_{a,b} \}_{a,b\in \Sigma}$ denote the measurement corresponding to the points pair question $(u,v)$.
\end{enumerate}
\end{definition}

Aside from the discussion in Section~\ref{sec:general-states} all strategies considered in the paper are tracial strategies as defined in Definition~\ref{def:tracial-strat}, and we simply say ``strategy'', omitting the qualifier ``tracial''. 

\begin{definition}\label{def:tracial-good}
We say that $\strategy$ is an $(\eps,\delta)$-good (tracial) strategy if the following hold:
\begin{enumerate}
	\item \emph{(Consistency between points and lines:)}
	\[
		B^\ell_{[g \mapsto g(u) | a]} \simeq_\eps A^u_a\;,
	\]
	where the expectation is over a uniformly random axis-parallel line $\ell \subset [n]^m$ and a uniformly random point $u \in \ell$, and the answer summation is over values $a \in \Sigma$. We wrote $g(u)$ to denote $g(u_j)$ where $\ell$ varies in the $j$-th coordinate.

	\item \emph{(Consistency between points and pairs:)}
	\[
		P^{u,v}_{[(a,b) \mapsto a |a]} \simeq_\delta A^u_a \qquad \text{and} \qquad  P^{u,v}_{[(a,b) \mapsto b |b]} \simeq_\delta A^v_b\;,
	\]
where the expectation is over a uniformly random subcube $H = H_{x_{m-j+2},\ldots,x_{m}}$ (sampled by choosing a random $j \sim [m]$ and random $x_{m-j+2},\ldots,x_{m} \sim [n]$) and uniformly random points $u,v \sim H$, and the answer summation is over values $a \in \Sigma$ and $b \in \Sigma$ respectively.
\end{enumerate}
\end{definition}

{
\floatstyle{boxed} 
\restylefloat{figure}
\begin{figure}[htb!]
Perform one of the following tests with probability~$\tfrac{1}{2}$ each. 
\begin{enumerate}
	\item \textbf{Axis-parallel lines test:}
		Let $u \sim [n]^m$ be a uniformly random point.
		Select $j \sim \{1,\ldots,m\}$ uniformly at random
		and let $\ell = \{ (u_1,\ldots,u_{j-1},s,u_{j+1},\ldots,u_m) \in [n]^m : s \in [n] \}$
		be the axis-parallel line passing through $u$ in the $j$-th direction.
		\begin{itemize}
			\item[$\circ$] Give $\ell$ to prover A and receive $g \in \code$.
			\item[$\circ$] Give $u$ to prover B and receive $a \in \Sigma$.
		\end{itemize}
		Accept if $g(u_j) = a$. Reject otherwise.
	\item \textbf{Subcube commutation test:}
	Select $j \sim \{1,\ldots,m\}$ uniformly at random, and select $x_{m-j+2},\ldots,x_{m} \sim [n]$ uniformly at random. Select $u,v$ independently and uniformly at random from the subcube $H_{x_{m-j+2},\ldots,x_{m}}$. Select $t \in \{0,1\}$ uniformly at random.
		\begin{itemize}
			\item[$\circ$] If $t = 0$, then give $(u,v)$ to prover A and receive $(b_u,b_v) \in \Sigma^2$; otherwise, give $(v,u)$ to prover A and receive $(b_v,b_u) \in \Sigma^2$.
			\item[$\circ$] Give $u$ to prover B;
				receive $a \in \Sigma$.
		\end{itemize}
		Accept if $b_u = a$. Reject otherwise.	
	\end{enumerate}
	\caption{The tensor code test.\label{fig:test}}
\end{figure}
}

Thus a strategy $\strategy$ that is $(\eps,\delta)$-good according to
Definition~\ref{def:tracial-good} passes the tensor code test with probability
$1 - \frac{1}{2} ( \eps + \delta) \geq 1 - \eps - \delta$.
Conversely, a strategy $\strategy$ that passes the tensor code test with
probability $1 - \eps$ is $(2\eps,4\eps)$-good.

%

\section{Analysis of quantum-soundness of tensor code test}
\label{sec:main-loop}


Our main result is the following. Let $\code$ denote an interpolable $[n,k,d]_\Sigma$ code.

\begin{restatable}[Quantum soundness of the tensor code test]{theorem}{thmmain}
\label{thm:main}
There exists a function 
\[\eta(m,t,r,\eps,n^{-1}) = \poly(m,t,r) \cdot \poly(\eps,n^{-1},e^{-\Omega(r/m^2)})\]
 such that the following holds. 
Let $\strategy = (\tau,A,B,P)$ be a synchronous strategy that passes the (augmented) tensor code test for $\code^{\otimes m}$ with probability $1 - \eps$, and let $r\geq 12mt$ be an integer. Let $\algebra$ denote the algebra associated with the tracial state $\tau$. Then there exists a projective measurement $\{ G_c \}_{c \in \code^{\otimes m}} \subset \algebra$ such that 
\[
	\E_{u \sim [n]^m} \sum_{c \in \code^{\otimes m}} \tau \big( G_c \, A^u_{c(u)} \big) \,\geq\, 1 - \eta\;,
\]
where $\eta = \eta(m,t,r,\eps,n^{-1})$ with $t = n - d + 1$. 
\end{restatable}

The proof of \Cref{thm:main} proceeds by induction on the dimension $m$. As described in the proof overview (\Cref{sec:proof-overview}) there are two main steps in the induction, called ``self-improvement'' and ``pasting''. The self-improvement step is based on the following lemma, whose proof we give in \Cref{sec:self-improvement}.

\begin{lemma}[Self-improvement]
\label{lem:self-improvement}
There exists a function $\zeta(m,t,\eps,n^{-1}) = \poly(m,t) \cdot \poly(\eps, n^{-1})$ such that the following holds. 
Let $\strategy = (\tau, A,B,P)$ be an $(\eps,\delta)$-good strategy for the tensor code test for $\code^{\otimes m}$. 
Suppose that $\{G_g\}$ is a measurement with outcomes in $\code^{\otimes m}$
satisfying:
\begin{enumerate}
	\item[] (Consistency with $A$): On average over $x \in [n]^m$, $G_{[g \mapsto g(x) \mid a]} \simeq_\nu A^x_a$.
\end{enumerate}
Then there exists a projective submeasurement $H = \{H_h\} \subset \algebra$ with outcomes in $\code^{\otimes m}$ with the following properties:
\begin{enumerate}
	\item \label{enu:self-improvement-completeness} (Completeness): $\tau(H) \geq 1 - \nu - \zeta$ where $H = \sum_h H_h$.
	\item \label{enu:self-improvement-consistency} (Consistency with $A$): On average over $x \sim [n]^m$, $H_{[h \mapsto h(x) \mid a]} \simeq_\zeta A^x_a$.
	\item \label{enu:self-improvement-boundedness} (Agreement is explained by the complete part of $H$): There exists a positive linear map $\psi:\algebra\to\C$ such that
	\[
		\psi(\Id - H) \leq \zeta
	\]
	and for each $h \in \code^{\otimes m}$ and all positive $X \in \algebra$, we have
	\[
		\psi(X) \geq \E_{x \sim [n]^m} \tau(X \cdot A^x_{h(x)}).
	\]
\end{enumerate}
where $\zeta = \zeta(m,t,\eps,n^{-1})$ with $t = n-d+1$.
\end{lemma}

While the first two conclusions of the lemma are intuitive, the third one deserves some explanation. Intuitively the condition guarantees the existence of a certain ``measure'', $\psi$, such that firstly the measure  of the ``incomplete part'' of $H$, i.e. $\Id-H$, is small (in the lemma, one should think of $\zeta$ as an error parameter that is $\ll \nu$), and secondly $\psi$ satisfies a ``non-triviality'' condition in that it is required to be large on any $X$ such that also $\E_x \tau(X \cdot A^x_{h(x)})$ is large, i.e.\ such that $X$ correlates well with the averaged operator $\E_x A^x_{h(x)}$ for a given $h$. This condition is used in the next lemma, Lemma~\ref{lem:pasting}, to guarantee that the $H$ measurements from Lemma~\ref{lem:self-improvement} for different subcubes need to be compatible in the sense that their incomplete parts mostly overlap. If this condition were not present in the assumptions of Lemma~\ref{lem:pasting} then instead of $\nu = \kappa + error$ we could only guarantee $\nu = t\cdot \kappa + error$, which is insufficient to complete the $m$ steps of induction.

The pasting step is based on the following lemma, whose proof is given
in \Cref{sec:pasting}.

\begin{restatable}[Pasting]{lemma}{pasting}
\label{lem:pasting}
There exists a function 
\begin{equation}\label{eq:kappa-def}
\nu(m,t,r,\eps,\delta,\zeta,n^{-1}) = \poly(m,t,r) \cdot  \poly(\eps,\delta,\zeta,n^{-1}) 
		\end{equation}
		such that the following holds.
Let $\strategy = (\tau, A,B,P)$ be an $(\eps,\delta)$-good strategy
for the code $\code^{\otimes m+1}$ with points measurements $A =
\{A^{u,x}\}_{(u,x) \in [n]^m \times [n]}$. 
Let $\{G^x\}_{x \in [n]}$ denote projective submeasurements with outcomes in $\code^{\otimes m}$ with the following properties:
\begin{enumerate}
	\item \label{enu:pasting-completeness} (Completeness):
          $\tau(G) \geq 1 - \kappa$ where $G = \E_x \sum_g G^x_g$. 
	\item \label{enu:pasting-consistency} (Consistency with $A$): On average over $(u,x) \in [n]^m \times [n]$, $G^x_{[g \mapsto g(u) \mid a]} \simeq_\zeta A^{u,x}_a$.
	\item \label{enu:pasting-boundedness} (Agreement is explained
          by the complete part of $G$): For each $x \in [n]$ there
          exists a positive linear map $\psi^x$ such that
	\[
		\E_{x }\psi^x (\Id - G^x) \leq \zeta
	\]
	and for all $x \in [n]$, $g \in \code^{\otimes m}$, and positive $X \in \algebra$, we have
	\[
		\psi^x(X) \geq \E_{u \sim [n]^m} \tau(X \cdot A^{u,x}_{g(u)}).
	\]
\end{enumerate}
 Then for all integers $r \geq 12mt$ there exists a ``pasted'' measurement $H = \{H_h\} \subset \algebra$ with outcomes $h\in \code^{\otimes (m+1)}$, satisfying the following property:
\begin{quote}
	 (Consistency with $A$): \label{enu:pasting-consistency-2} On average over $(u,x) \in [n]^m \times [n]$, $H_{[h \mapsto h(u,x) \mid a]} \simeq_{\mu} A^{u,x}_a$
        \end{quote}
        where 
				\begin{equation}\label{eq:sigma-def}
				\mu = \mu(\kappa,m,t,r,\eps,\delta,\zeta,n^{-1}) \,=\, \kappa\Big(1+\frac{1}{3m}\Big) + \nu(m,t,r,\eps,\delta,\zeta,n^{-1}) + e^{- \frac{r}{72m^2}}
				\end{equation}
				 with $t = n - d + 1$.
      \end{restatable}

      We now put these two steps together and prove
      \Cref{thm:main}. The proof is by induction on $m$. The inductive argument  directly establishes the following
      lemma which, together with the self-improvement lemma, implies \Cref{thm:main}.

\begin{lemma}
  \label{lem:induction}
  There exists a function 
	\[\sigma(m,t,r,\eps,\delta, n^{-1})= \poly(m, t,r) \cdot \big( \poly(\eps,\delta,n^{-1}) + \exp(-\Omega(r/m^2)) \big)\]
	such that the following holds. Let $\strategy = (\tau,A,B,P)$ be an $(\eps,\delta)$-good strategy for the code $\code^{\otimes m}$ and let $r\geq 12mt$. Let $\algebra$ denote the algebra associated with $\strategy$. Then there exists a measurement $\{ G_c \} \subset \algebra$ with outcomes in $\code^{\otimes m}$ such that 
\[
	\E_{u \sim [n]^m} \sum_{c \in \code^{\otimes m}} \tau \big(
          G_c \, A^u_{c(u)} \big) \,\geq\,1 - \sigma\;,
\]
where $\sigma = \sigma(m,t,r,\eps,\delta, n^{-1})$ with $t = n - d + 1$. 
\end{lemma}

\begin{proof}[Proof of \Cref{thm:main} from \Cref{lem:induction}]
  The measurement $\{G_c\}$ produced by \Cref{lem:induction} satisfies
  the conclusions of \Cref{thm:main} except it is not projective. To
  remedy this we apply the self-improvement lemma to produce a
  projective submeasurement, and then complete the result into a
  projective measurement.

  More formally, let $\{H_c\}$ be the projective submeasurement
  guaranteed by applying \Cref{lem:self-improvement} to $G_c$. This
  submeasurement has the following properties.
  \begin{enumerate}
  \item (Projectivity): $H_c$ is projective.
  \item (Completeness): If $H = \sum_h H_h$, then $\tau(H) \geq 1 -
   \sigma - \zeta$.
  \item (Consistency with $A$): On average over $u \sim [n]^m$, $H_{[h
      \mapsto h(u) \mid a]} \simeq_\zeta A^u_a$.
  \end{enumerate}
  Let $\tilde{H}$ be the measurement obtained by completing
  $H$ arbitrarily, that is by setting $\tilde{H}_{c^*} = H_{c^*} + \Id - \sum_h
  H_h$ for an arbitrarily chosen outcome $c^*$, and setting
  $\tilde{H}_{c} = H_{c}$ for all other outcomes $c \neq c^*$. Then
  $\tilde{H}$ has the following properties:
  \begin{enumerate}
    \item (Projectivity): $\tilde{H}$ is projective. This holds by the projectivity of $H$, which implies that
      $\Id - \sum_h H_h$ is a projector and is orthogonal to each
      $H_{c}$ for every outcome $c$.
    \item (Completeness): $\sum_{h} \tilde{H}_h = \Id$. This holds by
      definition.
    \item (Consistency with $A$): We evaluate the (in)consistency between $\tilde{H}$ and the points measurements. We have
    \begin{align*}
    	\E_u \sum_{a \neq b} \tau \Big (\tilde{H}_{[h \mapsto h(u) \mid a]} \, A^u_b \Big) &= \E_u \sum_{h, b \neq h(u)} \tau \Big (\tilde{H}_h \, A^u_b \Big) \\
	&= \E_u \sum_{h, b \neq h(u)} \tau \Big (H_h \, A^u_b \Big) + \tau \Big ((\Id - \sum_h H_h) \, \sum_{b \neq c^*(u)} A^u_b \Big) \\
	&\leq \zeta + \tau \Big (\Id - \sum_h H_h \Big) \\
	&\leq 2\zeta + \sigma
    \end{align*} 
    where in the second-to-last inequality we used that $H$ is $\zeta$-consistency with the points measurements $\{A^u\}$ and that $ \sum_{b \neq c^*(u)} A^u_b \leq \Id$, and in the last inequality we used that the completeness of $H$ is at least $1 - \zeta - \sigma$. Thus on average of $u \sim [n]^m$, we have
    \[
      \tilde{H}_{[h \mapsto h(u) \mid a]} \simeq_{2\zeta + \sigma}
      A^u_a.
     \]
    \end{enumerate}
    Setting $\delta = \eps$, and  $\eta(m,t, r,\eps, n^{-1}) = 2\zeta(m,t,\eps, n^{-1}) +
    \sigma(m,t,r,\eps,\delta, n^{-1})  = \poly(m,t,r) \cdot \Big( \poly(\eps, n^{-1})+  e^{-\Omega(r/m^2)} \Big)$, we see that $\tilde{H}$
    satisfies the conclusion of \Cref{thm:main}.
\end{proof}

\medskip

\begin{proof}[Proof of \Cref{lem:induction}]


Define $s_m = 2m$. We argue that there is a choice of universal constants $c_0, c_1, c_2, c_3 \geq 1$ such that \Cref{lem:induction} holds with the function
\[
\sigma(m,t,r,\eps, \delta,n^{-1}) = s_m \cdot \Big( c_0 \cdot (m t r)^{c_1} \cdot (\eps^{1/c_2} + \delta^{1/c_2} + n^{-1/c_2}) + e^{-c_3 r/m^2} \Big) \;.
\]
The constants $c_0, c_1, c_2, c_3$ are set in terms of the universal constants that are implicit in the $\poly(\cdot)$ notation for the functions $\zeta, \nu,\mu$ from \Cref{lem:self-improvement} and \Cref{lem:pasting}. We now explicitly identify these constants:
\begin{align}
    \zeta(m,t,\eps,n^{-1}) &= z_0 \cdot (mt)^{z_1} \cdot (\eps^{1/z_1} + n^{-1/z_1}) \\
    \nu(m,t,r,\eps, \delta, \zeta,  n^{-1}) &= b_0 \cdot (mtr)^{b_1} \cdot ( \eps^{1/b_2} + \delta^{1/b_2} + \zeta^{1/b_2} + n^{-1/b_2}) \\
    \mu(\kappa,m,t,r,\eps,\delta,\zeta,n^{-1}) &= \kappa\Big(1+\frac{1}{3m}\Big) + \nu(m,t,r,\eps,\delta,\zeta,n^{-1}) + e^{- \frac{r}{72m^2}}
\end{align}
where $z_0,z_1,z_2,b_0,b_1,b_2 \geq 1$. Now we set the constants $c_0,c_1,c_2,c_3$ as follows.
\begin{enumerate}
	\item Let $c_0 = 4 b_0 z_0$.
	\item Let $c_1 = b_1 + z_1 + 1$. 
\item Let $c_2 = b_2 z_2$.
\item Let $c_3 = 1/72$. 
\end{enumerate}
We prove that \Cref{lem:induction} holds with this choice of function $\sigma(m,t,r,\eps, \delta,n^{-1})$ via induction on $m$, starting with $m=1$.


\paragraph{Base case.}
For the case $m = 1$ 
there is only one line $\ell = \{(\beta):\beta \in [n]\}$ .
This means that the lines measurement $B = \{B^\ell\}$
contains a single measurement $B^\ell$.
Because~$\strategy$ passes the tensor codes test with probability $1-\eps$,
it is $(2\eps, 4\eps)$-good. 
As a result,
\begin{equation*}
B^{\ell}_{[g \mapsto g(u) \mid a]} \simeq_{2\eps} A^u_a \;.
\end{equation*}
The theorem follows by setting
$G = \{G_g\}$ equal to $B^\ell$, since $c_0 \geq 2$ (which implies that $\sigma \geq 2\eps$). 

\paragraph{Induction step.}
Assuming, as the inductive hypothesis, that \Cref{lem:induction} holds for some $m \geq 1$, we prove it for $m+1$. Let $\strategy =(\tau, A, B, P)$ be a $(\eps,\delta)$-good strategy for the tensor code test with $\code^{\otimes (m+1)}$, and let $r \geq 12(m+1)t$.  
For each $x \in [n]$,
we define $\strategy^x = (\tau, A^x, B^x, P^x)$ to be the strategy restricted on the last coordinate being~$x$.
In other words,
we write
\begin{IEEEeqnarray*}{rClrClrCl}
A^x &:=& \{A^{u,x}\}_{u \in [n]^m},\quad&\quad
B^x &:=& \{B^{\ell,x}\}_{\ell \subset [n]^m}, \quad&\quad 
P^x  &:=& \{P^{u, v,x}\}_{u, v \in [n]^m},\\
\text{where}\quad
A^{u, x} &:=& A^{u \circ x},&
B^{\ell, x} &:=& B^{\ell \circ x},&
P^{u, v, x} &:=& P^{u \circ x, v \circ x},
\end{IEEEeqnarray*}
where ``$\circ$'' denotes string concatenation, and $\ell \circ x = \{\alpha \circ x: \alpha \in \ell\}$.
Then $\strategy^x$ is a strategy for the $\code^{\otimes m}$  tensor code test.
Write $1-\eps_x$ and $1-\delta_x$ for the probability that $\strategy^x$
passes the axis-parallel lines test and subcube commutation test, respectively.
Then $\strategy^x$ is an $(\eps_x, \delta_x)$-good strategy.

We claim that on average over $x$ chosen uniformly from $[n]$,
\[\E_{x \sim [n]} \eps_x \leq \left(\frac{m+1}{m} \right) \cdot \eps,\; \E_{x
    \sim [n]} \delta_x \leq \left(\frac{m+1}{m} \right) \cdot \delta. \]
Indeed, observe that in the axis-parallel lines test over $[n]^{m+1}$, the
probability that an axis-parallel line along axis $m+1$ is chosen is
$1/(m+1)$. Thus,
\[ \eps = \frac{1}{m+1} \cdot \Pr \Big(\text{success for $\ell$ parallel to axis
    $m+1$} \Big) + \frac{m}{m+1} \cdot \E_{x} \eps_x \geq \frac{m}{m+1}
  \cdot \E_x \eps_x, \]
which yields the claimed bound on the expectation of $\eps_x$. A
similar calculation holds for the subcube test (the probability that
the subcube does not fix the $m+1$st coordinate is $1/(m+1)$).

Now we apply the inductive hypothesis, noting that $r \geq 12mt$. To do so, fix an $x \in [n]$.
The inductive hypothesis states the existence of a measurement $G^x = \{G^x_g\}_{g \in \code^{\otimes m}}$
such that on average over $u \in [n]^m$,
\begin{equation*}
G^x_{[g\mapsto g(u) \mid a]} \simeq_{\sigma_x} A^{u, x}_a,
\end{equation*}
where $\sigma_x = \sigma(m,t,r,\eps_x, \delta_x, n^{-1})$.
Setting $\zeta_x = \zeta(m,t,\eps_x, n^{-1})$, self-improvement (\Cref{lem:self-improvement})
then implies the existence of a projective submeasurement $H^x = \{H^x_h\}_{h \in \code^{\otimes m}}$ with the following properties:
\begin{enumerate}
\item (Completeness): If $H^x = \sum_h H^x_h$, then $\tau(H^x) \geq 1 - \sigma_x - \zeta_x$.
\item (Consistency with $A^x$): On average over $u \sim [n]^m$, $H^x_{[h \mapsto h(u) \mid a]} \simeq_{\zeta_x} A^{u, x}_a$.
\item (Agreement is explained by the complete part of $H^x$):  There exists a positive linear map $\psi^x$ such that
  \[
    \psi^x(\Id - H^x) \leq {\zeta_x},
  \]
  and for each $h \in \code^{\otimes m}$ and all positive $X \in \algebra$, we have
  \[
    \psi^x(X) \geq \E_{u \sim [n]^m} \tau(X \cdot A^{u, x}_{h(u)}).
  \]
\end{enumerate}

These guarantees are almost of the form demanded by the
hypotheses of the pasting lemma (\Cref{lem:pasting}). However,
the error bounds depend on $x$, whereas we would like the errors
to be averaged over $x$. To obtain such bounds, we will bound the
average of $\sigma_x, \zeta_x$ over $x \sim [n]$. Observe that
$\sigma_x$ and $\zeta_x$ depend on $x$ only through $\eps_x, \delta_x \in
[0,1]$, and are polynomial functions of $\eps_x, \delta_x$. 
For any
distribution over a random variable $\alpha$ drawn from the
interval $[0,1]$, for $c \geq 1$, it holds that $\E \alpha^c \leq \E \alpha$,
and by concavity, for any $c < 1$, $ \E \alpha^c
\leq (\E \alpha)^c$. Hence, the averaged quantities $\E_{x}
\sigma_x$ and $\E_{x} \zeta_x$, which we denote $\hat{\sigma}$ and $\hat{\zeta}$, are polynomial functions of
$\E_{x} \eps_x \leq \frac{m+1}{m} \eps$ and $\E_{x} \delta_x \leq
\frac{m+1}{m} \delta$: 
\begin{align}
  \hat{\sigma} = \E_{x} \sigma_x &= s_m \cdot c_0 (mtr)^{c_1} \cdot \E_{x} (\eps_x^{1/c_2} + \delta_x^{1/c_2} + n^{-1/c_2}) +  s_m \cdot e^{-c_3r/m^2} \\
    &\leq s_m \cdot c_0 (m t r)^{c_1} \cdot \left( \left(\frac{m+1}{m} \eps\right)^{1/c_2} + \left(\frac{m+1}{m} \delta\right)^{1/c_2} + n^{-1/c_2}\right)  +  s_m \cdot e^{-c_3r/m^2}\\
    &\leq s_m \cdot c_0 (mt r)^{c_1} \cdot \Big( \frac{m+1}{m} \Big)^{\frac{1}{c_2}} \cdot \left(  \eps^{1/c_2} +  \delta^{1/c_2} + n^{-1/c_2} \right)  +  s_m \cdot e^{-c_3r/m^2} \label{eq:induction-hat-sigma}\\
  \hat{\zeta} = \E_{x} \zeta _x &= z_0 (mt)^{z_1} \cdot \E_{x} (\eps^{1/z_2} + n^{-1/z_2}) \\
  &\leq z_0 (mt)^{z_1} \cdot \left( \left(\frac{m+1}{m} \eps\right)^{1/z_2} + n^{-1/z_2}\right) \\
  &\leq z_0 ((m+1) t)^{z_1} \cdot( \eps^{1/z_2} + n^{-1/z_2}) \label{eq:induction-hat-zeta}
\end{align}
Thus we have the following guarantees on the submeasurements $\{H^x\}$: 
\begin{enumerate}
\item (Completeness): If $H = \E_{x \sim [n]^m} \sum_h H^x_h$, then $\tau(H) \geq 1 - \hat{\sigma} - \hat{\zeta}$.
\item (Consistency with $A^x$): On average over $(u,x) \sim [n]^m
  \times [n] $, $H^x_{[h \mapsto h(u) \mid a]} \simeq_{\hat{\zeta}} A^{u, x}_a$.
\item (Agreement is explained by the complete part of $H^x$): For each
  $x$ there exists a positive linear map $\psi^x$ such that
  \[
    \E_{x} \psi^x(\Id - H^x) \leq {\hat{\zeta}},
  \]
  and for each $h \in \code^{\otimes m}$ and all positive $X \in \algebra$, we have
  \[
    \psi(X) \geq \E_{u \sim [n]^m} \tau(X \cdot A^{u, x}_{h(u)}).
  \]
\end{enumerate}
Define $\hat{\kappa} = \hat{\sigma} + \hat{\zeta}$. We see that the measurements $H^x$ satisfy the hypotheses of the pasting lemma (\Cref{lem:pasting}) with error parameters $\hat{\kappa}$ and $\hat{\zeta}$. Applying the lemma, we conclude that the resulting measurement $\tilde{H}$ satisfies
the following consistency relation with $A$: on average over $(u, x)
\in [n]^{m} \times [n]$,
\[ H_{[h \mapsto h(u,x) \mid a]} \simeq_{\mu} A^{u,x}_a, \]
where 
\begin{align} 
\mu &= \mu(\hat{\kappa},m,t,r,\eps,\delta,\hat{\zeta},n^{-1}) \,=\, \hat{\kappa}\Big(1+\frac{1}{3m}\Big) + \nu(m,t,r,\eps,\delta,\hat{\zeta},n^{-1}) + e^{- c_3 r/m^2} \\
&= \Big(\hat{\sigma} + \hat{\zeta} \Big)\Big(1 + \frac{1}{3 m} \Big) + b_0(mtr)^{b_1} \cdot \big(\eps^{1/b_2} + \delta^{1/b_2} + \hat{\zeta}^{1/b_2} + n^{-1/b_2} \big) +e^{-c_3r/m^2}. \label{eq:induction-mu}
\end{align}
If we show that $\mu \leq \sigma(m+1,t,r,\eps,\delta,n^{-1})$, then we will have established the inductive step and we would be done. To proceed, we first substitute~\eqref{eq:induction-hat-zeta} in for $\hat{\zeta}$. Let $b' = b_2 z_2$. 
\begin{align*}
	\hat{\zeta}^{1/b_2} &\leq (z_0 ((m+1) t)^{z_1})^{1/b_2} \cdot( \eps^{1/z_2} + n^{-1/z_2})^{1/b_2} \\
	&\leq z_0 ((m+1) t)^{z_1} \cdot ( \eps^{1/b'} + n^{-1/b'}) 
\end{align*}
where we used that $z_0 ((m+1) t)^{z_1} \geq 1$ and $b_2 \geq 1$. Plugging these bounds into~\eqref{eq:induction-mu}, we get
\begin{align}
	\mu &\leq \big(1 + \frac{1}{3m} \big) \cdot \big( \hat{\sigma}+ \hat{\zeta}\big) + 2b_0(mtr)^{b_1} \cdot z_0 ((m+1) t)^{z_1} \cdot \big(\eps^{1/b'} + \delta^{1/b'} + n^{-1/b'} \big) +e^{-c_3r/m^2} \notag \\
	&\leq \big(1 + \frac{1}{3m} \big) \cdot \big( \hat{\sigma}+ \hat{\zeta}\big) + 2 b_0 z_0 ((m+1)tr)^{b_1 + z_1} \big(\eps^{1/b'} + \delta^{1/b'} + n^{-1/b'} \big) +e^{-c_3r/m^2}~. \label{eq:induction-mu-2}
\end{align}
Similarly, we can upper bound
\[
\big(1 + \frac{1}{3m} \big) \cdot \hat{\zeta} \leq 2 \cdot z_0 ((m+1) t)^{z_1} \cdot ( \eps^{1/b'} + n^{-1/b'})  
\]
because, for example, $\eps^{1/b_2} \leq \eps^{1/b'}$. Thus we get that~\eqref{eq:induction-mu-2} is at most
\begin{align}
 \text{\eqref{eq:induction-mu-2}} &\leq \big(1 + \frac{1}{3m} \big)\hat{\sigma} + 4 b_0 z_0 ((m+1)tr)^{b_1 + z_1} \big(\eps^{1/b'} + \delta^{1/b'} + n^{-1/b'} \big) +e^{-c_3r/m^2} \label{eq:induction-mu-3}
\end{align}
Substituting~\eqref{eq:induction-hat-sigma} in for $\hat{\sigma}$ and using that $c_2 \geq b'$ we get
\begin{align}
\text{\eqref{eq:induction-mu-3}} &\leq \big(1 + \frac{1}{3m} \big) s_m \cdot c_0 (mt r)^{c_1} \cdot \Big( \frac{m+1}{m} \Big)^{\frac{1}{c_2}} \cdot \left(  \eps^{1/c_2} +  \delta^{1/c_2} + n^{-1/c_2} \right)  \notag \\
& \qquad \qquad + 4 b_0 z_0 ((m+1)tr)^{b_1 + z_1} \cdot \big(\eps^{1/c_2} + \delta^{1/c_2} + n^{-1/c_2} \big) \notag \\
& \qquad \qquad + \Big( \big(1 + \frac{1}{3m} \big) s_m + 1 \Big) \cdot e^{-c_3r/m^2}~. \label{eq:induction-mu-4}
\end{align}
Next, using that $\big(1 + \frac{1}{3m} \big) s_m \leq s_{m+1}$, we have
\begin{align}
\text{\eqref{eq:induction-mu-4}} &\leq \Big( s_{m+1} \cdot c_0 (mt r)^{c_1} \cdot \Big( \frac{m+1}{m} \Big)^{\frac{1}{c_2}} + 4 b_0 z_0 ((m+1)tr)^{b_1 + z_1} \Big) \cdot  \big(\eps^{1/c_2} + \delta^{1/c_2} + n^{-1/c_2} \big) \notag \\
& \qquad \qquad + \Big( \big(1 + \frac{1}{3m} \big) s_m + 1 \Big) \cdot e^{-c_3r/m^2}~. \label{eq:induction-mu-5}
\end{align}
If we argue both
\begin{align}
s_{m+1} \cdot c_0 (mt r)^{c_1} \cdot \Big( \frac{m+1}{m} \Big)^{\frac{1}{c_2}} + 4 b_0 z_0 ((m+1)tr)^{b_1 + z_1}  \qquad & \leq \qquad s_{m+1} c_0 ((m+1) t r)^{c_1} \label{eq:induction-mu-6} \\
\big(1 + \frac{1}{3m} \big) s_m + 1 \qquad & \leq \qquad s_{m+1} \label{eq:induction-mu-7}
\end{align}
then we would be done as~\eqref{eq:induction-mu-5} would be at most $\sigma(m+1,t,r,\eps,\delta,n^{-1})$ as desired. To argue~\eqref{eq:induction-mu-6} we first note that it is equivalent to showing that
\[
s_{m+1} c_0 (tr)^{c_1} \Big( (m+1)^{c_1} - m^{c_1 - \frac{1}{c_2}} (m+1)^{\frac{1}{c_2}} \Big) \geq 4 b_0 z_0 ((m+1)tr)^{b_1 + z_1}~.
\]
Since $s_{m+1} \geq 1$, $c_0 \geq 4 b_0 z_0$ by definition, and $c_1 \geq b_1 + z_1 + 1$ by definition, it suffices to argue that
\[
	(m+1)^{c_1} - m^{c_1 - \frac{1}{c_2}} (m+1)^{\frac{1}{c_2}} \geq (m+1)^{b_1 + z_1}~.
\]
This is equivalent to showing
\[
(m+1)^{c_1 - (b_1 + z_1)} - m^{c_1 - \frac{1}{c_2}} (m+1)^{\frac{1}{c_2} - (b_1 + z_1)}  \geq 1~.
\]
Since $b_1,z_1,c_2 \geq 1$ it holds that $(m+1)^{\frac{1}{c_2} - (b_1 + z_1)} \leq m^{\frac{1}{c_2} - (b_1 + z_1)}$ and thus 
\begin{align*}
(m+1)^{c_1 - (b_1 + z_1)} - m^{c_1 - \frac{1}{c_2}} (m+1)^{\frac{1}{c_2} - (b_1 + z_1)} &\geq (m+1)^{c_1 - (b_1 + z_1)} - m^{c_1 - \frac{1}{c_2}} m^{\frac{1}{c_2} - (b_1 + z_1)} \\
&= (m+1)^{c_1 - (b_1 + z_1)} - m^{c_1 - (b_1 + z_1)} \\
&\geq 1
\end{align*}
as desired, where we used that $c_1 \geq b_1 + z_1 + 1$ again.

To argue~\eqref{eq:induction-mu-7} one can see by definition of $s_m$ and $s_{m+1}$ that
\begin{align*}
	s_{m+1} - \big(1 + \frac{1}{3m} \big) s_m = 2(m+1) - 2m\big(1 + \frac{1}{3m} \big) \geq 1
\end{align*}
as desired.

\end{proof}

\subsection{Extension to general states}
\label{sec:general-states}

As mentioned in the introduction, Theorem~\ref{thm:main} is formulated in the framework of synchronous strategies, which is particularly convenient to state and prove the result as it involves a single Hilbert space and set of measurement operators. For some applications one may wish to extend the guarantees of Theorem~\ref{thm:main} to the more general setting of strategies that are not  assumed to be perfectly synchronous a priori. This extension is essentially provided in the work~\cite{vidick2021almost}. For completeness in this section we include all details of the reduction. 

We first recall the following definition of a (not necessarily synchronous) two-prover strategy for a game $G$, which we immediately specialize to the case of the tensor code test. 

\begin{definition}
A \emph{two-prover strategy} $\strategy$ for the tensor code test is a tuple $(\ket{\psi},A,B,P,\tilde{A},\tilde{B},\tilde{P})$ where $\ket{\psi}\in \mH_A\otimes \mH_B$ is a state in the tensor product of two finite-dimensional Hilbert spaces $\mH_A$ and $\mH_B$, and $A,B,P$ (resp.\ $\tilde{A},\tilde{B},\tilde{P}$) are the following families of projective measurements:
\begin{enumerate}
	\item \emph{Points measurements}: for every $u \in [n]^m$, let $A^u = \{A^u_a\}_{a \in \Sigma}$ (resp.\  $\tilde{A}^u = \{\tilde{A}^u_a\}_{a \in \Sigma}$)
	denote the first (resp.\ second) prover's measurement corresponding to the points question $u$.
	\item \emph{Lines measurements}: for every axis-parallel line $\ell \subset [n]^m$, let $B^\ell = \{B^\ell_g\}_{g \in \code}$ (resp.\ $\tilde{B}^\ell = \{\tilde{B}^\ell_g\}_{g \in \code}$)  denote the first (resp.\ second) prover's measurement corresponding to the lines question $\ell$.
	\item \emph{Pair measurements}: for every subcube $H$ and for every $u,v \in H$, let $P^{u,v} = \{P^{u,v}_{a,b} \}_{a,b\in \Sigma}$ (resp.\ $\tilde{P}^{u,v} = \{\tilde{P}^{u,v}_{a,b} \}_{a,b\in \Sigma}$) denote the first (resp.\ second) prover's measurement corresponding to the points pair question $(u,v)$.
\end{enumerate}
\end{definition}

This definition is the same as Definition~\ref{def:tracial-strat} given at the start of Section~\ref{sec:test}, except that we have decoupled the measurement operators associated with each prover and the tracial state $\tau$ has been replaced by the state $\ket{\psi}$, with the additional tensor product structure. (There is a more general class of strategies where all prover operators act on the same possibly infinite-dimensional Hilbert space $\mH$ but all operators associated with one prover are required to commute with the other prover's operators; we do not consider such \emph{commuting-operator} strategies here.)

Because we do not assume a priori that the provers must employ the same measurement operators we modify the tensor code test by adding a ``synchronicity test' described in Figure~\ref{fig:test-2}. To distinguish this test from the original tensor code test we refer to it as ``the two-prover tensor code test.''

{
\floatstyle{boxed} 
\restylefloat{figure}
\begin{figure}[H]
With probably $\tfrac{1}{3}$ each, perform either of the two tests described in Figure~\ref{fig:test}, where the roles of ``prover $A$'' and ``prover $B$'' are attributed at random (under the constraint that each prover gets assigned a different role). With the remaining probability $\frac{1}{3}$ perform the following test. 
\begin{enumerate}
	\item[3.] \textbf{Synchronicity test:}
	Let $\ell\subset [n]^m$ be a line sampled according to the same distribution as in the axis-parallel line test, i.e.\ $\ell$ is a uniformly random axis-parallel line. Let $(u,v)$ be a pair of points sampled according to the same distribution as in the commutation test, i.e.\ $u$ and $v$ are both uniformly distributed in the same random subcube. With probability $\frac{1}{3}$ each, send either $\ell$, or $u$, or the pair $(u,v)$, to both provers. Accept if and only if the provers' answers are identical. 
		\end{enumerate}
	\caption{The two-prover tensor code test.\label{fig:test-2}}
\end{figure}
}

The following is analogous to Definition~\ref{def:tracial-good}. 

\begin{definition}\label{def:good-2}
We say that $\strategy$ is an $(\eps,\delta,\xi)$-good strategy for the two-prover tensor code test if the following hold:
\begin{enumerate}
	\item \emph{(Consistency between points and lines:)}
The strategy passes the axis-parallel line test with probability at least $1-\eps$. Formally, 
	\[
	\E_{\ell \subset [n]^m}\E_{u\in\ell}\sum_{a\in\Sigma} \,\frac{1}{2}\Big( \bra{\psi} B^\ell_{[g \mapsto g(u) | a]}\otimes \tilde{A}^u_a \ket{\psi} + \bra{\psi} A^u_a \otimes \tilde{B}^\ell_{[g \mapsto g(u) | a]} \ket{\psi}\Big) \,\geq\,1-\eps\;,
	\]
	where the expectation is over a uniformly random axis-parallel line $\ell \subset [n]^m$ and a uniformly random point $u \in \ell$. 
	\item \emph{(Consistency between points and pairs:)}
	\begin{align*}
\E_{\substack{H\\ u,v\sim H}}	\sum_{a,b\in\Sigma}\, \frac{1}{2}\Big(\bra{\psi} P^{u,v}_{[(a,b) \mapsto a |a]} \otimes \tilde{A}^u_a\ket{\psi}+\bra{\psi} A^u_a \otimes \tilde{P}^{u,v}_{[(a,b) \mapsto a |a]} \ket{\psi}\Big) &\geq\,1-\delta\;, \\
		\E_{\substack{H\\ u,v\sim H}}	\sum_{a,b\in\Sigma}\, \frac{1}{2}\Big(\bra{\psi} P^{u,v}_{[(a,b) \mapsto b |b]} \otimes \tilde{A}^v_b\ket{\psi}+\bra{\psi} A^v_b \otimes \tilde{P}^{u,v}_{[(a,b) \mapsto b |b]} \ket{\psi}\Big) &\geq\,1-\delta\;,
		\end{align*}
where the expectation is over a uniformly random subcube $H = H_{x_{m-j+2},\ldots,x_{m}}$ (sampled by choosing a random $j \sim [m]$ and random $x_{m-j+2},\ldots,x_{m} \sim [n]$) and uniformly random points $u,v \sim H$.
	\item \emph{(Synchronicity:)}
	\[
	\frac{1}{3} \Big( \E_{u\sim [n]^m} \sum_{a\in\Sigma}\,  \bra{\psi}A^u_a \otimes \tilde{A}^u_a\ket{\psi} + \E_{\ell \subset [n]^m}\sum_f \bra{\psi}B^\ell_f \otimes \tilde{B}^\ell_f\ket{\psi} + 
	\E_{\substack{H\\ u,v\sim H}}	\sum_{a,b\in\Sigma}\bra{\psi} P^{u,v}_{a,b}  \otimes \tilde{P}^{u,v}_{a,b} \ket{\psi}\Big) \,\geq\,1-\xi\;,\]
	where the expectations are uniform over a point $u$ and an axis-parallel line $\ell$, and as in the previous item for $(H,u,v)$. 
\end{enumerate}
Finally, we say that a two-prover strategy \emph{passes the (two-prover) tensor code test with probability $1 - \eps$} if the strategy is $(\eps',\delta',\xi')$-good for some $\eps',\delta',\xi'$ such that $\frac{1}{3}(\eps'+\delta'+\xi')\leq \eps$. 
\end{definition}

We can now state our two-prover extension of Theorem~\ref{thm:main}.  

\begin{theorem}
\label{thm:main-bipartite}
There exists a function
\[\eta'(m,t,k,\eps,n^{-1}) = \poly(m,t,k) \cdot \poly(\eps,n^{-1},e^{-k/m^2})\]
such that the following holds. 
Let $\strategy = (\ket{\psi},A,B,P,\tilde{A},\tilde{B},\tilde{P})$ be a two-prover strategy that passes the two-prover tensor code test for $\code^{\otimes m}$ with probability $1 - \eps$, where $\ket{\psi}\in \mH_A\otimes \mH_B$, and $k\geq 12mt$. Then there exists projective measurements $\{ G_c \}_{c \in \code^{\otimes m}}$ and $\{ \tilde{G}_c \}_{c \in \code^{\otimes m}}$ on $\mH_A$ and $\mH_B$ respectively such that the following hold, on expectation over a uniformly random $u \sim [n]^m$ and a uniformly random axis-parallel line $\ell\subset [n]^m$:  
\begin{align}
	 A^u_{a}\otimes \id &\simeq_{\eta'} \id\otimes \tilde{G}_{[c\mapsto c(u)\mid a]} \;,\label{eq:sync-1a}\\
	 G_{[c\mapsto c|_\ell\mid f]} \otimes \id &\simeq_{\eta'} \id\otimes \tilde{B}^\ell_{f}\;,\label{eq:sync-1b}\\
	 G_c\otimes \id &\simeq_{\eta'} \id\otimes \tilde{G}_c\;,\label{eq:sync-1c}
\end{align}
where $\eta' = \eta'(m,t,\eps,n^{-1})$ with $t = n - d + 1$.  
\end{theorem}

\begin{proof}
Let $\strategy = (\ket{\psi},A,B,P,\tilde{A},\tilde{B},\tilde{P})$ be as in the statement of the theorem. 
The proof proceeds in three steps. The first and main step consists in applying~\cite[Corollary 4.1]{vidick2021almost} to Theorem~\ref{thm:main} to obtain the following.

\begin{claim}\label{claim:sync-1}
There exists measurements $\{G_c\}$ on $\mH_A$ and $\{\tilde{G}_c\}$ on $\mH_B$ such that 
\begin{align}
	 A^u_{a}\otimes \id &\simeq_{\eta_1} \id\otimes \tilde{G}_{[c\mapsto c(u)\mid a]} \;,\label{eq:sync-1a1}\\
	\id\otimes \tilde{A}^u_{a} &\simeq_{\eta_1}  {G}_{[c\mapsto c(u)\mid a]} \otimes \id\;,\label{eq:sync-1a2}
	\end{align}
	for some $\eta_1 = \eta_1(m,t,\eps,n^{-1}) = \poly(m,t) \cdot \poly(\eps,n^{-1})$. 
\end{claim}

\begin{proof}
To apply~\cite[Corollary 4.1]{vidick2021almost} we first need to construct a \emph{symmetric} strategy, which means that $\ket{\psi}$ can be written as
\begin{equation}\label{eq:psi-sym}
 \ket{\psi} \,=\, \sum_i \sqrt{\lambda_i} \ket{u_i} \otimes \ket{u_i} \,\in\,\mH\otimes \mH\;,
\end{equation}
where the $\lambda_i$ are non-negative and $\{\ket{u_i}\}$ orthonormal in some Hilbert space $\mH$, and the two provers' measurement operators are identical up to a transpose: $\tilde{A}^u_a = (A^u_a)^T$ for all $u,a$, where the transpose is taken with respect to the basis $\{\ket{u_i}\}$, and similarly for the $B$ and $P$ measurement operators. The idea to achieve this is standard. First assume without loss of generality that $\mH_A\simeq\mH_B\simeq\mH'$ for some Hilbert space $\mH'$. This can always be achieved by enlarging $\mH_A$ or $\mH_B$ as needed. Next define 
\[ \ket{\psi'} = \frac{1}{\sqrt{2}} \big( \ket{\psi} \otimes \ket{0}\otimes \ket{1} + (S\ket{\psi})\otimes \ket{1}\otimes \ket{0} \big) \in \mH'\otimes \mH' \otimes \C^2 \otimes \C^2\;,\]
where $S$ is the ``SWAP'' operator $S=(2P-I)$ with $P$ the projection on the symmetric subspace Span$\{\ket{u}\otimes\ket{u}: \, \ket{u}\in \mH''\}$. Let $\mH = \mH'\otimes \C^2$. For any $u,a$ define $(A')^{u}_a =A^u_a\otimes \proj{0} + (\tilde{A}^u_a)^T \otimes \proj{1}$ and similarly $B'$ and $P'$. Then (reordering the subsystems of $\ket{\psi'}$ in the obvious way so that $\ket{\psi'} \in \mH\otimes \mH$) for any $X',Y'$ among the newly defined measurement operators we have that
\begin{align*}
 \bra{\psi'} X'\otimes (Y')^T \ket{\psi'} &= \frac{1}{2} \big(\bra{\psi} X \otimes \tilde{Y} \ket{\psi} + \bra{\psi}S^\dagger ( \tilde{X}^T \otimes Y^T ) S\ket{\psi}\big)\\
&= \bra{\psi} X \otimes \tilde{Y} \ket{\psi}\;. 
\end{align*}
As a consequence, the symmetric strategy $\strategy'$ also passes the tensor code test with probability at least $1-\eps$. 

Next we introduce some notation associated to the application of the corollary. 
Let $\mathcal{X} = [n]^m \cup ([n]^{m-1}\times [m])\cup ([n]^m)^2$ be the set of all questions that can be asked to the first prover in the two-prover tensor code test, $\mathcal{Y}=\{\bot\}$ a singleton, and $\mathcal{B}=\code^{\otimes m}$. Let $p$ be the uniform distribution on $[n]^m\subset \mathcal{X}$, and for $u$ in the support of $p$, let $g_u:\Sigma\to 2^{\code^{\otimes m}}$ map $a\in \Sigma$ to the set of all $c\in\code^{\otimes m}$ such that  $c(u)=a$. Finally let $\kappa:[0,1]\to[0,1]$ be the function $\kappa(\omega) = 1-\eta'(\omega)$, where $\eta'$ is defined as follows. 
 Let $\eta(m,t,\eps,n^{-1})$ be the function from Theorem~\ref{thm:main}. Let $m,t$ and $n$ be the parameters of the two-prover tensor code test, which we consider to be fixed. If $\eta$ is a concave function of $\eps$ we can let $\eta'(\omega)=\eta(m,t,\omega,n^{-1})$, with $m,t$ and $n$ fixed by the parameters for If $\eta$ is not concave we proceed as follows. Since $\eta$ has polynomial dependence on $\eps$ we can write it as a finite sum $\eta(\eps)=\sum_{i=1}^M \alpha_i \eps^{c_i}$ where $\alpha_i=\alpha_i(m,t,n^{-1})\in\R$ and $c_i\in \R$ are an increasing sequence of positive constants depending on the form of $\eta$ only. Moreover, since $\eta(m,t,\eps,n^{-1}) = \poly(m,t) \cdot  \poly(\eps,n^{-1})$  then for all $i$,  $|\alpha_i| = \poly(m,t)\cdot \poly(n^{-1})$. Now for $\eps\in [0,1]$ we have that $\eta(\eps)\leq \eta''(\eps)= \alpha \eps^{c_1}$ where $\alpha = \sum_i |\alpha_i|$ and so $\alpha =\poly(m,t)\cdot\poly(n^{-1})$. We then define $\kappa=1-\eta''$.

With this setting, Theorem~\ref{thm:main} shows that the assumptions of~\cite[Corollary 4.1]{vidick2021almost} are satisfied, for the game $\mathfrak{G}=(\mathcal{X},\mA,\nu,D)$ which is specified by the two-prover tensor code test. This is because (using terminology from~\cite{vidick2021almost})  for any symmetric PME strategy $(\ket{\psi},A)$ there is a tracial strategy $(\tau,A)$ for $\mathfrak{G}$
 that reproduces exactly the same distribution on answers in the game; for this it suffices to define the trace as $\tau(X)=\frac{1}{d}\tr(X)$ for any operator $X$ on the $d$-dimensional space $\mH$ that underlies the strategy $A$. 

We apply the conclusion of~\cite[Corollary 4.1]{vidick2021almost} to the symmetric projective strategy $\strategy'$. 
Since $\strategy$ and hence $\strategy'$ must be $(3\eps,\frac{3}{2}\eps,3\eps)$-good according to Definition~\ref{def:good-2}, we deduce that in the terminology from~\cite{vidick2021almost}, $\delta_{sync}(\strategy';\nu)\leq 3\eps$. As a consequence we obtain the existence of a measurement $\{G'_c\}$ on $\mH$ such that 
\begin{equation}\label{eq:sync-2}
	\E_{u \sim [n]^m} \sum_{c \in \code^{\otimes m}} \bra{\psi'}  (A')^u_{c(u)} \otimes G'_c \ket{\psi'} \,\geq\, 1 - \eta''(m,t,\poly(\eps),n^{-1})\;,
\end{equation}
where 
\[\eta''(m,t,\poly(\eps),n^{-1}) = \eta'(m,t,\eps+\poly(\eps),n^{-1})+\poly(\eps) =\poly(m,t) \cdot  \poly(\eps,n^{-1})\;.\]
It remains to derive~\eqref{eq:sync-1a1} and~\eqref{eq:sync-1a2}. Let $G_c = (I_A\otimes \bra{0})G'_c(I_A\otimes \ket{0})$ and $\tilde{G}_c = (I_B\otimes \bra{1})G'_c(I_B\otimes \ket{1})$, where $I_A$ (resp.\ $I_B$) denotes the projection on $\mH_A\subseteq \mH'$ (resp.\ $\mH_B\subset \mH'$). Then $\{G_c\}$ and $\{\tilde{G}_c\}$ are measurements on $\mH_A$ and $\mH_B$ respectively, and~\eqref{eq:sync-1a1} and~\eqref{eq:sync-1a2} follow directly from~\eqref{eq:sync-2}, with $\eta_1=2\eta''$, by expanding $\ket{\psi'}$ and $(A')^u_{c(u)}$ on the left-hand side using their definition. 
\end{proof}

Claim~\ref{claim:sync-1} almost establishes~\eqref{eq:sync-1a}, except that $\{\tilde{G}_c\}$ is not necessarily a projective measurement. So in the second step we turn $\{G_c\}$ and $\{G'_c\}$ into projective measurements. 

\begin{claim}\label{claim:sync-1}
There exists a projective measurement $\{G'_c\}$ on $\mH_A$ (resp. $\{\tilde{G}'_c\}$ on $\mH_B$) such that 
\begin{align}
A^u_{a}\otimes \id &\simeq_{\eta_2} \id\otimes \tilde{G}'_{[c\mapsto c(u)\mid a]} \;,\label{eq:sync-2a}\\
\id\otimes \tilde{A}^u_a &\simeq_{\eta_2} {G}'_{[c\mapsto c(u)\mid a]} \otimes \id\;,\label{eq:sync-2b}
\end{align}
	for some $\eta_2 = \eta_2(m,t,\eps,n^{-1}) = \poly(m,t) \cdot \poly(\eps,n^{-1})$. 
\end{claim}

\begin{proof}
We do the proof for $\{G_c\}$, as the proof for $\{\tilde{G}_c\}$ is entirely analogous. 
Write the Schmidt decomposition $\ket{\psi} = \sum_i \sqrt{\lambda_i} \ket{u_i} \ket{v_i}$, where the $\lambda_i$ are non-negative reals and $\{\ket{u_i}\}$ (resp.\ $\{\ket{v_i}\}$) an orthonormal basis of $\mH_A$ (resp. $\mH_B$). Let $K = \sum_i \sqrt{\lambda_i} \ket{u_i}\bra{v_i}$. Let $G^u_a = \sum_{c:\,c(u)=a} G_c$. Then 
\begin{align}
1-\eta_1 &\leq \E_u \sum_c \bra{\psi} G_c \otimes \tilde{A}^u_{c(u)} \ket{\psi}\notag \\
 &= \sum_a \bra{\psi} G^u_a \otimes \tilde{A}^u_a \ket{\psi}\notag\\
&= \sum_a \tr\big( G^u_a K \tilde{A}^u_a K^T \big)\notag\\
&\leq  \Big(\sum_a\tr\big( G^u_a \rho_A^{1/2} G^u_a \rho_A^{1/2}\big)\Big)^{1/2}\Big(\sum_a\tr\big( \tilde{A}^u_a \rho_B^{1/2} \tilde{A}^u_a \rho_B^{1/2}\big)\Big)^{1/2}\;,\label{eq:sync-2a1}
\end{align}
where the first inequality is~\eqref{eq:sync-1a1} and the last inequality follows by Cauchy-Schwarz, and $\rho_A = KK^\dagger$ and $\rho_B = K^\dagger K$. 
It follows that
\begin{align}
\sum_c \tr\big(G_c \rho_A^{1/2} G_c \rho_A^{1/2}\big)
&= \E_{u} \sum_a \tr\big( G^u_a \rho_A^{1/2} G^u_a \rho_A^{1/2}\big) - \sum_{c\neq c'} \E_u \sum_a 1_{c(u)=c'(u)=a} \tr\big(G_c \rho_A^{1/2} G_{c'} \rho_A^{1/2}\big)\notag\\
&\geq (1-\eta_1)^2 - \gamma_m\,\sum_{c\neq c'} \tr\big(G_c \rho_A^{1/2} G_{c'} \rho_A^{1/2}\big)\notag\\
&\geq 1-2\eta_1 - \gamma_m\;,\label{eq:sync-2a3}
\end{align}
where the second line follows from~\eqref{eq:sync-2a1} for the first term and Proposition~\ref{prop:distance} and $ \tr(G_c \rho_A^{1/2} G_{c'} \rho_A^{1/2})\geq 0$ for all $c,c'$ for the second. Using the Cauchy-Schwarz inequality, 
\[ 1-2\eta_1 - \gamma_m \,\leq\,\sum_c \tr\big(G_c \rho_A^{1/2} G_c \rho_A^{1/2}\big) \,\leq\, \sum_c \tr\big(G_c^2 \rho_A\big)\;.\]
We apply \Cref{lem:projectivization} to the measurement $\{G_c\}$, with $\tau$ defined by $\tau(X)=\tr(X\rho_A)$. From the lemma it follows that there exists a projective measurement $\{G'_c\}$ such that 
\begin{equation}\label{eq:g-gp}
 \sum_c \tr\big( (G_c-G'_c)^2 \rho_A \big) \,=\, O\big( 2\eta_1 + \gamma_m\big)\;.
\end{equation}
Then 
\begin{align*}
\E_u &\sum_c \bra{\psi} G'_c \otimes \tilde{A}^u_{c(u)} \ket{\psi}\\
&= \E_u \sum_c \bra{\psi} G_c \otimes \tilde{A}^u_{c(u)} \ket{\psi} + \E_u \sum_c \bra{\psi} (G'_c-G_c) \otimes \tilde{A}^u_{c(u)} \ket{\psi}\\
&\geq 1-\eta_1 - \Big(\E_u \sum_a \tr\big( ((G')^u_a - G^u_a) \rho_A^{1/2}((G')^u_a - G^u_a)\rho_A^{1/2}\big)\Big)^{1/2}\Big(\E_u \sum_a \tr\big( \tilde{A}^u_a \rho_B^{1/2} \tilde{A}^u_a\rho_B^{1/2}\big)\Big)^{1/2}\\
&\geq 1-\eta_1 - O\big( \eta_1 + \gamma_m\big)^{1/4}\;,
\end{align*}
where the first inequality uses~\eqref{eq:sync-1a2} for the first term and Cauchy-Schwarz for the second, and the last is justified by the following:
\begin{align*}
\E_u \sum_a &\tr\big( ((G')^u_a - G^u_a) \rho_A^{1/2}((G')^u_a - G^u_a)\rho_A^{1/2}\big)\\
&\leq 2 - 2\E_u \sum_a \tr\big(( G')^u_a \rho_A^{1/2}G^u_a\rho_A^{1/2}\big)\\
&\leq 2-2\sum_c \tr\big( G'_c \rho_A^{1/2} G_c \rho_A^{1/2} \big) + \gamma_m\\
&= 2-2\sum_c \tr\big( (G'_c-G_c) \rho_A^{1/2} G_c \rho_A^{1/2} \big) +\sum_c \tr\big( G_c \rho_A^{1/2} G_c \rho_A^{1/2} \big)+ \gamma_m\\
&\leq 4\eta_1 + 3\gamma_m + \Big(\sum_c \tr\big( (G'_c-G_c)^2 \rho_A\big)\Big)^{1/2}\Big(\sum_c \tr\big(G_c^2 \rho_A\big)\Big)^{1/2}\\
&\leq 4\eta_1 + 3\gamma_m + O\big(\sqrt{2\eta_1+\gamma_m}\big)\;,
\end{align*}
where the second inequality uses that $\tr\big( G'_c \rho_A^{1/2} G_{c'} \rho_A^{1/2} \big)\geq 0$ for all $c,c'$ and Proposition~\ref{prop:distance}, the fourth line uses~\eqref{eq:sync-2a3}, and the last uses~\eqref{eq:g-gp}. The claimed bound on $\eta_2$ follows from the same bound on $\eta_1$ and the bound on $\gamma_m$ provided in Proposition~\ref{prop:distance}. 
\end{proof}

With~\eqref{eq:sync-2a} we have established~\eqref{eq:sync-2a}. 
In the last step we obtain~\eqref{eq:sync-1b} and~\eqref{eq:sync-1c}. 

\begin{claim}
Let $\{G'_c\}$ and $\{\tilde{G}'_c\}$ be projective measurements that satisfy~\eqref{eq:sync-2a} and~\eqref{eq:sync-2b} respectively. Then
\begin{align}
	 G'_{[c\mapsto c|_\ell\mid f]} \otimes \id &\simeq_{\eta_3} \id\otimes B^\ell_{f}\;,\label{eq:sync-1b2}\\
	 G'_c\otimes \id &\simeq_{\eta_3} \id\otimes \tilde{G}'_c\;,\label{eq:sync-1c2}
\end{align}
for some $\eta_3 = \poly(\eta_2,(mt)/n,\eps)$. 
\end{claim}

\begin{proof}
We start with~\eqref{eq:sync-1c2}. Write
\begin{align}
\sum_c G'_c\otimes\tilde{G}_c &\simeq_{\gamma_m} \E_u \sum_a {G}'_{[c\mapsto c(u)\mid a]} \otimes \tilde{G}'_{[c\mapsto c(u)\mid a]} \label{eq:p1-1a} \\
&\simeq_{\eta'_3} \E_u \sum_a A^u_a \otimes \tilde{A}^u_a \label{eq:p1-1b}\\
&\simeq_{3\xi} \id\otimes \id\;,\label{eq:p1-1c}
\end{align}
where $\eta'_3=2\sqrt{2\eta_2}$ and recall that $\xi=3\eps$. Each of the three approximations above is justified as follows. For~\eqref{eq:p1-1a}, form the difference 
\begin{align*}
 \sum_{c,c'} \E_u 1_{c(u)=c'(u)} G'_c\otimes\tilde{G}_{c'} 
&\leq \gamma_m \id\otimes \id\;,
\end{align*}
using Proposition~\ref{prop:distance}. For~\eqref{eq:p1-1b}, write
\begin{align*}
 \E_u & \sum_a {G}'_{[c\mapsto c(u)\mid a]} \otimes \tilde{G}'_{[c\mapsto c(u)\mid a]} -  \E_u \sum_a A^u_a \otimes \tilde{A}^u_a \\
&=  \E_u \sum_a \big({G}'_{[c\mapsto c(u)\mid a]} \otimes \id - \id\otimes \tilde{A}^u_a\big)\big(\id \otimes  \tilde{G}'_{[c\mapsto c(u)\mid a]}\big) + \E_u\sum_a \big(\id\otimes \tilde{A}^u_a\big) \big(\id\otimes \tilde{G}'_{[c\mapsto c(u)\mid a]} -{A}^u_a\otimes \id \big)\;.
\end{align*}
Each of the two terms above is bounded in the same way. We do it for the first:
\begin{align*}
 \E_u &\sum_a \bra{\psi}\big({G}'_{[c\mapsto c(u)\mid a]} \otimes \id - \id\otimes \tilde{A}^u_a\big)\big(\id \otimes  \tilde{G}'_{[c\mapsto c(u)\mid a]}\big)\ket{\psi}\\
&\leq \Big( \E_u \sum_a \bra{\psi}\big({G}'_{[c\mapsto c(u)\mid a]}\otimes \id -\id \otimes \tilde{A}^u_a\big)^2\ket{\psi} \Big)^{1/2} \Big( \E_u\sum_a \bra{\psi}\big(\id \otimes  \tilde{G}'_{[c\mapsto c(u)\mid a]}\big)^2\ket{\psi}\Big)^{1/2}\\
&\leq \sqrt{2\eta_2}\;,
\end{align*}
by~\eqref{eq:sync-2b} and the fact that $\{G'\}$ and $\{\tilde{A}\}$ are families of projective measurements. Finally,~\eqref{eq:p1-1c} follows directly from the synchronicity condition. 

To conclude we show~\eqref{eq:sync-1b2}. 
\begin{align}
\E_\ell  \sum_f {G}'_{[c\mapsto c|_\ell\mid f]} \otimes \tilde{B}^\ell_f
&\simeq_{\gamma_m}  \E_\ell\E_{u\in \ell}  \sum_a {G}'_{[c\mapsto c(u)\mid a]} \otimes \tilde{B}^\ell_{f\mapsto f(u)\mid a}\label{eq:p1-2a}\\
&\simeq_{\eta''_3} \E_\ell\E_{u\in \ell}  \sum_a A^u_a \otimes \tilde{G}'_{[c\mapsto c(u)\mid a]} \label{eq:p1-2b}\\
&\simeq_{\eta_2} \id\otimes \id\;,\label{eq:p1-2c}
\end{align}
where $\eta''_3 = \sqrt{2\eta'_3}+\sqrt{4\eps}$. 
We justify each of these three approximations as follows. For~\eqref{eq:p1-2a}, form the difference
\begin{equation*}
\E_\ell \sum_{f,f'} \E_{u\in \ell} 1_{f(u)=f'(u)} G'_{[c\mapsto c|_\ell\mid f']} \otimes B^\ell_f \,\leq\, \gamma_1\,\id\otimes \id\,\leq\,\gamma_m\,\id\otimes \id\;,
\end{equation*}
by Proposition~\ref{prop:distance}. For~\eqref{eq:p1-2b}, write 
\begin{align*}
 \E_\ell&\E_{u\in \ell}  \sum_a {G}'_{[c\mapsto c(u)\mid a]} \otimes \tilde{B}^\ell_{f\mapsto f(u)\mid a} - \E_\ell\E_{u\in \ell}  \sum_a A^u_a \otimes \tilde{G}'_{[c\mapsto c(u)\mid a]} \\
&= \E_\ell\E_{u\in \ell}  \sum_a  \bra{\psi} \big({G}'_{[c\mapsto c(u)\mid a]} \otimes \id - \id \otimes \tilde{G}'_{[c\mapsto c(u)\mid a]} \big)\big(\id \otimes \tilde{B}^\ell_{f\mapsto f(u)\mid a}\big)\ket{\psi}\\
&\qquad + \E_\ell\E_{u\in \ell}  \sum_a \bra{\psi} \big(\id \otimes \tilde{G}'_{[c\mapsto c(u)\mid a]}\big)\big( \id\otimes \tilde{B}^\ell_{f\mapsto f(u)\mid a}- A^u_a \otimes \id\big)\ket{\psi}\;.
\end{align*}
Each of the two terms above is bounded as previously, using~\eqref{eq:p1-1a} for the first and the consistency condition between points and lines for the second. This gives $\eta''_3 = \sqrt{2(\gamma_m + \eta'_3 + 3\xi)} + \sqrt{4\eps}$. Finally,~\eqref{eq:p1-2c} is by~\eqref{eq:sync-2a}. The claimed bound on $\eta_3$ follows from Proposition~\ref{prop:distance} to bound $\gamma_m$ and $\xi\leq 3\eps$. 
\end{proof}

The theorem follows using the bound on $\eta_2$ given in Claim~\ref{claim:sync-1}.
\end{proof}

\section{Self-improvement}
\label{sec:self-improvement}

In this section we prove \Cref{lem:self-improvement}. The proof of the lemma relies on three main ingredients.

\paragraph{Orthogonalization.} The first is a lemma that allows us to round a submeasurement that is self-consistent to a projective measurement. The lemma first appears in~\cite{kempe2011parallel} for the finite-dimensional case, in~\cite{ji2020quantum} with an improved error dependence and was extended to the infinite-dimensional case in~\cite{de2021orthogonalization}. The proof uses the polar decomposition. 

\begin{lemma}[Orthogonalization]
\label{lem:projectivization}
Let $\algebra$ be a von Neumann algebra with a normal state $\tau$.\footnote{Even though in our application $\tau$ will be tracial, this is not required for the lemma.}
Let $A = \{A_a\}_{a\in \mA} \subset \algebra$ denote a submeasurement. Suppose that
\[
	\sum_a \tau(A_a(\Id - A_a)) \leq \zeta.
\]
Then there exists a projective submeasurement $\{P_a\} \subset \algebra$ such that
\[
	A_a \approx_{\sqrt{18 \zeta}} P_a \;.
\]
\end{lemma}

\begin{proof}
We indicate how the lemma follows from~\cite[Theorem 1.2]{de2021orthogonalization}. Let $n=|\mA|+1$ and identify $\mA$ with the set $\{1,\ldots,n-1\}$. Let $A_{n}=\Id-\sum_a A_a$ so $\{A_a\}_{a\in \{1,\ldots,n\}}$ is a measurement. Applying~\cite[Theorem 1.2]{de2021orthogonalization}, we get orthogonal projections $P_a$ satisfying
\begin{equation}\label{eq:proj-2}
A_a \approx_{\sqrt{9\delta}} P_a\;,
\end{equation}
 where 
\begin{align}
\delta &= 1- \sum_{a=1}^n \tau\big( A_a^2 \big)\notag\\
&= 	\sum_{a=1}^{n-1} \tau(A_a(\Id - A_a)) + \tau\big(A_{n}- A_{n}^2\big)\;,\label{eq:proj-1}
\end{align}
where for the second line we used that $\{A_a\}_{a\in \{1,\ldots,n\}}$ is a measurement. The first term is at most $\zeta$ by assumption. For the second term we write 
\begin{align*}
\tau\big(A_n^2\big) &=\tau\Bigg( A_n\bigg( \Id-\sum_{a=1}^{n-1} A_a\bigg)\Bigg)\\
&= \tau\big(A_n) - \sum_{a=1}^{n-1}\tau\big( A_n A_a \big)\\
&\geq \tau\big(A_n) - \sum_{a=1}^{n-1}\tau\big(  \big(\Id-A_a\big)A_a\big)\\
&\geq \tau\big(A_n\big) - \zeta\;,
\end{align*}
where the inequality on the third line uses that 
\[ A_n \,=\, \Id- \sum_{a=1}^{n-1} A_a  \,\leq\, \Id-A_a\]
 for all $a\in\{1,\ldots,n-1\}$, and $A_a \geq 0$ for all $a$.  
Thus~\eqref{eq:proj-1} is at most $2\zeta$, which together with~\eqref{eq:proj-2} proves the lemma.
\end{proof}

\paragraph{Complementary slackness.} The second ingredient is a  duality result about positive linear maps which appears in~\cite{de2021orthogonalization}. The proof is based on an application of the Hahn-Banach theorem.

\begin{lemma}[Proposition 7.1 in~\cite{de2021orthogonalization}]
\label{lem:duality}
Let $\algebra$ be a von Neumann algebra
	and let $\varphi_1,\ldots,\varphi_k: \algebra \to \C$ be normal positive linear maps. Then there is a unique normal linear map $\psi : \algebra \to \C$ of minimal norm  such that $\psi \geq \varphi_i$ for all $i$. Moreover there exists a measurement $\{ T_i \}_{i \in [k]}$ such that for all $X \in \algebra$, 
	\begin{equation}
		\psi(X) = \sum_i \varphi_i(T_i X)
	\end{equation}
\end{lemma}

In finite dimensions this lemma corresponds to the complementary slackness condition of a particular semidefinite program satisfying strong duality (see~\cite[Lemma 10]{natarajan2018two} for the finite-dimensional version of this statement). 
In the proof of \Cref{lem:self-improvement}, we use \Cref{lem:duality} to define the ``maximum'' $\psi$ of a collection of linear forms $\varphi_i$. This step can be thought of as a replacement for ``majority decoding'' in analyses of low-degree testing in the commutative (i.e.\ classical) case~\cite{gemmell1991self,rubinfeld1992self}.

\paragraph{Local-to-global variance bound.}  
The third ingredient is the following lemma. Recall that $\code$ is a $[n,k,d]_\Sigma$ code. Define $\gamma = 1 - \frac{d}{n}$. 

\begin{lemma}\label{lem:variance}
Let $(\tau,A,B,P)$ be an $(\eps,\delta)$-good strategy for the tensor code test with $\code^{\otimes m}$. 
  Let $\{T_{g}\}$ be a measurement with outcomes in $\code^{\otimes m}$. Then on average over uniform $x,y\in [n]^{m}$,
  \begin{equation*}
	A^{x}_{g(x)} T^{1/2}_{g} \approx_{\zeta_{var}} A^{y}_{g(y)} T^{1/2}_{g}
  \end{equation*}
  where $\zeta_{var} = \sqrt{m} \, \Big(2\sqrt{2\eps} + 2\sqrt{\gamma} \Big)$. 
\end{lemma}

Before proving \Cref{lem:variance} we first prove a ``local'' version of \Cref{lem:variance}, where the closeness of $A^{x}_{g(x)} T^{1/2}_{g}$ and $A^{y}_{g(y)} T^{1/2}_{g}$ hold when $x,y$ are random points from a random axis-parallel line (rather than being random points from the entire domain $[n]^m$).

\begin{lemma}\label{lem:local-variance}
Let $(\tau,A,B,P)$ be an $(\eps,\delta)$-good strategy for the tensor code test with $\code^{\otimes m}$. 
  Then on overage over axis-parallel lines $\ell$, and $x,y \in \ell$,
  \begin{equation*}
	A^{x}_{g(x)} T^{1/2}_{g} \approx_{\zeta_{local}} A^{y}_{g(y)} T^{1/2}_{g}
  \end{equation*}
 where $\zeta_{local} = 2\sqrt{2\eps} + 2\sqrt{\gamma}$.
\end{lemma}
\begin{proof}
  By definition of $(\eps,\delta)$-good strategy we have $A^{x}_{a} \simeq_{\eps} B^{\ell}_{[f \mapsto f(x) \mid a]}$, and thus by \Cref{lem:consistency-consequences} we have $A^{x}_{a} \approx_{\sqrt{2\eps}} B^{\ell}_{[f \mapsto f(x) \mid a]}$. This implies that 
 \begin{equation}
 \label{eq:local-variance-1}
  A^x_{g(x)} T_g^{1/2} \approx_{\sqrt{2\eps}} B^{\ell}_{[f : f(x) = g(x)]} \cdot T_g^{1/2}
	\end{equation}
  on average over a uniformly random axis-parallel line $\ell$ and a uniformly random point $x \in \ell$. This holds because 
  \begin{gather*}
  	\E_{\ell, x \in \ell} \sum_{g \in \code^{\otimes m}} \left \| \Big( A^x_{g(x)} - B^{\ell}_{[f : f(x) = g(x)]} \Big) T_g^{1/2} \right \|^2_\tau 
	= \E_{\ell, x \in \ell} \sum_{a \in \Sigma} \sum_{g:g(x) = a} \left \| \Big( A^x_{g(x)} - B^{\ell}_{[f \mapsto f(x) \mid a]} \Big) T_g^{1/2} \right \|^2_\tau  \\
	= \E_{\ell, x \in \ell} \sum_{a \in \Sigma} \sum_{g:g(x) = a} \tau \left ( \Big( A^x_a - B^{\ell}_{[f \mapsto f(x) \mid a]} \Big)^2 T_g  \right) 
	\leq \E_{\ell, x \in \ell} \sum_{a \in \Sigma}\tau \Big(\Big( A^x_a - B^{\ell}_{[f \mapsto f(x) \mid a]} \Big)^2\Big) \leq 2\eps
  \end{gather*}
  where the first inequality follows from the fact that $\{T_g\}$ is a measurement. Next, we have that
 \begin{equation}
 \label{eq:local-variance-2}
  B^{\ell}_{[f : f(x) = g(x)]} \cdot T_g^{1/2} \approx_{\sqrt{\gamma}} B^{\ell}_{g|_\ell} \cdot T_g^{1/2}~.
	\end{equation}  
	This holds because
	\begin{align*}
	&\E_{\ell, x \in \ell} \sum_{g \in \code^{\otimes m}} \left \| \Big( B^{\ell}_{[f : f(x) = g(x)]} - B^{\ell}_{g|_\ell} \Big) T_g^{1/2} \right \|^2_\tau  \\
	&= \E_{\ell, x \in \ell} \sum_{h \in \code} \tau \left ( \Big( B^{\ell}_{[f : f(x) = h(x)]} - B^{\ell}_h \Big)^2 T_{[g:g|_\ell = h]}  \right) \\
	&= \E_{\ell, x \in \ell} \sum_{h \in \code} \tau \left ( B^{\ell}_{[f : f(x) = h(x), f \neq h]} \cdot T_{[g:g|_\ell = h]}  \right) \\
	&= \E_{\ell} \sum_{h,f \in \code} \Big(\E_{x \in \ell} \indicator[f(x) = h(x), f \neq h] \Big) \cdot \tau \left ( B^{\ell}_{f} \cdot T_{[g:g|_\ell = h]}  \right) \\
	&\leq \E_{\ell} \sum_{h,f \in \code} \gamma \cdot \tau \left ( B^{\ell}_{f} \cdot T_{[g:g|_\ell = h]}  \right) \\
	&= \gamma
	\end{align*}
	where the third line follows from the fact that the $\{B^\ell_h\}$ measurements are projective, in the fourth line we used $\indicator[f(x) = h(x), f \neq h]$ to denote the characteristic function of the indicated set, the fifth line follows from the fact that two distinct codewords $f \neq h$ can only agree on at most $\gamma$ fraction of coordinates, and the last line follows from the fact that $\{B^\ell_f\}$ and $\{T_g\}$ are measurements. Thus putting together~\eqref{eq:local-variance-1} and~\eqref{eq:local-variance-2} we get
  \begin{gather*}
	A^x_{g(x)} T_g^{1/2}
	 \approx_{\sqrt{2\eps}} B^{\ell}_{[f : f(x) = g(x)]} \cdot T_g^{1/2}
	  \approx_{\sqrt{\gamma}} B^{\ell}_{g|_\ell} \cdot T_g^{1/2} \\
	\approx_{\sqrt{\gamma}} B^{\ell}_{[f: f(y) = g(y)]} \cdot T_g^{1/2}
	\approx_{\sqrt{2\eps}} A^y_{g(y)} \cdot T_g^{1/2}.
  \end{gather*}
  on average over a random axis-parallel line $\ell$ and independent random points $x,y \in \ell$.
\end{proof}

The proof of \Cref{lem:variance} now follows by using the expansion of the hypercube to transfer the ``local variance'' bound of \Cref{lem:local-variance} to a ``global variance'' bound.
\begin{proof}[Proof of \Cref{lem:variance}]

  Define the graph $G$ to have vertex set $[n]^m$, and vertices
  $x,y \in [n]^m$ are connected if and only if they differ in at most one
  coordinate.
  Define $N = n^{m}$.
  The Laplacian $L$ associated with this graph $G$ has a zero eigenvector
  \[
	\ket{\varphi_0} = \frac{1}{\sqrt{N}} \sum_{x \in [n]^m} \ket{x}.
  \]
Furthermore the second largest eigenvalue $\lambda_2$ is equal to $\frac{1}{mN}$.

For all $g \in \code^{\otimes m}$, define the positive linear functional
$\psi_g(X) = \tau\big(T_g^{1/2} \cdot X \cdot T_g^{1/2}\big)$.
By \Cref{lem:local-to-global-expander}, we have that
\begin{align*}
  \E_{x,y \sim {[n]}^{m}} \sum_{g \in \code^{\otimes m}} \norm{
  \bigl(A^{x}_{g(x)} - A^{y}_{g(y)}\bigr) T_{g}^{1/2}}_{\tau}^{2}
  & = \E_{x,y \sim [n]^m} \sum_{g \in \code^{\otimes m}} \psi_g
	\Paren{\Big(A^x_{g(x)} - A^y_{g(y)} \Big)^2 }\\
  &\leq m \E_{\ell, x,y \in \ell} \sum_{g \in \code^{\otimes m}}
  \psi_g \Paren{\Big(A^x_{g(x)} - A^y_{g(y)} \Big)^2 } \\
  & = m\, \E_{(x,y) \sim G} \sum_{g \in \code^{\otimes m}} \norm{
	\bigl(A^{x}_{g(x)} - A^{y}_{g(y)}\bigr) T_{g}^{1/2}}_{\tau}^{2}\\
  &\leq m \, \zeta_{local}^{2}, & \text{\Cref{lem:local-variance}}
\end{align*}
which completes the proof by definition.
\end{proof}

We now proceed to prove \Cref{lem:self-improvement}.

\subsection{Proof of \Cref{lem:self-improvement}, Step 1: Construction of (non-projective) $H$}
\label{sec:si-step1}

Let $A = \{A^x\}_{x \in [n]^m}$ and $\{G_g\}$ be as in the statement of \Cref{lem:self-improvement}. Let $\algebra$ and $\tau$ be the underlying von Neumann algebra and normal trace. For all $g\in \code^{\otimes m}$ define 
\[
	A_g = \E_x A^x_{g(x)}\;,
\]
where the expectation is uniform over all $x\in [n]^m$. 
For all $g \in \code^{\otimes m}$ define $\varphi_g: \algebra \to \C$ by
\[
	\forall X \in \algebra\;,\qquad \varphi_g(X) = \tau(X \cdot A_g) \;.
\]
Notice that $\varphi_g$ is a normal positive linear map. It is positive because $\tau(X A_g) = \tau(A_g^{1/2} X A_g^{1/2}) \geq 0$ for positive $X \in \algebra$. It is normal because $\tau$ is normal, which by definition (see \Cref{sec:prelims}) implies that there is a positive trace class operator $M$ such that $\tau(X A_g) = \tau(A_g^{1/2} X A_g^{1/2}) = \tr(A_g^{1/2} X A_g^{1/2} M)$. Since $X A_g^{1/2} M$ is a trace class operator and $A_g^{1/2}$ is bounded, we have that $\tr(A_g^{1/2} X A_g^{1/2} M) = \tr(X A_g^{1/2} M A_g^{1/2})$; this implies that $\varphi_g(X) = \tr(X A_g^{1/2} M A_g^{1/2})$ is normal because $A_g^{1/2}  M A_g^{1/2}$ is a positive trace class operator. 

By \Cref{lem:duality} there exists a unique normal linear map $\psi$ of minimal norm satisfying $\psi \geq \varphi_g$ for all $g\in \code^{\otimes m}$. Moreover there exists a measurement $\{T_g\}$ such that 
\begin{equation}
\label{eq:complementary-slackness}
\forall X \in \algebra\;,\qquad \psi(X) = \sum_g \varphi_g(T_g X)\;.
\end{equation}
We refer to \Cref{eq:complementary-slackness} as the \emph{complementary slackness} condition.

For all $x \in [n]^m$ and $g\in \code^{\otimes m}$ let
\[
	H^x_{g} = A^x_{g(x)} \cdot T_g \cdot A^x_{g(x)}\;.
\]
We verify that for all $x \in [n]^m$, $\{H^x_{g}\}$ forms a submeasurement:
\begin{gather}
	\sum_g H^x_{g} = \sum_g A^x_{g(x)} \cdot T_g \cdot A^x_{g(x)} = \sum_a A^x_{a} \Big ( \sum_{g : g(x) = a} T_g \Big) A^x_{a} \leq \sum_a (A^x_{a})^2 = \sum_a A^x_{a} = \Id.
\end{gather}
Finally, for $g\in\code^{\otimes m}$ define 
\[
	H_g = \E_{x \sim [n]^m} H^x_{g}\;.
\]
We proceed to argue that, except for the projectivity condition, the submeasurement $\{H_g\}$ satisfies the conclusions of \Cref{lem:self-improvement}. We first show a technical lemma. 
\begin{lemma}\label{lem:self-improvement-helper}
  Let $\zeta_{var} = \sqrt{m} (2\sqrt{2\eps} + 2\sqrt{\gamma})$ as in
  \cref{lem:variance}.
  Then
  \[
	\E_x\sum_g \tau\big(\big| H_g - A^x_{g(x)} T_g A^x_{g(x)} \big|\big) \leq {2\, \zeta_{var}} \;.
  \]
\end{lemma}

\begin{proof}
We have
\begin{align}
\E_x\sum_g \tau\big(\big| H_g - A^x_{g(x)}  T_g A^x_{g(x)} \big|\big)
&=  \E_x\sum_g \tau\Big(\Big| \E_y \big(   A^y_{g(y)}- A^x_{g(x)} \big)T_g A^y_{g(y)} +\E_y A^x_{g(x)}T_g \big(   A^y_{g(y)}- A^x_{g(x)} \big)  \Big|\Big)\notag\\
&\leq \E_{x,y}\sum_g \tau\big(\big|  \big(   A^y_{g(y)}- A^x_{g(x)} \big)T_g A^y_{g(y)}\big|\big) + \tau\big(\big|  A^x_{g(x)}T_g \big(   A^y_{g(y)}- A^x_{g(x)} \big)\big|\big)\notag\\
&\leq \E_{x,y}\sum_g \Big(\big\|  \big(   A^y_{g(y)}- A^x_{g(x)} \big)T_g^{1/2}\big\|_\tau \cdot \big\| T_g^{1/2} A^y_{g(y)}\big\|_\tau \notag\\
&\hskip3cm+ \big\|  A^x_{g(x)}T_g^{1/2}\big\|_\tau \cdot \big\|T_g^{1/2}\big(   A^y_{g(y)}- A^x_{g(x)} \big)\big\|_\tau\Big)\notag\\
&\leq 2\Big(\E_{x,y}\sum_g \big\|  \big( A^y_{g(y)}- A^x_{g(x)} \big)T_g^{1/2}\big\|_\tau^2\Big)^{1/2}\Big(\E_{x}\sum_g \big\| T_g^{1/2}A^x_{g(x)}\big\|_\tau^2\Big)^{1/2}\;,\notag
\end{align}
where the second line uses the triangle inequality for the $1$-norm
(\cref{item:triangle} of \cref{prop:holder}), the third the H\"{o}lder
inequality (\cref{item:holder-1} of \cref{prop:holder}), and the fourth the
Cauchy-Schwarz inequality.
The lemma then follows from Lemma~\ref{lem:variance} together with the fact that
$\|A^x_a\|\leq 1$ for every $x,a$ and $\sum_g T_g\leq \Id$.
\end{proof}

\subsubsection{Bounding the completeness of $H$}

\begin{lemma}
$\sum_g \tau(H_g) \geq 1-\nu$.
\end{lemma}

\begin{proof}
We compute
	\begin{align}
		\sum_g \tau(H_g) &= \E_x \sum_g \tau \Paren{  A^x_{g(x)} \cdot T_g \cdot A^x_{g(x)} } \label{eq:si-hcomp1a}\\
						&= \E_x \sum_g \tau \Paren{  T_g \cdot A^x_{g(x)} } & \text{Cyclicity of trace, projectivity of $A^x_{a}$} \notag\\
						&= \sum_g \tau \Paren{ T_g \cdot A_{g} } \notag\\
						&= \sum_g \varphi_g (T_g) & \text{Definition of $\varphi_g$} \notag\\
						&= \psi(\Id) & \text{Complementary slackness} \label{eq:si-hcomp1b}\\
						&= \sum_g \psi(G_g) & \text{$\{G_g\}$ is a measurement} \notag\\
						&\geq \sum_g \varphi_g(G_g) & \text{$\psi \geq \varphi_g$ for all $g$} \notag\\
						&= \sum_g \tau(G_g \cdot A_g) & \text{Definition of $\varphi_g$}\notag \\
						&= \E_x \sum_a \tau(G_{[g(x) = a]} \cdot A^x_{a}) \notag\\
						&= \E_x \sum_a \tau(G_{[g(x) = a]}) - \E_x \sum_{b \neq a} \tau(G_{[g(x) = a]} \cdot A^x_b) \notag\\
						&\geq 1 - \nu\;,\notag
	\end{align}
	by the assumption that $\{G_g\}$ is $\nu$-inconsistent with $A$. 
\end{proof}

\subsubsection{Bounding the consistency of $H$ with $A$}

Let $\zeta_{var}$ be as in Lemma~\ref{lem:self-improvement-helper}.

\begin{lemma}\label{lem:si-cons}
On average over $x \sim [n]^m$, $H_{[g \mapsto g(x) \mid a]} \simeq_{2\, {\zeta_{var}}} A^x_a$.
\end{lemma}

\begin{proof}
By expanding the definition of $H$ we get
	\begin{align}
		\E_x \sum_{a \neq b} \tau \Paren{ A^x_{b} \cdot H_{[g \mapsto g(x) \mid a]} } &= \E_{x} \sum_{g,b: g(x) \neq b} \tau \Paren{ A^x_{b} H_g}\label{eq:si-cons-1a} \\
		&\approx_{2\, {\zeta_{var}}} \E_{x} \sum_{g,b: g(x) \neq b} \tau \Paren{ A^x_{b} \cdot A^x_{g(x)} \cdot T_g \cdot A^x_{g(x)} } & \text{\Cref{lem:self-improvement-helper}} \label{eq:si-cons-1b}\\
		&= 0 & \text{Projectivity of $A^x_{b}$}
	\end{align}
	Here the application of Lemma~\ref{lem:self-improvement-helper} is justified by forming the difference between~\eqref{eq:si-cons-1b} and~\eqref{eq:si-cons-1a} and bounding it as
	\begin{align*}
	\big|\eqref{eq:si-cons-1b}-\eqref{eq:si-cons-1a}\big|
	&\leq  \E_{x} \Big| \sum_{g,b: g(x) \neq b}  \tau \big( A^x_{b} \cdot (H_g - A^x_{g(x)} \cdot T_g \cdot A^x_{g(x)}\big)\big)\Big|\\
	&\leq  \E_{x} \sum_{g}  \big|\tau \big( (\Id-A^x_{g(x)}) \cdot (H_g - A^x_{g(x)} \cdot T_g \cdot A^x_{g(x)}\big)\big)\big|\\
	&\leq   \E_{x} \sum_{g}  \tau \Big( \big|H_g - A^x_{g(x)} \cdot T_g \cdot A^x_{g(x)}\big| \Big)\\
	&\leq {2 \,{\zeta_{var}}}\;,
	\end{align*}
	where the third line follows from H\"older's inequality (\cref{item:holder-2} of \cref{prop:holder}) and $\|\Id-A^x_{g(x)}\|\leq 1$.
\end{proof}

\subsubsection{Bounding the projectivity of $H$}
\label{sec:proof-of-self-improvement-projectivity}

\begin{lemma}\label{lem:si-proj}
Let $\gamma_m = 1 - (d/n)^m$. Then
$\sum_g \tau(H_g - H_g^2) \leq 5\,{\zeta_{var}} + \gamma_m$.
\end{lemma}

\begin{proof}
Write
\begin{align}
	\sum_g \tau(H_g^2) &= \E_{x} \sum_g \tau \Paren{ H^x_{g} \cdot H_g } \notag\\
	&\approx_{2 \,{\zeta_{var}}} \E_{x} \sum_{g} \tau \Paren{ A^x_{g(x)} \cdot H^x_{g} \cdot A^x_{g(x)} \cdot T_{g} } \label{eq:self-improvement-projectivity-1}\;,
\end{align}
where the approximation on the second line follows from \Cref{lem:self-improvement-helper}. Now we write 
\begin{align}
\text{RHS of \cref{eq:self-improvement-projectivity-1}} &\approx_{\gamma_m} \E_{x} \sum_{g,g'} \tau \Paren{ A^x_{g(x)} \cdot H^x_{g'} \cdot A^x_{g(x)} \cdot T_{g} } & \text{To be proved below} \label{eq:self-improvement-projectivity-2}\\
&= \E_{x} \sum_{g,g'} \tau \Paren{ H^x_{g'} \cdot A^x_{g(x)} \cdot T_{g} } & \text{$A^x_{g(x)} \cdot H^x_{g'} \cdot A^x_{g(x)} = H^x_{g'} \cdot A^x_{g(x)}$} \\
&\approx_{ \,{\zeta_{var}}} \E_{x,y} \sum_{g,g'} \tau \Paren{ H^x_{g'} \cdot A^y_{g(y)} \cdot T_{g} } & \text{To be proved below} \label{eq:self-improvement-projectivity-3} \\
&= \sum_{g,g'} \tau \Paren{ H_{g'} \cdot A_g \cdot T_{g} } \\
&= \sum_{g,g'} \varphi_g (T_g \cdot H_{g'}) & \text{Definition of $\varphi_g$} \\
&= \sum_{g'} \psi(H_{g'}) & \text{Complementary slackness} \\
&\geq \sum_{g'} \varphi_{g'} (H_{g'}) & \text{$\psi \geq \varphi_{g'}$ for all $g'$} \\
&= \sum_{g'} \tau(H_{g'} \cdot A_{g'}) & \text{Definition of $\varphi_{g'}$} \\
&= \E_x \sum_a \tau \Paren{ H_{[g \mapsto g(x) \mid a]} \cdot A^x_{a} } \\
&= \E_x \sum_a \tau \Big( H_{[g \mapsto g(x) \mid a]} \cdot \Big(\Id - \sum_{b \neq a} A^x_{b} \Big)\Big) \\
&= \tau (H) - \E_x \sum_{a \neq b} \tau \big( H_{[g \mapsto g(x) \mid a]} \cdot A^x_{b} \big)\\
&\geq \tau(H) - 2\,{\zeta_{var}} & \text{\cref{lem:si-cons}} \label{eq:si-p-1}
\end{align}
This completes the proof of the lemma, modulo Eq.~\eqref{eq:self-improvement-projectivity-2} and Eq.~\eqref{eq:self-improvement-projectivity-3}, which we now turn to.

\paragraph{Proof of approximation in \cref{eq:self-improvement-projectivity-2}.}
We aim to show that 
\[
	\Big | \E_x \sum_{g' \neq g}  \tau \Paren{ A^x_{g(x)} \cdot H^x_{g'} \cdot A^x_{g(x)} \cdot T_{g} } \Big | \leq  \gamma_m\;.
\]
To see this, we show
\begin{align}
	 &\E_x \sum_{g' \neq g}  \tau \Paren{ A^x_{g(x)} \cdot H^x_{g'} \cdot A^x_{g(x)} \cdot T_{g} } \notag \\
	 &= \E_{x} \sum_{g\neq g'} \tau \Paren{ A^x_{g(x)} \cdot A^x_{g'(x)} \cdot T_{g'} \cdot A^x_{g'(x)} \cdot A^x_{g(x)} \cdot T_{g} } \notag\\
	&= \E_{x} \sum_{g \neq g'} \tau \Paren{ A^x_{g(x)} \cdot T_{g'} \cdot A^x_{g(x)} \cdot T_{g} } \cdot \indicator[g(x) = g'(x)] \notag\\
	&= \E_{x} \sum_{g \neq g'} \tau \Paren{ T_{g'} \cdot A^x_{g(x)} \cdot T_{g} \cdot A^x_{g(x)} } \cdot \indicator[g(x) = g'(x)] & \text{Cyclicity of trace} \notag \\
	&= \E_{x} \sum_{g \neq g'} \tau \Paren{ T_{g'} \cdot H_g} \cdot \indicator[g(x) = g'(x)]~. & \text{Definition of $H_g$}  \label{eq:self-improvement-self-consistency-2}
\end{align}
To conclude, note that since $\tau(T_{g'}H_g) \geq 0$ for all $g,g'$ and $\sum_{g\neq g'} \tau(T_{g'}H_g)\leq 1$ it follows that~\eqref{eq:self-improvement-self-consistency-2} is at most 
\[\max_{g\neq g'}\; \E_{x} \,\indicator[g(x) = g'(x)]  \,\leq\, 1 - \frac{d^m}{n^m} = \gamma_{m}\;,\]
 because two distinct codewords $g \neq g'$ of $\code^{\otimes m}$ can only agree on at most $n^m - d^m$ coordinates.

\paragraph{Proof of approximation in \cref{eq:self-improvement-projectivity-3}.}
We aim to show that
\begin{equation}
\label{eq:self-improvement-projectivity-7}
\E_{x} \sum_{g,g'} \tau \Paren{ H^x_{g'} \cdot A^x_{g(x)} \cdot T_{g} }
 \approx_{ \,{\zeta_{var}}} \E_{x,y} \sum_{g,g'} \tau \Paren{ H^x_{g'} \cdot A^y_{g(y)} \cdot T_{g} } \;.
\end{equation}
To do this, we bound the magnitude of the difference:
\begin{align*}
& \Big | \E_{x,y} \sum_{g,g'} \tau \Paren{ H^x_{g'} \cdot (A^x_{g(x)} - A^y_{g(y)}) \cdot T_{g} } \Big | \\
&= \Big | \E_{x,y} \sum_{g} \tau \Big(T_g^{1/2} \cdot \Big ( \sum_{g'} H^x_{g'} \Big) \cdot (A^x_{g(x)} - A^y_{g(y)}) \cdot T_{g}^{1/2} \Big) \Big | \\
&\leq  \sqrt{\E_{x} \sum_{g} \tau \Big(T_g^{1/2} \cdot \Big ( \sum_{g'} H^x_{g'} \Big)^2 \cdot T_g^{1/2}\Big)} \cdot \sqrt{ \E_{x,y} \sum_g  \tau \Big( T_g^{1/2} \cdot \big(A^x_{g(x)} - A^y_{g(y)}\big)^2 \cdot T_{g}^{1/2} \Big) }  \\
&\leq \sqrt{1} \cdot \,{\zeta_{var}}
\end{align*}
where we used Cauchy-Schwarz, the fact that $\{H^x_{g'} \}$ and $\{T_g\}$ are (sub)measurements, and \Cref{lem:variance}.
This concludes the proof of~\eqref{eq:si-p-1} and hence the lemma. 
\end{proof}

\subsection{Proof of \Cref{lem:self-improvement}, Step 2: Making $H$ projective}

To conclude the proof of \Cref{lem:self-improvement}, we show how to modify the submeasurement $\{H_g\}$ defined and analyzed in the previous section  so that it is projective. Recall that according to Lemma~\ref{lem:si-proj} $\{H_g\}$ satisfies
\[
	\sum_g \tau(H_g - H_g^2) \,\leq\, \eta\;,
\]
for $\eta = 5\,{\zeta_{var}} + \gamma_m$. This condition allows us to apply the Orthogonalization Lemma (\Cref{lem:projectivization}). From the lemma it follows that there exists a projective submeasurement $H' = \{H_g'\}$ satisfying
\[
	H_g \approx_{\sqrt{18 \eta}} H_g'\;.
\]

\subsubsection{Proof of \Cref{enu:self-improvement-completeness} (Completeness)}

The completeness of the projective submeasurement $\{H_g'\}$ is related to the completeness of $\{H_g\}$ in the following manner:
 \begin{align}
 	\sum_g \tau(H_g') &=\sum_g \tau(H_g' \cdot H_g') & \text{$\{H_g'\}$ is projective}\notag \\
					  &\approx_{\sqrt{18 \eta}} \sum_g \tau(H_g \cdot H_g') & \text{\Cref{lem:closeness-to-close-ips}} \notag\\
					  &\approx_{\sqrt{18 \eta}} \sum_g \tau(H_g \cdot H_g) & \text{\Cref{lem:closeness-to-close-ips}} \notag\\
					  &\geq \sum_g \tau(H_g) - \eta & \text{Lemma~\ref{lem:si-proj}}\label{eq:hp-proj}
 \end{align}

\subsubsection{Proof of \Cref{enu:self-improvement-consistency} (Consistency with $A$)}

Next we show consistency of $\{H_g'\}$ with the measurements $\{A^x_a\}$. Towards this we first prove that for all $x$,
\[
	H_{[g \mapsto g(x) \mid a]} \approx_{\eta'} H_{[g \mapsto g(x) \mid a]}' \;,
\]
where
\[
  \eta' = \sqrt{2\sqrt{18\eta} + \eta}.
\]
This is because
\begin{align*}
	& \sum_a \tau \big ( (H_{[g \mapsto g(x) \mid a]} - H_{[g \mapsto g(x) \mid a]}')^2 \big)\\
  = & \sum_a \tau(H_{[g \mapsto g(x) \mid a]}^2) + \sum_a \tau((H_{[g \mapsto g(x) \mid a]}')^2) - 2 \sum_a \tau(H_{[g \mapsto g(x) \mid a]} \cdot H_{[g \mapsto g(x) \mid a]}') \\
	\leq & \sum_a \tau(H_{[g \mapsto g(x) \mid a]}) + \sum_a \tau((H_{[g \mapsto g(x) \mid a]}')) - 2 \sum_a \tau(H_{[g \mapsto g(x) \mid a]} \cdot H_{[g \mapsto g(x) \mid a]}') \\
	= & \sum_g \tau(H_g) + \sum_g \tau(H_g') - 2 \sum_a \tau(H_{[g \mapsto g(x) \mid a]} \cdot H_{[g \mapsto g(x) \mid a]}') \\
	\leq& \sum_g \tau(H_g) + \sum_g \tau(H_g') - 2 \sum_g \tau(H_g \cdot H_g') \\
	\leq& 2\sqrt{18\eta} + \eta.
\end{align*}
The third line follows from the fact that $\{ H_{[g(x)=a]} \}$ and  $\{ H_{[g(x)=a]}' \}$ are submeasurements. The fifth line follows from dropping terms $\tau(H_g \cdot H_{g'}')$ for $g \neq g'$, which are nonnegative. The sixth line follows from (each step of) the derivation~\eqref{eq:hp-proj}.

Using the fact that $\{A^x_a\}$ forms a complete measurement,
\begin{align}
\E_x \sum_{a \neq b} \tau \Big ( H_{[g \mapsto g(x) \mid a]}' \cdot A^x_b \Big) &= \tau(H') - \E_x \sum_a \tau \Big ( H_{[g \mapsto g(x) \mid a]}' \cdot A^x_a \Big) \notag\\
																 &\approx_{\eta'} \tau(H') - \E_x \sum_a \tau \Big ( H_{[g \mapsto g(x) \mid a]} \cdot A^x_a \Big) \notag\\
																 &\approx_{\eta'} \tau(H) -  \E_x \sum_a \tau \Big ( H_{[g \mapsto g(x) \mid a]} \cdot A^x_a \Big)\notag \\
																 &= \E_x \sum_{a \neq b} \tau \Big ( H_{[g \mapsto g(x) \mid a]} \cdot A^x_b \Big) \notag\\
																 &\leq 2\,{\zeta_{var}}\label{eq:si-hpacons}
\end{align}
where the second line follows from \Cref{lem:closeness-to-close-ips}, the third line follows from~\eqref{eq:hp-proj}, and the last line follows from Lemma~\ref{lem:si-cons}.

\subsubsection{Proof of \Cref{enu:self-improvement-boundedness} (Agreement is well-explained)}

The map $\psi$ constructed at the start of Section~\ref{sec:si-step1} satisfies $\psi(X) \geq \varphi_g(X) = \tau(X \cdot A_g)$ for all $g \in \code^{\otimes m}$ and positive $X \in \Bounded(\hilb)$. We compute
\begin{align*}
	\psi(\Id - H') &= \psi(\Id) - \psi(H') \\
					&= \psi(\Id) - \sum_g \psi (H_g') \\
					&\leq \psi(\Id) - \sum_g \tau(H_g' \cdot A_g) \\
					&= \psi(\Id) - \E_x \sum_a \tau (H_{[g \mapsto g(x) \mid a]}' \cdot A^x_a) \\
					&= \tau(H) - \E_x \sum_a \tau (H_{[g \mapsto g(x) \mid a]}' \cdot A^x_a) \\
					&\approx_{\eta'} \tau(H') - \E_x \sum_a \tau (H_{[g \mapsto g(x) \mid a]}' \cdot A^x_a) \\
					&= \E_x \sum_{a \neq b} \tau(H_{[g \mapsto g(x) \mid a]}' \cdot A^x_b) \\
					&\leq 2\eta' + 2\,{\zeta_{var}}
\end{align*}
where the equality $\psi(\Id) = \tau(H)$ follows from the equality between~\eqref{eq:si-hcomp1a} and~\eqref{eq:si-hcomp1b}, the sixth line follows from~\eqref{eq:hp-proj}, and the last line follows from~\eqref{eq:si-hpacons}.
This implies
\begin{equation*}
  \psi(\Id - H') \le 3\eta' + 2\,\zeta_{var}.
\end{equation*}

\subsubsection{Putting everything together}

To conclude the proof of \cref{lem:self-improvement}, we set
\[
  \zeta =  3\eta' + 2\,{\zeta_{var}}\;.
\]
First we note that, since $t = n - d + 1$, we can express $\gamma = 1 - \frac{d}{n} \leq \frac{t}{n}$ and
\[
	\gamma_m = 1 - \frac{d^m}{n^m} = 1 - \Big ( 1 - \frac{t - 1}{n} \Big)^m \leq  1 - \Big ( 1 - \frac{mt}{n} \Big) = \frac{mt}{n}.
\]
By the definitions of $\eta'$ and $\zeta_{var}$, we have
\begin{align*}
  \zeta_{var} & = \sqrt{m} (2\sqrt{2\eps} + 2\sqrt{\gamma}) = \poly(m,t) \cdot \poly\left(\eps,\frac{1}{n}\right)~, \\
  \eta & = 5 \zeta_{var} + \gamma_{m} = \poly(m,t) \cdot \poly\left(\eps,\frac{1}{n}\right),\\
  \eta' & = \sqrt{2\sqrt{18\eta} + \eta} = \poly(\eta) = \poly(m,t) \cdot \poly \left(\eps,\frac{1}{n}\right),\\
  \zeta & = 3\eta' + 2\zeta_{var} = \poly(m,t) \cdot \poly \left(\eps,\frac{1}{n}\right).
\end{align*}

\section{Pasting}\label{sec:pasting}

We now prove \Cref{lem:pasting}.
For convenience we restate the lemma here:

\pasting*


\subsection{Commutativity of the $A$'s}

We first derive a consequence of passing the subcube commutation test with high probability. 

\begin{lemma}\label{lem:a-comm}
	Let $\strategy = (\tau, A,B,P)$ be strategy for the tensor code test with 
	code $\code^{\otimes (m+1)}$ that passes the subcube commutation test with probability $1 - \delta$. 	
	Let $A = \{A^{u,x}\}_{u \in [n]^m, x \in [n]}$ denote the points measurements.
	On average over uniform and independently random $(u,x),(v,y) \sim [n]^{m} \times [n]$,
	\begin{equation*}
		A^{u,x}_{a} A^{v,y}_{b} \approx_{\sqrt{32(m+1)\delta}} A^{v,y}_{b} A^{u,x}_{a}.
	\end{equation*}
\end{lemma}

\begin{proof}
	Suppose the strategy $\strategy$ passes the subcube commutation test with probability $1-\delta$. Since the index $j=1$ in the test (see \Cref{fig:test})
	with probability $(m+1)^{-1}$, the strategy passes the test conditioned on $j=1$ with probability at least $1 - \delta (m+1)$. When $j=1$, the ``subcube'' is
	actually the entire space $[n]^{m+1}$. 
	This means that
	\begin{equation*}
		A^{u,x}_a \simeq_{(m+1)\delta} P^{u,x,v,y}_{[(a,b)\mapsto a]},
	\end{equation*}
	on average over $(u,x),(v,y) \in [n]^{m+1}$ which implies
	\begin{equation}\label{eq:a-comm-1}
		A^{u,x}_{a} \approx_{\sqrt{2(m+1)\delta}} P^{u,x,v,y}_{[(a,b)\mapsto a]},
	\end{equation}
	by \cref{lem:consistency-consequences}.

	By multiple applications of \cref{lem:add-a-proj}, we have
	\begin{align*}
		A^{u,x}_a A^{v,y}_b
		& \approx_{\sqrt{2(m+1)\delta}} A^{u,x}_a P^{u,x,v,y}_{[(a,b)\mapsto b]}\\
  	& \approx_{\sqrt{2(m+1)\delta}} P^{u,x,v,y}_{[(a,b)\mapsto a]} P^{u,x,v,y}_{[(a,b)\mapsto b]}\\
		& = P^{u,x,v,y}_{[(a,b)\mapsto b]} P^{u,x,v,y}_{[(a,b)\mapsto a]}\\
		& \approx_{\sqrt{2(m+1)\delta}} A^{v,y}_b P^{u,x,v,y}_{[(a,b)\mapsto a]}\\
		& \approx_{\sqrt{2(m+1)\delta}} A^{v,y}_b A^{u,x}_a.
	\end{align*}
	The theorem now follows from the triangle inequality.
\end{proof}

\subsection{Commutativity of the $G$'s}

Having deduced that the points measurements approximately commute on average, we now aim to deduce that the ``subspace  measurements'' $\{G^x\}$ corresponding to the parallel $m$-dimensional subspaces $S_x = \{ (u,x) : u \in [n]^m \}$ approximately commute. We first prove a slightly weaker statement, which is that the subspace measurements $\{G^x\}$ give rise to points measurements $\{G^{u,x}_a\}_{a \in \Sigma}$ that approximately commute. 

\begin{lemma}\label{lem:g-comm-after-eval}
Let $\nu_1 = 8(\sqrt{\zeta} + \sqrt{(m+1)\delta})$.
  Let $\{G^{x}_{g}\}$ be projective submeasurements satisfying the conditions in
  \cref{lem:pasting}.
	Define $G^{u,x}_{a} = G^{x}_{[g \mapsto g(u) \mid a]}$ for all $(u,x) \in {[n]}^{m} \times [n]$.
  Then on average over uniform and independently uniform $(u,x)$,
  $(v,y) \sim {[n]}^m \times [n]$,
  \[
		G^{u,x}_{a} G^{v,y}_{b} \approx_{\nu_{1}} G^{v,y}_{b} G^{u,x}_{a}~.
  \]
\end{lemma}

\begin{proof}
We compute
	\begin{align*}
	  & \E_{(u,x),(v,y)} \sum_{a,b} \tau \Paren{ \Bigl( G^{u,x}_{a} G^{v,y}_{b}
		- G^{v,y}_{b} G^{u,x}_{a}  \Bigr)^* \Bigl( G^{u,x}_{a}  G^{v,y}_{b}
		- G^{v,y}_{b} G^{u,x}_{a}  \Bigr) } \\
       = 2 & \E_{(u,x),(v,y)} \sum_{a,b} \tau \Paren{ G^{u,x}_{a} G^{v,y}_{b}
		- G^{v,y}_{b} G^{u,x}_{a}  G^{v,y}_{b} G^{u,x}_{a} }
	\end{align*}
	where in the second line we used the cyclicity of the trace $\tau$ and the
	fact that $\{G^{u,x}_a\}$ is projective.
	By the consistency assumption (Assumption~\ref{enu:pasting-consistency} of
	\cref{lem:pasting}), we have $G^{u,x}_{a} \simeq_{\zeta} A^{u,x}_a$ and
	therefore by \cref{lem:cons-sub-meas} we have, on average over
	$(u,x) \sim {[n]}^{m+1}$,
	\begin{align}
		G^{u,x}_{a} & \approx_{\sqrt{\zeta}} G^{u,x}_{a} A^{u,x}_{a}
									\label{eq:g-comm-after-eval}\\
		G^{u,x}_{a} & \approx_{2 \sqrt{\zeta}} G^{x} A^{u,x}_a
									\label{eq:g-comm-after-eval-0}
	\end{align}

	where $G^{x} = \sum_a G^{x}_{[g \mapsto g(u) \mid a]}$.

	The rest of the proof is outlined in the following sequence of approximations.
	\begin{alignat}{3}
		& && \E_{(u,x),(v,y)} \sum_{a,b} \tau \Paren{ G^{v,y}_{b} G^{u,x}_{a}
			G^{v,y}_{b} G^{u,x}_{a} } \nonumber \\
	  & \approx_{2 \sqrt{\zeta}} && \E_{(u,x),(v,y)} \sum_{a,b} \tau
		\Paren{ G^{v,y}_{b} G^{u,x}_{a} G^{v,y}_{b} G^{x} A^{u,x}_a } &
		\qquad\qquad\text{To be proved below} \label{eq:g-comm-after-eval-a} \\
	  & \approx_{\sqrt{\zeta}} && \E_{(u,x),(v,y)} \sum_{a,b} \tau
		\Paren{ G^{v,y}_{b} G^{u,x}_{a} G^{v,y}_{b} A^{u,x}_a }
	  & \text{To be proved below} \label{eq:g-comm-after-eval-1} \\
	  & = && \E_{(u,x),(v,y)} \sum_{a,b} \tau
		\Paren{A^{u,x}_a G^{v,y}_{b} G^{u,x}_{a} G^{v,y}_{b} }
	  & \text{Cyclicity of the trace} \nonumber \\
	  & \approx_{2 \sqrt{\zeta}} && \E_{(u,x),(v,y)} \sum_{a,b} \tau
		\Paren{A^{u,x}_a G^{v,y}_{b} G^{u,x}_{a} G^y A^{v,y}_b }
	  & \text{To be proved below} \label{eq:g-comm-after-eval-b} \\
	  & \approx_{\sqrt{\zeta}} && \E_{(u,x),(v,y)} \sum_{a,b} \tau
		\Paren{A^{u,x}_a G^{v,y}_{b} G^{u,x}_{a} A^{v,y}_{b} }
	  & \text{To be proved below} \label{eq:g-comm-after-eval-2} \\
	  & = && \E_{(u,x),(v,y)} \sum_{a,b} \tau
		\Paren{ G^{v,y}_{b} G^{u,x}_{a} A^{v,y}_{b} A^{u,x}_a }
	  & \text{Cyclicity of the trace} \nonumber \\
	  & \approx_{\sqrt{32(m+1)\delta}} && \E_{(u,x),(v,y)} \sum_{a,b} \tau
		\Paren{ G^{v,y}_{b} G^{u,x}_{a} A^{u,x}_a A^{v,y}_{b} }
	  & \text{\cref{lem:a-comm}, Prop.~\ref{lem:closeness-to-close-ips}} \nonumber \\
	  & \approx_{\sqrt{\zeta}} && \E_{(u,x),(v,y)} \sum_{a,b} \tau
		\Paren{ G^{v,y}_{b} G^{u,x}_{a} A^{v,y}_{b} }
	  & \text{To be proved below} \label{eq:g-comm-after-eval-3} \\
	  & \approx_{\sqrt{\zeta}} && \E_{(u,x),(v,y)} \sum_{a,b} \tau
		\Paren{ G^{v,y}_{b} G^{u,x}_{a}}.
	  & \text{To be proved below} \label{eq:g-comm-after-eval-4}
	\end{alignat}
	This shows that
	\[
	  \E_{(u,x),(v,y)} \sum_{a,b} \tau \Paren{ G^{v,y}_{b} G^{u,x}_{a}}
	  \approx_{8\sqrt{\zeta} + \sqrt{32(m+1)\delta}}
	  \E_{(u,x),(v,y)} \sum_{a,b} \tau \Paren{ G^{v,y}_{b} G^{u,x}_{a}
		G^{v,y}_{b} G^{u,x}_{a} }.
	\]

	\paragraph{Proof of approximation in \cref{eq:g-comm-after-eval-a,eq:g-comm-after-eval-b}.}
	Define $G^{u,x,v,y}_{a} = \sum_{b} G^{v,y}_{b} G^{u,x}_{a} G^{v,y}_{b}$ and it
	is easy to verify that $\{ G^{u,x,v,y}_{a} \}$ forms a submeasurement.
	Hence, we have
	\begin{equation*}
		\begin{split}
			\E_{(u,x),(v,y)} \sum_{a,b} \trace{ G^{v,y}_{b} G^{u,x}_{a}
				G^{v,y}_{b} G^{u,x}_{a}}
			& =_{\phantom{2\sqrt{\zeta}}} \E_{(u,x),(v,y)} \sum_{a} \trace { G^{u,x,v,y}_{a}
				\cdot G^{u,x}_{a} }\\
			& \approx_{2\sqrt{\zeta}} \E_{(u,x),(v,y)} \sum_{a} \trace{ G^{u,x,v,y}_{a} \cdot
				G^{x} A^{u,x}_{a}}\\
			& =_{\phantom{2\sqrt{\zeta}}} \E_{(u,x),(v,y)} \sum_{a,b} \trace{ G^{v,y}_{b}
				G^{u,x}_{a} G^{v,y}_{b} G^{x} A^{u,x}_{a}},
		\end{split}
	\end{equation*}
	where the approximation follows from
	\cref{lem:closeness-to-close-ips,eq:g-comm-after-eval-0}.
	This gives \cref{eq:g-comm-after-eval-a}.

	Similarly, to prove the approximation in \cref{eq:g-comm-after-eval-b}, we
	define $H^{u,x,v,y}_{b} = \sum_{a} A^{u,x}_{a} G^{v,y}_{b} G^{u,x}_{a}$.
	Then
	\begin{equation*}
		\begin{split}
		  \sum_{b} H^{u,x,v,y}_{b} {\bigl( H^{u,x,v,y}_{b} \bigr)}^{*}
		  & = \sum_{b}\sum_{a,a'} A^{u,x}_{a} G^{v,y}_{b} G^{u,x}_{a}
		  G^{u,x}_{a'} G^{v,y}_{b} A^{u,x}_{a'}\\
  		& = \sum_{b}\sum_{a} A^{u,x}_{a} G^{v,y}_{b} G^{u,x}_{a} G^{v,y}_{b} A^{u,x}_{a} \\
			& \le \sum_{a,b} A^{u,x}_{a} G^{v,y}_{b} A^{u,x}_{a}\\
			& \le \Id,
		\end{split}
	\end{equation*}
	where the second line uses the projectivity of $\bigl\{ G^{u,x}_{a} \bigr\}$
	and the third line uses $G^{u,x}_{a} \le \Id$.
	By \cref{lem:closeness-to-close-ips,eq:g-comm-after-eval-0}, we have
	\begin{equation*}
		\begin{split}
			\E_{(u,x),(v,y)} \sum_{a,b} \trace{ A^{u,x}_{a} G^{v,y}_{b}
				G^{u,x}_{a} G^{v,y}_{b} }
			& =_{\phantom{2\sqrt{\zeta}}} \E_{(u,x),(v,y)} \sum_{b}
			\trace{ H^{u,x,v,y}_{b} G^{v,y}_{b} }\\
			& \approx_{2\sqrt{\zeta}} \E_{(u,x),(v,y)} \sum_{b}
			\trace{ H^{u,x,v,y}_{b} G^{y} A^{v,y}_{b} }\\
			& =_{\phantom{2\sqrt{\zeta}}} \E_{(u,x),(v,y)} \sum_{a,b}
			\trace{ A^{u,x}_{a} G^{v,y}_{b} G^{u,x}_{a} G^{y} A^{v,y}_{b} }.
		\end{split}
	\end{equation*}

	\paragraph{Proof of approximations in \cref{eq:g-comm-after-eval-1,eq:g-comm-after-eval-2}.}
	Let $R^x_g = \E_{v,y} \sum_b G^{v,y}_b G^{x}_g G^{v,y}_b$.
	It can be verified that $\{R^x_g\}$ forms a submeasurement:
	\[
	  \sum_g R^x_g \leq \E_{v,y} \sum_b {(G^{v,y}_b)}^2 \le \Id.
	\]
	To show the approximation in \cref{eq:g-comm-after-eval-1}, we bound the
	magnitude of the difference:
	\begin{alignat*}{3}
	  & && \left| \E_{(u,x),(v,y)} \sum_{{a,b}} \tau \Paren{ G^{v,y}_{b} G^{u,x}_{a} G^{v,y}_{b}
		  \cdot (\Id - G^x) \cdot A^{u,x}_a } \right|\\
	  & = && \left| \E_{u,x} \sum_{a} \tau \Paren{ R^x_{[g \mapsto g(u) \mid a]}
		  \cdot (\Id - G^x) \cdot A^{u,x}_a } \right|  \\
	  & = && \left |\E_{u,x} \sum_{g} \tau \Paren{ R^x_{g}
		  \cdot (\Id - G^x) \cdot A^{u,x}_{g(u)} } \right | \\
	  &\leq_{\phantom{\sqrt{\eps}}} && \sqrt{\E_{x} \sum_{g} \tau \Paren{ R^x_g }}
	  \cdot \sqrt{ \E_{u,x} \sum_{g} \tau \Paren{ A^{u,x}_{g(u)} \cdot (\Id - G^x)
		  \cdot R^x_g \cdot (\Id - G^x) \cdot A^{u,x}_{g(u)} } }
	  & \quad \text{Cauchy-Schwarz} \\
	  &\leq && \sqrt{ \E_{u,x} \sum_{g} \tau \Paren{ A^{u,x}_{g(u)} \cdot (\Id - G^x)
		  \cdot R^x_g \cdot (\Id - G^x) \cdot A^{u,x}_{g(u)} } }
	  & \quad \text{$\{R^x_g\}$ is a submeasurement} \\
	  & = && \sqrt{ \E_{x} \sum_{g} \tau \Paren{ (\Id - G^x) \cdot R^x_g
		  \cdot (\Id - G^x) \cdot \E_u A^{u,x}_{g(u)} } }
	  & \quad \text{$\{A^x_a\}$ is projective} \\
	  & \leq && \sqrt{ \E_{x} \sum_{g} \psi^x \Paren{ (\Id - G^x)
		  \cdot R^x_g \cdot (\Id - G^x)} }
	  & \quad \text{Assumption~\ref{enu:pasting-boundedness}} \\
	  & \leq && \sqrt{ \E_{x} \psi^x \Paren{ \Id - G^x} }
	  & \quad \llap{\text{$\{R^x_g\}$ is submeasurement, $\{G^x\}$ is projective}} \\
	  & \leq && \sqrt{\zeta} & \quad \text{Assumption~\ref{enu:pasting-boundedness}}.
	\end{alignat*}
	The proof of approximation in \cref{eq:g-comm-after-eval-2} follows in an
	identical manner.
	
	\paragraph{Proof of approximations in \cref{eq:g-comm-after-eval-3,eq:g-comm-after-eval-4}.}
	To show the approximation in \cref{eq:g-comm-after-eval-3}, we define
	\begin{equation*}
		S^{v,y} = \sum_{b} A^{v,y}_{b} G^{v,y}_{b},
	\end{equation*}
	and verify that
	\begin{equation*}
		S^{v,y}{(S^{v,y})}^{*} = \sum_{b,b'} A^{v,y}_{b} G^{v,y}_{b} G^{v,y}_{b'} A^{b,y}_{b'}
		= \sum_{b} A^{v,y}_{b} G^{v,y}_{b} A^{b,y}_{b} \le \Id,
	\end{equation*}
	and
	\begin{equation*}
		\sum_{a} S^{v,y}G^{u,x}_{a} {(S^{v,y}G^{u,x}_{a})}^{*} \le S^{v,y} {(S^{v,y})}^{*} \le \Id.
	\end{equation*}
	Applying \cref{lem:closeness-to-close-ips} to \cref{eq:g-comm-after-eval}, we
	have
	\begin{equation*}
		\begin{split}
			\E_{(u,x),(v,y)} \sum_{a,b} \trace{ A^{v,y}_{b} G^{v,y}_{b} G^{u,x}_{a}
				A^{u,x}_{a} }
			& =_{\phantom{\sqrt{\zeta}}} \E_{(u,x),(v,y)} \sum_{a}
			\trace {S^{v,y} G^{u,x}_{a} \cdot G^{u,x}_{a} A^{u,x}_{a}}\\
		  & \approx_{\sqrt{\zeta}} \E_{(u,x),(v,y)} \sum_{a}
			\trace {S^{v,y} G^{u,x}_{a} \cdot G^{u,x}_{a}} \\
			& =_{\phantom{\sqrt{\zeta}}}
		  \E_{(u,x),(v,y)} \sum_{a,b} \trace{ A^{v,y}_{b} G^{v,y}_{b} G^{u,x}_{a}},
		\end{split}
	\end{equation*}
	which proves \cref{eq:g-comm-after-eval-3} by cyclicity of the trace.

	The proof of approximation in \cref{eq:g-comm-after-eval-4} follows in an
	identical manner.

\end{proof}

We now show that the subspace measurements $\{G^x\}$ themselves approximately commute. 

\begin{lemma}\label{lem:g-comm}
  Let $\gamma_m$ be as defined in Proposition~\ref{prop:distance} and $\nu_2 = 4 (\gamma_m + \nu_1)$. Let $\{G^{x}_{g}\}$ be projective submeasurements satisfying the conditions in
	\cref{lem:pasting}.
	
	Then on average over independently uniform $x,y \sim [n]$,
	\[
		G^x_g \, G^y_h \approx_{\nu_2} G^y_h \, G^x_g.
	\]
\end{lemma}

\begin{proof}
	By definition, we need to bound
	\begin{equation*}
		\E_{x,y} \sum_{g,h} \tau \Paren{ {\Big(G^x_g \, G^y_h - G^y_h \,
			G^x_g  \Big)}^* \Big(G^x_g \, G^y_h - G^y_h \, G^x_g \Big) }
		=  2 \E_{x,y} \sum_{g,h} \tau \Paren{G^x_g \, G^y_h - G^y_h \,
			G^x_g \,  G^y_h \, G^x_g }
	\end{equation*}
	where we used the cyclicity of the trace $\tau$ and the projectivity of $\{G^x_g\}$.

  For notational convenience we use the abbreviation
	$G^{u,x}_{a} = G^x_{[g \mapsto g(u) \mid a]}$ for all
	$(u,x) \in {[n]}^m \times [n]$.
	We have
	\begin{alignat}{3}
		& && \E_{x,y} \sum_{g,h} \tau
		\Paren{G^y_h \, G^x_g \,  G^y_h \, G^x_g } \nonumber \\
		&\approx_{\gamma_m\phantom{m}} && \E_{(u,x),y} \sum_{a,h} \tau
		\Paren{G^y_h \, G^{u,x}_{a}
			\, G^y_h \, G^{u,x}_{a} }
		& \qquad \qquad \qquad \qquad \text{To be proved below} \label{eq:g-comm-1} \\
		& \approx_{\gamma_m} && \E_{(u,x),(v,y)} \sum_{a,b} \tau
		\Paren{G^{v,y}_{b} \, G^{u,x}_{a} \, G^{v,y}_{b} \, G^{u,x}_{a}}
		& \text{To be proved below} \label{eq:g-comm-2} \\
		& \approx_{\nu_1} && \E_{(u,x),(v,y)} \sum_{a,b} \tau
		\Paren{G^{v,y}_{b} \, G^{u,x}_{a} \, G^{u,x}_{a} \, G^{v,y}_{b}}
		& \text{\cref{lem:closeness-to-close-ips,lem:g-comm-after-eval}} \nonumber\\
		& = && \E_{(u,x),(v,y)} \sum_{a,b} \tau \Paren{G^{v,y}_{b} \, G^{u,x}_{a}}
		& \text{Projectivity of $G$'s} \nonumber \\
		& = && \E_{x,y} \sum_{g,h} \tau \Paren{G^x_g \, G^y_h}. \nonumber
	\end{alignat}
	
	\paragraph{Proof of approximation in \cref{eq:g-comm-1}.} 
	To show the approximation in \cref{eq:g-comm-1}, we bound the magnitude of the
	difference:
	\begin{alignat*}{3}
		& && \left | \E_{(u,x),y} \sum_{\substack{g,g',h: \\ g \neq g'}} \,
			\indicator[g(u) = g'(u)] \, \tau \Paren{G^y_h \, G^x_g
				\, G^y_h \, G^x_{g'} } \right | \\
		& =_{\phantom{\sqrt{\eps}}} && \E_{x,y} \sum_{\substack{g,g',h: \\ g \neq g'}}
		\tau \Paren{G^y_h \, G^x_g \, G^y_h \, G^x_{g'} }
		\, \E_u \indicator[g(u) = g'(u)]
		& \qquad \qquad \text{$\tau(G^y_h \, G^x_g \,  G^y_h \, G^x_{g'}) \geq 0$} \\
		& \leq && \E_{x,y} \sum_{\substack{g,g',h: \\ g \neq g'}}  \tau
		\Paren{G^y_h \, G^x_g \,  G^y_h \, G^x_{g'} } \,
		\Big (1 - \frac{d^m}{n^m} \Big)
		& \text{Proposition~\ref{prop:distance}} \\
		& \leq && \gamma_{m}
	\end{alignat*}
	where in the last inequality we used that summing $\tau
		\Paren{G^y_h \, G^x_g \,  G^y_h \, G^x_{g'} }$ over $g,g',h$ is at most $1$.

	\paragraph{Proof of approximation in \cref{eq:g-comm-2}.} 	
	To show the approximation in \cref{eq:g-comm-2}, we bound the magnitude of the difference:	
	\begin{alignat*}{3}
		& && \left | \E_{(u,x),(v,y)} \sum_{\substack{a,h,h': \\ h \neq h'}} \, \indicator[h(v) = h'(v)]
			\, \tau \Paren{G^y_h \, G^{u,x}_{a} \, G^y_{h'} \, G^{u,x}_{a} } \right | \\
		& =_{\phantom{\sqrt{\eps}}} && \E_{(u,x),y} \sum_{\substack{a,h,h': \\ h \neq h'}} \tau
		\Paren{G^y_h \, G^{u,x}_{a} \, G^y_{h'} \, G^{u,x}_{a} } \, \E_v \indicator[h(v) = h'(v)] \\
		& \leq && \gamma_{m}\;.
	\end{alignat*}
\end{proof}

\subsection{Pasting the $G$'s together: Method 1}
\label{sec:pasting-1}

Now that we have established the subspace submeasurements $\{G^x\}$ approximately commute, we ``paste'' them together into a single submeasurement $H = \{H_h\}$ with outcomes in $\code^{\otimes (m+1)}$ that are consistent with the $G$'s. We give two methods for achieving this. The first, detailed in this section, is the simpler method, but it achieves a weaker bound: instead of the value of $\mu$ claimed in \Cref{lem:pasting} this method obtains
\begin{equation}\label{eq:kappa-b2}
\mu(\kappa,m,t,\eps,\delta,\zeta,n^{-1}) = \kappa + \poly(m,t) \cdot \left( \poly\left
    (\eps,\delta,\zeta,n^{-1}\right) \right)^{1/t} \;,
\end{equation}
with a dependence on $1/t$ in the exponent.
Intuitively, for this method we define $H$ as a measurement that tries to
``simultaneously'' measure $t = n -d + 1$ subspace submeasurements
$G^{x_1},\ldots,G^{x_t}$ for a randomly chosen tuple of distinct coordinates
$(x_1,\ldots,x_t) \in [n]^t$, obtain $m$-dimensional words
$g_1,\ldots,g_t \in \code^{\otimes m}$, and then interpolate a ``global''
codeword $h \in \code^{\otimes m+1}$ such that $h(u,x_i) = g_i(u)$.
Since the base code $\code$ is interpolable, by \Cref{prop:interpolate-tuple},
there is always a global codeword $h$ consistent with $g_1,\ldots,g_t$.

While the exponential dependence on $1/t$ in~\eqref{eq:kappa-b2} remains acceptable for applications where $t$ is thought of as constant, if the parameter grows even e.g.\ logarithmically with $n$ the bound can become trivial. To address this case in Section~\ref{sec:pasting-2} we give a different method (explained in that section), which obtains the bound claimed in \Cref{lem:pasting}.

We now formalize the first ``interpolation'' approach. We start by precisely defining the pasted submeasurement $H$.


Let $x_1,\ldots,x_t \in [n]$.
Define an initial pasted submeasurement as follows:
\[
	H^{x_1,\ldots,x_t}_{g_1,\ldots,g_t} = G^{x_1}_{g_1} \cdots G^{x_t}_{g_t} \cdots
	G^{x_1}_{g_1}.
\]
This forms a submeasurement because of each of the individual $G^x$'s are
submeasurements.
Then, for all $h \in \code^{\otimes(m+1)}$, define the following operators
\[
	H^{x_1,\ldots,x_t}_h = H^{x_1,\ldots,x_t}_{h|_{x_1},\ldots,h|_{x_t}}
\]
where $h|_{x_j}$ denotes the codeword in $\code^{\otimes m}$ that comes from
setting the $(m+1)$-st coordinate of $h$ to $x_j$.
This set $H^{x_1,\ldots,x_t} = \{H^{x_1,\ldots,x_t}_h\}$ forms a submeasurement.
This is because the set
\[
	S = \{ (g_1,\ldots,g_t) \in \codemt :
	\text{there exists $h \in \code^{\otimes (m+1)}$ such that } h|_{x_j} = g_j \}
\]
is a subset of $\codemt$, and furthermore $S$ is in one-to-one correspondence
with $\code^{\otimes (m+1)}$ as shown in \Cref{prop:tuple-to-code-correspondence}.
Thus
\[
	\sum_h H^{x_1,\ldots,x_t}_h = \sum_{(g_1,\ldots,g_t) \in S}
	H^{x_1,\ldots,x_t}_{g_1,\ldots,g_t} \leq \sum_{(g_1,\ldots,g_t) \in \codemt}
	H^{x_1,\ldots,x_t}_{g_1,\ldots,g_t} \leq \Id.
\]
And finally, define the operators
\[
	H_h = \E_{(x_1,\ldots,x_t) \sim \distinct([n],t)} H^{x_1,\ldots,x_t}_h
\]
where we define $\distinct([n],t)$ to denote the set of tuples
$(x_1,\ldots,x_t) \in {[n]}^t$ such that all of the coordinates are distinct.
It is easy to verify that $H = \{H_h\}$ forms a submeasurement.
In what follows, $H$ will interchangeably refer to both the set $\{H_h\}$ as
well as the sum $\sum_h H_h$; it should be clear from context which we are
referring to. We will repeatedly use the following simple claim.

\begin{proposition}\label{prop:ld-dnoteq}
Let $n,k\geq 1$ be integer.
Let $x = (x_1, \ldots, x_k) \sim [n]^k$ be sampled uniformly at random
and let $y = (y_1, \ldots, y_k) \sim \distinct([n],k)$.
Then
\begin{equation*}
d_{\mathrm{TV}}(x, y) \leq \frac{k^2}{n}.
\end{equation*}
\end{proposition}

\begin{proof}
For any $z = (z_1, \ldots, z_k) \in [n]^k$, 
\begin{align*}
\Pr(x = z) 
& = \frac{1}{q^k},\\
\Pr(y = z) 
& = \left\{\begin{array}{cl}
		\frac{1}{\binom{n}{k} k!} & \text{if $z \in \distinct([n],k)$,}\\
		0 & \text{otherwise.}
		\end{array}\right.
\end{align*}
Because $\frac{1}{\binom{n}{k} k!} \geq \frac{1}{n^k}$, $\Pr(x = z) \geq \Pr(y = z)$ if and only if $z \notin \distinct([n],k)$.
Hence,
\begin{align*}
d_{\mathrm{TV}}(x, y)
&= \max_{S \subseteq [n]^k}\{\Pr(x \in S) - \Pr(y \in S)\}\\
&= \Pr(x \in \overline{\distinct([n],k)}) - \Pr(y \in \overline{\distinct([n],k)})
= \Pr(x \in \overline{\distinct([n],k)}]\;.
\end{align*}
We can upper-bound this probability as follows.
\begin{align*}
\Pr(x \in \overline{\distinct([n],k)})
&= \Pr(\exists i : x_i \in \{x_1, \ldots, x_{i-1}\})\\
&\leq \sum_{i=2}^k \Pr(x_i \in \{x_1, \ldots, x_{i-1}\})
\leq \sum_{i=2}^k \left(\frac{i-1}{n}\right)
= \frac{k(k-1)}{2n}\;.
\end{align*}
\end{proof}

\subsubsection{Bounding the consistency of $H$ with the $A$ measurements}
\label{sec:h-bu}


First we bound the consistency of $H$ with line measurements $B^\ell$ where
$\ell$ are parallel to the $(m+1)$-st axis.
In other words, we consider lines $\ell$ consisting of points that only vary in
the $(m+1)$-st coordinate.
Since we focus on such lines, we use the following abbreviation:
\[
	B^u_f = B^\ell_f
\]
where $\ell = \ell(m+1,u)$ (see \Cref{def:axis-line} for the notation for axis-parallel lines).

Since $(\tau,A,B,P)$ is an $(\eps,\delta)$-good strategy for
$\code^{\otimes (m+1)}$, we have by definition that
\[
	\E_{\ell, y \sim \ell} \sum_{f \in \code}
	\trace{ B^\ell_f \, A^{y}_{f(y_j)}} \geq 1 - \eps
\]
where the expectation is over all axis-parallel lines
$\ell = \ell(j,\alpha_{-j})$ (not necessarily the ones parallel to direction
$m+1$) and points $y \in \ell$.
Since a line $\ell$ that is parallel to the $(m+1)$-st direction is selected
with probability $1/(m+1)$, this implies that
\begin{equation}
	\label{eq:pasting-b-cons-a}
	\E_{(u,x) \in {[n]}^m \times [n]} \sum_{f \in \code}
	\trace{ B^u_f \, A^{u,x}_{[a : a \neq f(x)]}} \leq (m+1) \eps \;.
\end{equation}
Define measurement $\{ B^{u,x}_{a}\}$ as
\begin{equation*}
	B^{u,x}_{a} = B^{u}_{[f \mapsto f(x) \mid a]}\;.
\end{equation*}
Then \cref{eq:pasting-b-cons-a} can be also written as
\begin{equation}
	\label{eq:pasting-b-cons-a-2}
	\E_{(u,x) \in {[n]}^m \times [n]} \sum_{a,b: a\ne b}
	\trace{ B^{u,x}_{b} A^{u,x}_{a}} \leq (m+1) \eps\;.
\end{equation}
By \cref{lem:consistency-consequences} this implies that on average over
$(u,x) \in {[n]}^m \times [n]$,
\begin{equation}
\label{eq:pasting-b-close-a}
	B^{u,x}_{a} \approx_{\sqrt{2 (m+1) \eps}} A^{u,x}_a \;.
\end{equation}
\noindent We first show that the $\{H_{h}\}$ submeasurement and the $\{B^{u}_{f}\}$
measurements are consistent.
\begin{lemma}\label{lem:pasting-h-cons-b}
Let $\nu_3 = t \cdot \left (t \cdot \nu_{2} + \sqrt{\zeta + \sqrt{2 (m+1) \eps}} \right)$.
	On average over $u$ sampled uniformly from ${[n]}^m$, we have
  \[
		H_{[h \mapsto h|_u \mid f]} \simeq_{\nu_3} B^u_f
	\]
	where the answer summation is over $f \in \code$.
\end{lemma}

\begin{proof}
	Let $\D$ be the uniform distribution on $\distinct([n],t)$ and $\D_{i}$ be the marginal
	distribution of $\D$ on $i$ coordinates.
	\begin{align*}
		\E_{u} \sum_{f' \ne f} \tau \Paren{ H_{[h \mapsto h|_u \mid f']} \, B^{u}_{f} }
		& = \E_u \sum_{\substack{h \in \code^{\otimes (m+1)}, f \in \code \\ h|_u \neq f}}
		\trace{H_h \, B^u_f }\\
		& = \E_u \E_{(x_1,\ldots,x_t) \sim \D}
	  	\sum_{\substack{h \in \code^{\otimes (m+1)}, f \in \code: \\ h|_u \neq f}}
		\trace{H_h^{x_1,\ldots,x_t} \, B^u_f } \\
	  & = \E_u \E_{(x_1,\ldots,x_t) \sim \D}
			\sum_{\substack{h ,f:\\ \exists \,i : h(u,x_i) \neq f(x_i)}}
		\trace{H_h^{x_1,\ldots,x_t} \, B^u_f } \\
	  & \leq \sum_{i=1}^{t} \E_u \E_{(x_1,\ldots,x_t) \sim \D}
			\sum_{\substack{h,f: \\ h(u,x_i) \neq f(x_i)}}
		\trace{H_h^{x_1,\ldots,x_t} \, B^u_f }
	\end{align*}
	For any fixed $i \in [t]$, we have
	\begin{alignat}{2}
		& && \E_{u} \E_{(x_1,\ldots,x_t) \sim \D}
		\sum_{\substack{h,f: \\ h(u,x_i) \neq f(x_i)}}
		\trace{H_h^{x_1,\ldots,x_t} \, B^u_f } \notag\\
		& = && \E_{u}\E_{(x_1,\ldots,x_t) \sim \D}
		\sum_{\substack{h,f: \\ h(u,x_i) \neq f(x_i)}}
		\trace{ G^{x_1}_{h|_{x_1}} \cdots\, G^{x_t}_{h|_{x_t}} \cdots\,
			G^{x_1}_{h|_{x_1}} \cdot B^u_f } \notag \\
		& = && \E_{u} \E_{(x_1,\ldots,x_t) \sim \D}
		\sum_{\substack{g_1,\ldots,g_t,f: \\ g_i(u) \neq f(x_i)}}
		\trace{ G^{x_1}_{g_1} \cdots G^{x_t}_{g_t} \cdots
			G^{x_1}_{g_1} \cdot B^u_f } \notag \\
	  & \leq && \E_{u} \E_{(x_1,\ldots,x_i) \sim \D_{i}}
		\sum_{g_1,\ldots,g_i}
		\trace{B^u_{[f : f(x_i) \neq g_i(u)]} \cdot G^{x_1}_{g_1}
			\cdots G^{x_i}_{g_i} \cdot \Bigl( G^{x_{i}}_{g_{i}} \cdot
			G^{x_{i-1}}_{g_{i-1}} \cdots G^{x_1}_{g_1} \Bigr) } \notag \\
	  & \approx_{i \cdot \nu_{2} \;} && \E_{u} \E_{(x_1,\ldots,x_i) \sim \D_{i}}
		\sum_{g_1,\ldots,g_i} \trace{ B^u_{[f : f(x_i) \neq g_i(u)]} \cdot
			G^{x_1}_{g_1} \cdots G^{x_i}_{g_i} \cdot \Bigl( G^{x_{i-1}}_{g_{i-1}}
			\cdots G^{x_1}_{g_1} \cdot G^{x_i}_{g_i} \Bigr) }
		\label{eq:pasting-h-cons-b-1}
	\end{alignat}
	where the inequality comes from the fact that $G^{x_{i+1}},\ldots,G^{x_t}$ are submeasurements, and the last approximation follows from
	\cref{lem:closeness-to-close-ips,lem:switcheroo} and the approximate
	commutativity of the $G^{x}$ measurements (\Cref{lem:g-comm}).
	Continuing on, we can use Cauchy-Schwarz to bound~\eqref{eq:pasting-h-cons-b-1} by
	\begin{align*}
		\text{\eqref{eq:pasting-h-cons-b-1}}
		& \leq \sqrt{ \E_{(x_1,\ldots,x_i) \sim \D_{i}} \sum_{g_1,\ldots,g_i}
		 \trace{ G^{x_1}_{g_1} \cdots G^{x_i}_{g_i} \cdots G^{x_1}_{g_1}}} \\
		& \qquad \qquad \cdot \sqrt{ \E_{u} \E_{(x_1,\ldots,x_i) \sim \D_{i}}
		 \sum_{g_1,\ldots,g_i} \trace{ G^{x_1}_{g_1} \cdots G^{x_{i-1}}_{g_{i-1}}
		 \cdots G^{x_1}_{g_1} \cdot G^{x_i}_{g_i} \cdot
		 {\bigl( B^u_{[f : f(x_i) \neq g_i(u)]} \bigr)}^2 \cdot G^{x_i}_{g_i}}} \\
		& \leq \sqrt{1} \cdot \sqrt{ \E_{u,x_i} \sum_{g_i}
		 \trace{ G^{x_i}_{g_i} \, B^u_{[f : f(x_i) \neq g_i(u)]}}}
	\end{align*}
	where in the last inequality we used the fact that the $G^{x}$ is projective
	and
	${\bigl( B^u_{[f : f(x_i) \neq g_i(u)]} \bigr)}^{2} \le B^u_{[f : f(x_i) \neq g_i(u)]}$
	for all $u$.
	Notice that
	\begin{align*}
	  \E_{u,x_i} \sum_{g_i} \trace{ G^{x_i}_{g_i} B^u_{[f : f(x_i) \neq g_i(u)]}}
		& = \E_{u,x}  \sum_{a,b:a\ne b} \trace{ G^{u,x}_{a} B^{u,x}_{b} } \\
	  & \approx_{ \sqrt{2(m+1) \eps} } \E_{u,x} \sum_{a,b:a \ne b}
			\trace{ G^{u,x}_{a} A^{u,x}_b } \\
	  & \leq \zeta \;.
	\end{align*}
where the approximation in the second line follows from \cref{lem:closeness-to-close-ips} and \cref{eq:pasting-b-close-a}.
	Putting everything together, this implies that
	\[
		\E_u \sum_{\substack{h \in \code^{\otimes (m+1)}, f \in \code \\ h|_u \neq f}}
		\trace{H_h \, B^u_f } \leq \nu_3,
	\]
	for $\nu_{3}$ given in the statement of the lemma.
\end{proof}

Next we argue that $H$ is consistent with the points measurements $A$. 

\begin{lemma}\label{lem:pasting-h-cons-a}
Let $\nu_4 = \nu_3 + \sqrt{2(m+1)\eps}$. On average over $(u,x)$ sampled uniformly from ${[n]}^m \times [n]$, we have
\[
	H_{[h \mapsto h(u,x) \mid a]} \simeq_{\nu_4} A^{u,x}_a \;.
\]
where the answer summation is over $a \in \Sigma$.
\end{lemma}

\begin{proof}
	Define $H^{u,x}_{a} = H_{[h\mapsto h(u,x)|a]}$. We have
	\begin{equation*}
		\begin{split}
			\E_{u,x} \sum_{\substack{a,b:\\a\ne b}} \trace{H^{u,x}_{a} B^{u,x}_{b}}
			& = \E_{u,x} \sum_{\substack{a,b:\\a\ne b}}
			\trace{\biggl( \sum_{h: h(u,x)=a} H_{h} \biggr)
				\biggl( \sum_{f: f(x)=b} B^{u}_{f} \biggr)}\\
			& = \E_{u,x} \sum_{\substack{h,f:\\h(u,x) \ne f(x)}} \trace{H_{h} B^{u}_{f}}\\
			& \le \E_{u,x} \sum_{\substack{h,f}: h|_{u} \ne f} \trace{H_{h} B^{u}_{f}},
		\end{split}
	\end{equation*}
	which is at most $\nu_{3}$ by \cref{lem:pasting-h-cons-b}. That is
	\begin{equation*}
		H^{u,x}_{a} \simeq_{\nu_{3}} B^{u,x}_{a}.
	\end{equation*}
	Applying \cref{lem:transfer-cons} to the above and \cref{eq:pasting-b-close-a}
	proves that
	\begin{equation*}
		H^{u,x}_{a} \simeq_{\nu_{4}} A^{u,x}_{a},
	\end{equation*}
	for $\nu_{4} = \nu_{3} + \sqrt{2(m+1)\eps}$.

\end{proof}
	
\subsubsection{Bounding the completeness of $H$}

We now bound the completeness of $H$. Let $G = \E_x G^x$ where $G^x = \sum_g G^x_g$.

\begin{lemma}
\label{lem:pasting-h-close-hg}
Let $\nu_5 = \sqrt{\nu_4} + \sqrt{\zeta}$. Then
\[
	\tau(H) \approx_{\nu_5} \tau(HG)~.
\]
\end{lemma}

\begin{proof}
	We bound the difference
	\begin{align}
		\trace{H (\Id - G)}
		& = \E_{x} \sum_h \trace{H_h (\Id - G^x)}\nonumber\\
		& \approx_{\sqrt{\nu_4}\;\;} \E_{u,x} \sum_h \trace{H_h A^{u,x}_{h(u,x)}
		 (\Id - G^x)}.
		\label{eq:pasting-h-close-hg-a}
	\end{align}
	To see the approximation in the second line, it suffices to bound the
	difference as
	\begin{alignat*}{3}
		& && \Big | \E_{u,x} \sum_h \trace{ H_h \bigl(\Id - A^{u,x}_{h(u,x)}\bigr)
		(\Id - G^x)} \Big |\\
		& \leq \;\; && \sqrt{ \E_{u,x} \sum_h \trace{ H_h \bigl(\Id - A^{u,x}_{h(u,x)}\bigr)^2 }}
		\cdot \sqrt{ \E_{u,x} \sum_h \trace{ H_h (\Id - G^x)^2 }}
		& \qquad\qquad \text{Cauchy-Schwarz}\\
		& \leq && \sqrt{ \E_{u,x} \sum_a \trace{ H_{[h : h(u,x) \neq a]}
			A^{u,x}_a}} \cdot \sqrt{1} \\
		& \leq && \sqrt{\nu_4} & \text{\Cref{lem:pasting-h-cons-a}}
	\end{alignat*}
	where the third line follows from the projectivity of the points measurements, $G^x \leq \id$, and $\{H_h\}$ is a submeasurement. 
	Continuing the proof, we bound the absolute value of~\cref{eq:pasting-h-close-hg-a} as
	\begin{alignat*}{3}
		& && \left| \E_{u,x} \sum_h \trace{H_h A^{u,x}_{h(u,x)} (\Id - G^x)} \right|\\
		& \le\;\; && \sqrt{\E_{u,x}\sum_{h} \trace{H_{h}}} \cdot \sqrt{\E_{u,x} \sum_{h}
			\trace{(\id-G^{x}){\bigl( A^{u,x}_{h(u,x)} \bigr)}^{2}(\id-G^{x})H_{h}}}
		& \text{Cauchy-Schwarz} \\
		& \leq && \sqrt{1} \cdot \sqrt{ \E_x \sum_h \trace{(\Id - G^x) H_h
			(\Id - G^x) \E_u A^{u,x}_{h(u,x)} }} \\
		& \le && \sqrt{ \E_x \sum_h \psi^x ( (\Id - G^x) \, H_h \, (\Id - G^x)) }
		& \text{Assumption~\ref{enu:pasting-boundedness} of \cref{lem:pasting}}\\
		& \le && \sqrt{ \E_{x} \psi^{x}(\Id - G^{x}) } & \text{$\psi^x$ is positive} \\
		& \le && \sqrt{\zeta} & \text{Assumption~\ref{enu:pasting-boundedness} of \cref{lem:pasting}}.
	\end{alignat*}

\end{proof}

We now argue that $HG$ can be approximated by $G^{t+1}$. 

\begin{lemma}
\label{lem:nu-6}
Let $\nu_6 = 2t^2( \nu_{2} + n^{-1})$. Then
	\[
		\tau(H) \approx_{\nu_6} \tau(G^t) \qquad \text{and} \qquad
		\tau(HG) \approx_{\nu_6} \tau(G^{t+1})~.
	\]
\end{lemma}

\begin{proof}
  We prove the first approximation. 
	\begin{align}
          \trace{H} & = \E_{(x_1, \ldots, x_t) \sim \distinct([n],t)}
                      \sum_{h \in \code^{\ot (m+1)}} \trace{
                      G^{x_1}_{h|_{x_1}} \cdots G^{x_t}_{h|_{x_t}}
                      \cdots G^{x_1}_{h|_{x_1}}}\\
                    &= \E_{(x_1,\ldots,x_t) \sim \distinct([n],t)} \sum_{g_1,\ldots,g_t}
								\trace{ G^{x_1}_{g_1} \cdots G^{x_t}_{g_t} \cdots G^{x_1}_{g_1}} \\
							& \approx_{t^2/n} \E_{(x_1,\ldots,x_t) \sim [n]^t} \sum_{g_1,\ldots,g_t}
							 \trace{ G^{x_1}_{g_1} \cdots G^{x_t}_{g_t} \cdots G^{x_1}_{g_1}},
							 \label{eq:pasting-gh-close-gt-1}
	\end{align}
        where in going from the first to the second line we used
        \Cref{prop:interpolate-tuple} (which uses the interpolability of $\code$) to rewrite the sum over $h$ to
        sums over $t$-tuples of polynomials $g_1, \dots, g_t$; and in
        going from the second to the third line, we used \Cref{prop:ld-dnoteq} and that the absolute value of
	\[
		\sum_{g_1,\ldots,g_t} \trace{ G^{x_1}_{g_1} \cdots G^{x_t}_{g_t} \cdots G^{x_1}_{g_1}}
	\]
	for all $x_1,\ldots,x_t$ is at most $1$.
	Continuing, we have
\begin{alignat*}{3}
	\text{\eqref{eq:pasting-gh-close-gt-1}}
	& = && \E_{(x_1,\ldots,x_t) \sim [n]^t} \sum_{g_1,\ldots,g_t}
	\trace{ G^{x_1}_{g_1} \cdots G^{x_t}_{g_t} \cdots G^{x_2}_{g_2}}
	& \qquad \qquad \text{Cyclicity of the trace} \\
	& \approx_{(2t - 5)\nu_{2}\;} && \E_{(x_1,\ldots,x_t) \sim [n]^t} \sum_{g_1,\ldots,g_t}
	\trace{ G^{x_1}_{g_1} \cdots G^{x_t}_{g_t} \cdots G^{x_3}_{g_3}}
\end{alignat*}
In other words, the right-most $G^{x_2}_{g_2}$ has been commuted leftwards to where the
left occurence of $G^{x_2}_{g_2}$ is; this requires $2t - 5$ (approximate)
commutations.
This follows from \cref{lem:switcheroo}, the approximate commutativity of the
$G$'s, and the projectivity of the $G$'s.
We can then commute the right-hand $G^{x_3}_{g_3}$ to the left $2t - 7$ places,
and so on.
Adding up all the errors together, we get that
\[
	\tau(H) \approx_{t^2 ( \nu_{2} + n^{-1})} \E_{(x_1,\ldots,x_t) \sim [n]^t}
	\sum_{g_1,\ldots,g_t} \trace{ G^{x_1}_{g_1} \cdots G^{x_t}_{g_t}} = \trace{G^t}
\]
where we used that $\sum_{j=2}^{t-1} (2(t-j) - 1) \leq t^2$. 

The second approximation in the Lemma statement follows from a nearly identical
argument.
\end{proof}

\begin{lemma}
Let $\nu_7 = \nu_6 + 2 ( \nu_5 + 2\nu_6)^{1/t}$. Then
	\[
		\tau(H) \geq \tau(G) - \nu_7~.
	\]
\end{lemma}
\begin{proof}
	Since for all real numbers $0 \leq \lambda \leq 1$ and integer $t \geq 2$, we have
	\[
		\lambda(1 - \lambda^{t-1}) \leq 2 \Big (\lambda^t (1 - \lambda) \Big)^{1/t}
	\]
	we have the operator inequality
	\[
		G - G^t \leq 2 \Big(G^t ( \Id - G) \Big)^{1/t} \;.
	\]
	and thus applying $\trace{\cdot}$ on both sides, we have
	\[
		\tau(G - G^t) \leq 2 \trace{ \Big(G^t ( \Id - G) \Big)^{1/t} }
		\leq 2 {\Bigl( \trace{ G^t - G^{t+1}} \Bigr)}^{1/t}
	\]
	where we used that $\trace{X^{1/t}} \leq \trace{X}^{1/t}$ for all positive operators $X$ (this is a consequence of H\"{o}lder's inequality, see \cite{pisier2003non}). This implies that 
	\[
		| \tau(H) - \tau(G) | \leq \tau(G - G^t) + |\tau(H - G^t)| \leq \tau(G - G^t) + \nu_6
	\]
	and
	\[
	\tau(G - G^t) \leq 2 \Big( | \trace{G^t - H} | + |\trace{H - HG}| + \left |\trace{HG - G^{t+1}} \right | \Big)^{1/t} \leq 2 ( \nu_6 + \nu_5 + \nu_6)^{1/t}
	\]
	where we used \Cref{lem:pasting-h-close-hg,lem:nu-6}. 
	Put altogether we have 
	\[
		|\tau(H) - \tau(G)| \leq \nu_6 + 2 ( \nu_5 + 2\nu_6)^{1/t}.
	\]
\end{proof}

\subsubsection{Putting everything together}\label{sec:together}

We now finish the proof of \Cref{lem:pasting}, with the error function $\mu$ specified in~\eqref{eq:kappa-b2} instead of~\eqref{eq:sigma-def}. (The claimed bound~\eqref{eq:sigma-def} is shown in the next section.) Let $H' = \{H'_h\}$ denote the completion of $H = \{H_h\}$ as follows: letting $h^*$ be the lexicographically first element of $\code^{\otimes (m+1)}$, we set $H'_{h^*} = H_{h^*} + (\Id - H)$. Then for all $h \neq h^*$, define $H_h' = H_h$. By construction $H'$ is a perfectly complete measurement. We now evaluate its consistency with $A$:
\begin{align*}
	\E_{u,x} \sum_{h,a: h(u,x) \neq a} \trace{ H'_h \, A^{u,x}_a} &\leq \E_{u,x} \sum_{h,a: h(u,x) \neq a} \trace{ H_h \, A^{u,x}_a} + \sum_a \trace{(\Id - H) A^{u,x}_a} \\
	&= \E_{u,x} \sum_{h,a: h(u,x) \neq a} \trace{ H_h \, A^{u,x}_a} + \trace{(\Id - H)} \\
	&\leq \nu_4 + (1 - \trace{G}) + \nu_7\;.
\end{align*}
Letting $\mu = \kappa + \nu_4 + \nu_7$ where $\kappa = \tau(G)$, we get that $H'$ satisfies the conclusions stated in \Cref{lem:pasting}, once we verify that $\mu$ has the correct asymptotic dependence on $m,t,\eps,\delta,\zeta,n,t$. 

Without loss of generality, we assume $\eps \le 1/(m+1)$, $\delta \le 1$,
$\zeta \le 1$ as otherwise the lemma is trivial.
Recall also that by Proposition~\ref{prop:distance}, $\gamma_m  \le (mt)/n$. 
Under these conditions, we have
\begin{align*}
	\nu_{1} & = 8(\sqrt{\zeta} + \sqrt{(m+1)\delta})= \poly(m) \cdot \poly(\delta, \zeta) \\
	\nu_{2} & = 4(\gamma_m + \nu_{1}) = \poly(m,t) \cdot \poly\left(\delta, \zeta, n^{-1}\right)\\
	\nu_{3} & = t \left( t \nu_{2} + \sqrt{\zeta + \sqrt{2(m+1)\eps}}\right) =
						\poly(m,t) \cdot \poly\left(\eps, \delta, \zeta, n^{-1}\right)\\
	\nu_{4} & = \nu_{3} + \sqrt{2(m+1)\eps} =
						\poly(m,t) \cdot \poly\left(\eps, \delta, \zeta, n^{-1}\right)\\
	\nu_{5} & = \sqrt{\nu_{4}} + \sqrt{\zeta} =
						\poly(m,t) \cdot \poly\left(\eps, \delta, \zeta, n^{-1}\right)\\
	\nu_{6} & = 2 t^{2} \left(v_{2} + n^{-1}\right) =
						\poly(m,t) \cdot \poly\left(\eps, \delta, \zeta, n^{-1}\right)\\
	\nu_{7} & = \nu_{6} + 2{(\nu_{5} + 2\nu_{6})}^{1/t} \\
					& = \poly(m,t) \cdot \poly\left(\eps, \delta, \zeta, n^{-1}\right) +
					 \left(\poly(m,t) \poly\left(\eps,\delta,\zeta,n^{-1}\right) \right)^{\frac{1}{t}}\\
					& = \poly(m,t) \cdot \left( \poly\left(\eps,\delta,\zeta,n^{-1} \right)\right)^{\frac{1}{t}}.
\end{align*}
Therefore,
$\mu = \kappa + \poly(m,t) \cdot \left( \poly\left(\eps, \delta, \zeta, n^{-1}\right) \right)^{1/t}$
as required.

\subsection{Pasting the $G$'s together: Method 2}
\label{sec:pasting-2}

The second construction of $H = \{H_h\}$ 
is designed to circumvent the problem of the first construction,
which is that its soundness has a rather poor dependence on the parameter $t$. 
Before describing the construction, we need the following definitions.

\begin{definition}[$G$'s incomplete part]
For each $x \in [n]$ we write $G^x = \sum_g G^x_g$ and $G^x_{\bot} = \id - G^x$
for the ``complete" and ``incomplete" parts of $G^x$ respectively.
Let $\widehat{G} = \{\widehat{G}^x_g\}_{g\in\code^{\otimes m}\cup\{\bot\}}$
be the projective measurement defined as
\begin{equation*}
\widehat{G}^x_g
= \left\{\begin{array}{rl}
	G^x_g & \text{if } g \in \code^{\otimes m},\\
	G^x_{\bot} & \text{if } g = \bot.
	\end{array}\right.
\end{equation*}
For succinctness we denote $\code^+=\code^{\otimes m}\cup\{\bot\}$, leaving the parameter $m$ implicit. 
\end{definition}

\begin{definition}[Types]
A type $\tau$ is an element of $\{0, 1\}^k$ for some integer~$k$.
We write $|\tau| = \tau_1 + \cdots + \tau_k$ for the Hamming weight of~$\tau$.
We often associate $\tau$ with the set $\{i : \tau_i = 1\}$ and write $i \in \tau$ if $\tau_i = 1$.
\end{definition}

Suppose that the measurement $\wG$ is performed some $k\geq 1$ times in succession,
generating the outcomes $g_1, \ldots, g_k$.
Let us write $\tau \in \{0, 1\}^k$ for the ``type" of these outcomes, where
\begin{equation*}
\tau_i =
\left\{\begin{array}{rl}
	1 & \text{if $g_i \in \code^{\otimes m}$}\;,\\
	0 & \text{if $g_i = \bot$\;.}
	\end{array}
	\right.
\end{equation*}
Whenever $\tau \geq t$, assuming that the $g_i$ are not inconsistent
using Proposition~\ref{prop:tuple-to-code-correspondence} we can interpolate them to produce an~$h\in\code^{\otimes (m+1)}$.
Hence, we would like to understand the probability that $|\tau| \geq t$
and ensure that it is as large as possible.
The probability that $\wG$ returns a $g \in \code^{\otimes m}$
is equal to the completeness of~$G$, which is $1-\kappa$.
This tells us that the probability that $\tau_1 = 1$ is $1-\kappa$.
We might naively expect that the same holds for the other $\tau_i$ as well.
We might also naively expect that the $\tau_i$ are independent.
These two assumptions should not be expected to hold in general,
as they ignore correlations between the measurements.
However, if we make these assumptions, then we at least have a simple toy model
for the measurement outcomes: $\tau \sim \mathrm{Binomial}(k, 1-\kappa)$.

In this toy model,
we expect $|\tau| \approx k \cdot (1-\kappa)$ on average.
This was the limitation of the first construction:
if $k = t$ and $\kappa$ is reasonably large (say, on the order of $1/t$),
then we don't expect $|\tau|$ to be larger than $t$ with high probability,
and so we can't interpolate to produce a global codeword.
This suggests an alternative strategy:
choose $k$ large enough so that $k \cdot (1 - \kappa) \gg t$.
In fact, as we are aiming for $H$ to have completeness close to $1-\kappa$,
we should choose~$k$ so large that $|\tau| \geq t$ with probability 
roughly $1-\kappa$.
On the other hand, if we set $k$ \emph{too} large,
then we increase the risk that the $k$ outcomes $g_1, \ldots, g_k$ are inconsistent with each other,
which is an additional source of error.

This ``naive analysis" motivates our second construction of~$H$,
which is stated below.
We will show that the naive analysis,
in which we treat $\tau$ as a binomial random variable
and bound $|\tau|$ using a Chernoff bound,
can be formalized.

\begin{definition}[The pasted measurement]
Let $k \geq t$ be an integer.
\begin{enumerate}
\item (Pasting):
Let $x_1, \ldots, x_k \in [n]$.
Define initial ``sandwiched" measurement operators as follows:
\begin{equation*}
\widehat{H}^{x_1, \ldots, x_{k}}_{g_1, \ldots, g_{k}} \,=\, \widehat{G}^{x_1}_{g_1} \cdot \widehat{G}^{x_2}_{g_2} \cdots \widehat{G}^{x_{k}}_{g_{k}} \cdots \widehat{G}^{x_2}_{g_2} \cdot \widehat{G}^{x_1}_{g_1}\;.
\end{equation*}
\item (Interpolation):
Let $(x_1, \ldots, x_{k}) \in \distinct([n],k)$.
For any $w \in \{0, 1\}^k$ and $h \in \code^{\otimes (m+1)}$,
 define $h_w=(g_1, \ldots, g_k)  \in (\code^+)^k$
as $g_i = \bot$ if $w_i = 0$ and $g_i = h|_{x_i}$ otherwise.
Define the interpolated measurement
\begin{equation*}
H^{x_1, \ldots, x_{k}}_{h}
\,=\, \sum_{w : |w| \geq t} \widehat{H}^{x_1, \ldots, x_{k}}_{h_w}\;.
\end{equation*}
\item (Averaging):
Define
\begin{equation*}
H_{h} \,=\, \E_{(x_1, \ldots, x_{k}) \sim \distinct([n],k)} H^{x_1, \ldots, x_{k}}_h\;.
\end{equation*}
\end{enumerate}
\end{definition}

To analyze this construction,
we first need to show that $\wG$ commutes with itself.
This is shown in \Cref{sec:hat-consistency},
using that a similar property holds for $G$.
Using this in \Cref{sec:ld-sandwiching} we prove that $\widehat{H}$ is consistent with~$B$,
which we use to prove that $H$ is consistent with~$A$.
Finally, we analyze the completeness of~$H$ in \Cref{sec:completeness-of-h-low-degree}.

\subsubsection{Commutativity of $\widehat{G}$}
\label{sec:hat-consistency}

In this section we show that $\wG$ commutes with itself.
As we already know this holds for the sub-measurement~$G$,
our task essentially reduces to showing that this holds for $G$'s incomplete part,
i.e.\ $G_{\perp}$.
As it is more convenient to work with $G = \id - G_{\perp}$ rather than $G_{\perp}$,
we will first show that these properties hold for~$G$;
the fact that they also hold for $G_{\perp}$ will then follow as an immediate corollary.

\begin{lemma}[Commutativity with~$G_g^x$ implies commutativity with~$G^x$]\label{lem:commutativity-switcheroo}
For all $y \in [n]$ let $M^y = \{M^y_o\}$ be a projective sub-measurement with outcomes in some set~$\calO$ and let $\chi>0$.
Suppose that
\begin{equation}\label{eq:M-commutes-with-G}
G^x_g M^y_o  \,\approx_{\chi}\, M^y_o G^x_g \;,
\end{equation}
on average over independent and uniformly random $x,y\in[n]$. Then
\begin{equation*}
G^x M^y_o \,\approx_{ \sqrt{3\chi}}\, M^y_o G^x\;.
\end{equation*}
\end{lemma}

\begin{proof}
We first show that 
\begin{equation}\label{eq:M-commutes-with-G-1}
 \E_{x,y} \sum_{g,o} \tau\big(G^x_g M^y_o G^{x} M^y_o\big) \,\approx_{{\chi}}\, \E_{x,y} \sum_{g,o} \tau\big(G^x_g M^y_o G^{x}_g M^y_o\big)
\end{equation}
(note that on the left hand side, the third factor is $G^x$ instead of $G^x_g$). To show this, we take the difference to obtain
\begin{equation}\label{eq:M-commutes-with-G-1b}
\E_{x,y} \sum_{g,o} \tau \Big( G^x_g M^y_o \big(G^x - G^x_g \big) M^y_o \Big) = \E_{x,y} \sum_{g,o} \tau \Big( G^x_g G^x_g M^y_o \big(G^x - G^x_g \big) M^y_o \Big)
\end{equation}
where we used the fact that $\{G^x_g\}$ is projective. Notice that 
\begin{equation}\label{eq:M-commutes-with-G-1c}
\E_{x,y} \sum_{g,o} \tau\big(G^x_g M^y_o G^x_g (G^x-G^{x}_{g}) M^y_o\big)\,=\,0
\end{equation}
 using again the fact that $\{G^x_g\}$ is projective. Thus to prove~\eqref{eq:M-commutes-with-G-1} it suffices to show that the quantity in~\eqref{eq:M-commutes-with-G-1b} is $\chi$-close to the quantity in~\eqref{eq:M-commutes-with-G-1c}. Taking the difference, we get
\begin{align}
 \Big | \E_{x,y} \sum_{g,o} &\tau\big(G^x_g(G^x_g M^y_o-M^y_o G^x_g) (G^x-G^{x}_{g}) M^y_o\big) \Big |\notag\\
 &= \Big | \E_{x,y} \sum_{g,o} \tau\big((G^x-G^{x}_{g}) M^y_o G^x_g (G^x_g M^y_o-M^y_o G^x_g) \big) \Big | \notag \\
&\leq \Big(\E_{x,y} \sum_{g,o}\tau\big((G^x-G^{x}_{g}) M^y_o G^x_g M^y_o (G^x-G^{x}_{g}) \big)\Big)^{1/2} \Big(\E_{x,y} \sum_{g,o}\tau\big( (G^x_g M^y_o-M^y_o G^x_g)^2\big)\Big)^{1/2}\notag\\
&\leq {\chi}\;,\label{eq:M-commutes-with-G-2}
\end{align}
where the second line follows from the cyclicity of the trace, the third line uses Cauchy-Schwarz and the fourth bounds the first term by $1$ using $\|G^x-G^x_g\|\leq 1$ for all $x,g$, cyclicity of the trace and $\{G^x\}$, $\{M^y\}$ sub-measurements, and the second term by $\chi$ using~\eqref{eq:M-commutes-with-G}.

To prove the lemma we expand
\begin{align*}
 \E_{x,y} \sum_{o}\tau\big( (G^x M^y_o-M^y_oG^x )^2\big)
&=  \E_{x,y} \sum_{o}\Big( 2 \,\tau\big( G^x M^y_o G^x M^y_o\big) - 2\,\tau\big(  G^x M^y_o \big) \Big)\\
&=  \E_{x,y} \sum_{g,o}\Big( 2 \,\tau\big( G^x_g M^y_o G^x M^y_o\big) - 2\,\tau\big(  G^x_g M^y_o \big) \Big)\\
&\approx_{2\chi} \E_{x,y} \sum_{g,o}\Big( 2\, \tau\big( G^x_g M^y_o G^x_g M^y_o\big) - 2\,\tau\big(  G^x_g M^y_o \big) \Big)\\
&= \E_{x,y} \sum_{g,o} \tau\Big( \big( G^x_g M^y_o- M^y_oG^x_g \big)^2\Big)\\
&\leq \chi^2\;.
\end{align*}
The third line follows from~\eqref{eq:M-commutes-with-G-2}, and the last line follows again from~\eqref{eq:M-commutes-with-G}. 
\end{proof}

\begin{corollary}[Commutativity of~$G$'s]\label{cor:g-comm}
Recall the error parameter $\nu_2$ from \Cref{lem:g-comm}. Then 
\[
	\E_{x,y} \tau \big((G^x G^y - G^x G^y)^2 \big) \leq 9\nu_2^{1/4}~.
\]
\end{corollary}
\begin{proof}
	We prove this via two applications of \Cref{lem:commutativity-switcheroo}. First, by \Cref{lem:g-comm}, we have that
	\[
		G^x_g G^y_h \approx_{\nu_2} G^y_h G^x_g 
	\]
	on average of $x,y \sim [n]$. By \Cref{lem:commutativity-switcheroo}, we get that
	\[
		G^x G^y_h \approx_{3\sqrt{\nu_2}} G^y_h G^x~.
	\]
	By applying \Cref{lem:commutativity-switcheroo} again with the projective sub-measurement $M^x = \{G^x \}$ with a single outcome, we get that
	\[
		G^y G^x \approx_{9 \nu_2^{1/4}} G^x G^y~.
	\]
\end{proof}

\begin{corollary}[Commutativity of~$\wG$]\label{cor:G-hat-facts} Recall the error parameter $\nu_2$ from \Cref{lem:g-comm}. Let $\nu'_2 = 27 \nu_2^{1/4}$. Then 
  \begin{align}
    \widehat{G}^{x}_{g}\widehat{G}^{y}_{h}  &\approx_{\nu'_2} \widehat{G}^{y}_{h}
                               \widehat{G}^{x}_{g} \;. \label{eq:gcomall}
  \end{align}
\end{corollary}

\begin{proof}
We can write
\begin{align}
\E_{x, y} \sum_{g, h \in \code^{\otimes m}} \tau\big( (\widehat{G}^{x}_g \widehat{G}^{y}_h - \widehat{G}^{y}_h \widehat{G}^{x}_g)^2\big)
&=~\E_{x, y} \sum_{g, h \in \code^{\otimes m}} \tau\big((G^{x}_g G^{y}_h - G^{y}_h G^{x}_g)^2\big)	+ \E_{x, y} \tau\big((G^{x}_{\bot} G^{y}_{\bot} - G^{y}_{\bot} G^{x}_{\bot})^2\big) \notag \\
&  + \E_{x, y} \sum_{g\in \code^{\otimes m}} \tau\big( (G^{x}_g G^{y}_{\bot} - G^{y}_{\bot} G^{x}_g)^2\big)
	+ \E_{x, y} \sum_{h\in \code^{\otimes m}} \tau\big( (G^{x}_{\bot} G^{y}_h - G^{y}_h G^{x}_{\bot})^2 \big)~. \label{eq:cor:G-hat-facts-1}
\end{align}
We can bound the first term by $\nu_2^2$ by \Cref{lem:g-comm}. We can bound the second term by first writing
\begin{align*}
\E_{x, y} \tau\big((G^{x}_{\bot} G^{y}_{\bot} - G^{y}_{\bot} G^{x}_{\bot})^2\big) &= 
\E_{x, y} \tau\big(((\id-G^x)  (\id-G^y) - (\id-G^y)  (\id-G^x))^2\big) \\
&= \E_{x, y} \tau\big((G^{x} G^{y} - G^{y} G^{x})^2\big) \\ 
&\leq 9^2 \cdot \nu_2^{1/2}
\end{align*}
where the last line follows from \Cref{cor:g-comm}. Similarly, we can bound the third and fourth terms of~\eqref{eq:cor:G-hat-facts-1} by
\begin{align*}
\E_{x, y} \sum_{g} \tau\big( (G^{x}_g G^{y}_{\bot} - G^{y}_{\bot} G^{x}_g)^2\big) &= \E_{x, y} \sum_{g} \tau\big( (G^{x}_g ( \id - G^{y}) - ( \id - G^{y}) G^{x}_g)^2\big) \\
&= \E_{x, y} \sum_{g} \tau\big( (G^{x}_g G^{y} - G^{y} G^{x}_g)^2\big) \\
&\leq 9\nu_2
\end{align*}
where the last line follows from \Cref{lem:g-comm} and \Cref{lem:commutativity-switcheroo}. 
%
%

Putting everything together, this implies that we can upper bound~\eqref{eq:cor:G-hat-facts-1} by $\nu_2^2 + 9^2 \nu_2^{1/2} + 18 \nu_2 \leq 3 \cdot 9^2 \cdot \nu_2^{1/2}$. The lemma follows.  
\end{proof}

\subsubsection{Consistency of $\widehat{H}$ and $H$ with~$B$ and $A$}
\label{sec:ld-sandwiching}

We use the same notation as in Section~\ref{sec:h-bu}. The following lemma is analogous to Lemma~\ref{lem:pasting-h-cons-b}. 

\begin{lemma}[Consistency of $\widehat{H}$ with $B$]
\label{lem:ld-sandwich-line-one-point}
Let $\nu''_3= k\cdot\nu'_2+(\zeta+(2(m+1)\eps)^{1/2})^{1/2}$. For any $1 \leq i \leq k$,
\begin{equation}
 \E_{u}  \E_{x_1, \dots, x_k} \sum_{a \neq \bot} \sum_{b\neq a}  \tau\big( \widehat{H}^{x_1, \dots, x_k}_{[(g_1, \dots, g_k)\mapsto g_i(u) \mid a]}  B^u_{[f\mapsto f(x_i)\mid b]}\big)\,\leq\,\nu''_3
\end{equation}
where the expectation over $(x_1,\ldots,x_k)$ is over independently random $x_1,\ldots,x_k \sim [n]$. 
\end{lemma}

\begin{proof}
The proof follows closely that of Lemma~\ref{lem:pasting-h-cons-b}.
We have
\begin{alignat}{2}
		& && \E_{u} \E_{x_1,\ldots,x_k}
		\sum_{g_1, \dots, g_k: g_i \neq \bot}\sum_{f: f(x_i)\neq g_i(u)}
		\trace{\widehat{H}_{g_1,\ldots,g_k}^{x_1,\ldots,x_k} \, B^u_f } \notag\\
		& = && \E_{u}\E_{x_1,\ldots,x_k }
		\sum_{g_1, \dots, g_k: g_i \neq \bot}\sum_{f: f(x_i)\neq g_i(u)}
		\trace{ \widehat{G}^{x_1}_{g_1} \cdots \widehat{G}^{x_k}_{g_k} \cdots
			\widehat{G}^{x_1}_{g_1} \cdot B^u_f } \notag \\
	  & = && \E_{u} \E_{x_1,\ldots,x_i}
		\sum_{g_1,\ldots,g_i: g_i \neq \bot}
		\trace{B^u_{[f : f(x_i) \neq g_i(u)]} \cdot \widehat{G}^{x_1}_{g_1}
			\cdots \widehat{G}^{x_i}_{g_i} \cdot \Bigl( \widehat{G}^{x_{i}}_{g_{i}} \cdot
			\widehat{G}^{x_{i-1}}_{g_{i-1}} \cdots \widehat{G}^{x_1}_{g_1} \Bigr) } \notag \\
	  & \approx_{i \cdot \nu'_{2} \;} && \E_{u} \E_{x_1,\ldots,x_i}
		\sum_{g_1,\ldots,g_i} \trace{ B^u_{[f : f(x_i) \neq g_i(u)]} \cdot
			\widehat{G}^{x_1}_{g_1} \cdots \widehat{G}^{x_i}_{g_i} \cdot \Bigl( \widehat{G}^{x_{i-1}}_{g_{i-1}}
			\cdots \widehat{G}^{x_1}_{g_1} \cdot \widehat{G}^{x_i}_{g_i} \Bigr) }
		\label{eq:pasting-h-cons-b-12}
	\end{alignat}
	where the last approximation follows from
	\cref{lem:closeness-to-close-ips,lem:switcheroo} and the approximate
	commutativity of the $\widehat{G}^{x}$ measurements (\Cref{cor:G-hat-facts}).
	Continuing, we can use Cauchy-Schwarz to bound~\eqref{eq:pasting-h-cons-b-12} by
	\begin{align*}
		\text{\eqref{eq:pasting-h-cons-b-12}}
		& \leq \sqrt{ \E_{x_1,\ldots,x_i } \sum_{g_1,\ldots,g_i: g_i \neq \bot}
		 \trace{ \widehat{G}^{x_1}_{g_1} \cdots \widehat{G}^{x_i}_{g_i} \cdots \widehat{G}^{x_1}_{g_1}}} \\
		& \qquad \qquad \cdot \sqrt{ \E_{u} \E_{x_1,\ldots,x_i}
		 \sum_{g_1,\ldots,g_i: g_i\neq \bot} \trace{ \widehat{G}^{x_1}_{g_1} \cdots \widehat{G}^{x_{i-1}}_{g_{i-1}}
		 \cdots \widehat{G}^{x_1}_{g_1} \cdot \widehat{G}^{x_i}_{g_i} \cdot
		 {\bigl( B^u_{[f : f(x_i) \neq g_i(u)]} \bigr)}^2 \cdot \widehat{G}^{x_i}_{g_i}}} \\
		& \leq \sqrt{1} \cdot \sqrt{ \E_{u,x_i} \sum_{g_i:g_i\neq \bot}
		 \trace{ \widehat{G}^{x_i}_{g_i} \, B^u_{[f : f(x_i) \neq g_i(u)]}}}\;,
	\end{align*}
	where in the last inequality we used the fact that $\widehat{G}^{x}$ is projective
	and
	${\bigl( B^u_{[f : f(x_i) \neq g_i(u)]} \bigr)}^{2} \le B^u_{[f : f(x_i) \neq g_i(u)]}$
	for all $u$.
	Notice that
	\begin{align*}
	  \E_{u,x_i} \sum_{g_i:g_i\neq \bot} \trace{ \widehat{G}^{x_i}_{g_i} B^u_{[f : f(x_i) \neq g_i(u)]}}
		& = \E_{u,x}  \sum_{a,b:a\ne b} \trace{ G^{u,x}_{a} B^{u,x}_{b} } \\
	  & \approx_{ \sqrt{2(m+1) \eps} } \E_{u,x} \sum_{a,b:a \ne b}
			\trace{ G^{u,x}_{a} A^{u,x}_b } \\
	  & \leq \zeta \;.
	\end{align*}
where the approximation in the second line follows from \cref{lem:transfer-cons} and \cref{eq:pasting-b-close-a}.\end{proof}

Next we move from $\widehat{H}$ to $H$. 

\begin{lemma}[Consistency of~$H$ with~$B$]\label{lem:h-b-consistency}
Let $\nu'_3 = k\nu''_3 + k^2/n$\;. Then
\begin{equation*}
H_{[h|_u =f]} \simeq_{\nu'_3}  B^u_f\;.
\end{equation*}
\end{lemma}

\begin{proof}
We have
\begin{align}
\E_{u} \sum_{f \neq f'}  &\tau\big(H_{[h\mapsto h|_{u} \mid f']}  B^{u}_f \big)\\
&=\E_{u} \sum_{h}\sum_{f \neq h|_{u}} \tau\big( H_{h} B^{u}_f \big)\nonumber\\
&=\E_{u}\E_{(x_1, \ldots, x_k) \sim \distinct([n],k)} \sum_{h}\sum_{f \neq h|_{u}} \tau\big( H^{x_1, \ldots, x_k}_{h} B^{u}_f \big)\nonumber\\
&=\E_{u}\E_{(x_1, \ldots, x_k) \sim \distinct([n],k)} \sum_{h} \sum_{w: |w| \geq t} \sum_{f \neq h|_{u}} \tau\big(\widehat{H}^{x_1, \ldots, x_k}_{g_1,\ldots,g_k}  B^{u}_f \big) \label{eq:keep-on-expandin}
\end{align}
where $g_i \in \code^+$ is defined to be $\bot$ if $w_i = 0$, and otherwise $g_i = h|_{x_i}$. 
Because $|w| \geq t$,
there exist at least $t$ coordinates~$i$ such that $g_i \neq \bot$
and hence $g_i = h|_{x_i}$. Since $t\geq n-d+1$, using the distance property of $\code$ it follows that if $f\neq h|_{u}$
then there must exist an $i$ such that $g_i \neq \bot$ and $g_i(u) \neq f(x_i)$.
Thus
\begin{align}
\eqref{eq:keep-on-expandin}
~=~& \E_{u} \E_{(x_1, \ldots, x_k) \sim \distinct([n],k)} \sum_{h} \sum_{w: |w| \geq t} \sum_{\substack{f: \exists i: g_i \neq \bot, \\g_i(u) \neq f(x_i)} }\tau\big(\widehat{H}^{x_1, \ldots, x_k}_{g_1, \ldots, g_k}  B^{u}_f \big) \nonumber\\
\leq~&\E_{u} \E_{(x_1, \ldots, x_k) \sim \distinct([n],k)} \sum_{g_1, \ldots, g_k} \sum_{\substack{f: \exists i: g_i \neq \bot, \\g_i(u) \neq f(x_i)}} \tau\big( \widehat{H}^{x_1, \ldots, x_k}_{g_1, \ldots, g_k}  B^{u}_f \big)\;.\label{eq:keep-on-contractin}
\end{align}
Using \Cref{prop:ld-dnoteq} (which allows us to switch from sampling $(x_1,\ldots,x_k) \sim \distinct([n],k)$ to sampling $x_1,\ldots,x_k$ independently from $[n]$) and the fact that $\tau\big( \widehat{H}^{x_1, \ldots, x_k}_{g_1, \ldots, g_k}  B^{u}_f \big)$ is nonnegative and, when summed over $g_1,\ldots,g_k$ and $f$, is at most $1$, we have that \Cref{eq:keep-on-contractin}
is $(k^2/n)$-close to
\begin{align*}
&\E_{u}\E_{x_1, \ldots, x_k} \sum_{g_1, \ldots, g_k} \sum_{\substack{f: \exists i: g_i \neq \bot, \\g_i(u) \neq f(x_i)}} \tau\big(\widehat{H}^{x_1, \ldots, x_k}_{g_1, \ldots, g_k}  B^{u}_f \big)\\
\leq~&
\sum_i \E_{u}\E_{x_1, \ldots, x_k} \sum_{g_1, \ldots, g_k} \sum_{\substack{f: g_i \neq \bot, \\g_i(u) \neq f(x_i)}} \tau\big( \widehat{H}^{x_1, \ldots, x_k}_{g_1, \ldots, g_k}  B^{u}_f \big)\tag{by the union bound}\\
=~&
\sum_i \E_{u}\E_{x_1, \ldots, x_k} \sum_{g_1, \ldots, g_k:g_i \neq \bot} \sum_{b \neq g_i(u)} \tau\big( \widehat{H}^{x_1, \ldots, x_k}_{g_1, \ldots, g_k}  B^{u}_{[f\mapsto f(x_i)\mid b]} \big)\\
\leq~& \sum_i \nu''_3 \tag{by \Cref{lem:ld-sandwich-line-one-point}}\\
=~& k \cdot \nu''_3\;.
\end{align*}
This establishes the lemma.
\end{proof}

\begin{corollary}[Consistency of~$H$ with~$A$]\label{cor:ha-cons}
Let $\nu'_4 = \nu'_3 + \sqrt{2 (m+1) \eps}$. Then on average over $(u,x) \sim [n]^m \times [n]$, 
\begin{equation*}
H_{[h\mapsto h(u, x)\mid a]}  \simeq_{\nu'_4} A^{u,x}_a\;.
\end{equation*}
\end{corollary}
\begin{proof}
Applying \Cref{lem:data-processing} to \Cref{lem:h-b-consistency} we get that
\begin{equation}\label{eq:h-b-consistency-at-a-point}
H_{[h\mapsto h(u, x)\mid a]}  \simeq_{\nu'_3}  B^{u}_{[f \mapsto f(x) \mid a]}\;.
\end{equation}
Using~\eqref{eq:pasting-b-close-a} together with \Cref{lem:transfer-cons} (with measurement ``$A^x_a$'' in the Proposition being set to $B^u_{[f \mapsto f(x) \mid a]}$, measurement ``$B^x_a$'' being set to $A^x_a$ and the submeasurement ``$C^x_a$'' being set to $H_{[h\mapsto h(u, x)\mid a]}$) proves the lemma.
\end{proof}

\subsubsection{Completeness of~$H$}
\label{sec:completeness-of-h-low-degree}

\begin{definition}
Let $\tau \in \{0, 1\}^k$ be a type.
We define the following two subsets of $(\code^+)^k$.
\begin{itemize}
\item
We define $\mathsf{Outcomes}_{\tau}$ to be the set of tuples $(g_1, \ldots, g_k)$ such that $g_i  \in \code^{\otimes m}$ for each $i \in \tau$ and $g_i = \bot$ for each $i \notin \tau$.
This is the set of possible outcomes of the measurement $\widehat{H}$  of type~$\tau$.
\item
Let $x_1, \ldots, x_k \in [n]$.
We define $\mathsf{Global}_{\tau}(x)$ to be the subset of $\mathsf{Outcomes}_{\tau}$ containing only those tuples which are consistent with a global codeword. In other words, $(g_1, \ldots, g_k) \in \mathsf{Global}_{\tau}(x)$ iff there exists an $h \in \code^{\otimes(m+1)}$ such that $g_i = h|_{x_i}$ for each $i \in \tau$ (recall that $h|_{x_i}$ denotes the function $h(\cdots,x_i) \in \code^{\otimes m}$). 
Define
\begin{equation*}
\overline{\mathsf{Global}_{\tau}(x)} = \mathsf{Outcomes}_{\tau} \setminus \mathsf{Global}_{\tau}(x).
\end{equation*}
This contains those tuples of type~$\tau$ with no consistent global codeword.
\end{itemize}
\end{definition}

\begin{lemma}
\label{lem:over-all-outcomes}
Let $\nu'_5 = 2\frac{k^2}{n}  + k \cdot\nu''_3 + \gamma_m$\;. Let $x_1, \ldots, x_k \sim [n]$ be sampled uniformly at random.
Then
\begin{equation*}
\tau\big( H \big)
\,\approx_{\nu'_5}\,  \E_{x_1, \ldots, x_k} \sum_{\tau:|\tau| \geq t} \sum_{(g_1, \ldots, g_k) \in \mathsf{Outcomes}_{\tau}} \tau\big( \widehat{H}^{x_1, \ldots, x_k}_{g_1, \ldots, g_k} \big)\;.
\end{equation*}
\end{lemma}

\begin{proof}
By definition,
\begin{align}
\tau\big( H \big)
&= \E_{(y_1, \ldots, y_k) \sim \distinct([n],k)}  \sum_{\tau:|\tau| \geq t} \sum_{(g_1, \ldots, g_k) \in \mathsf{Global}_{\tau}(y)} \tau\big( \widehat{H}^{y_1, \ldots, y_k}_{g_1, \ldots, g_k} \big)\;. \label{eq:sum-restricted-to-global-polynomial}
\end{align}
The sum in~\eqref{eq:sum-restricted-to-global-polynomial}
is over $g_1, \ldots, g_k$ which are consistent with a global codeword.
We show that the value of this sum
remains largely unchanged if we drop this condition.
In particular, we claim that
\begin{equation}\label{eq:remove-the-restriction}
\eqref{eq:sum-restricted-to-global-polynomial}
\approx_{\frac{k^2}{n} + k \cdot\nu''_3 + \gamma_m}
\E_{(y_1, \ldots, y_k) \sim \distinct([n],k)}  \sum_{\tau:|\tau| \geq t} \sum_{(g_1, \ldots, g_k) \in \mathsf{Outcomes}_{\tau}} \tau\big( \widehat{H}^{y_1, \ldots, y_k}_{g_1, \ldots, g_k} \big)\;.
\end{equation}
To show this, we note that $\eqref{eq:remove-the-restriction} \geq \eqref{eq:sum-restricted-to-global-polynomial}$. Hence, it suffices to upper bound their difference.
\begin{align}
\eqref{eq:remove-the-restriction} - \eqref{eq:sum-restricted-to-global-polynomial}
&= \E_{(y_1, \ldots, y_k) \sim \distinct([n],k)} \sum_{\tau:|\tau| \geq t} \sum_{(g_1, \ldots, g_k) \in \overline{\mathsf{Global}_{\tau}(y)}} \tau\big( \widehat{H}^{y_1, \ldots, y_k}_{g_1, \ldots, g_k} \big)\nonumber\\
&=  \E_{u} \E_{(y_1, \ldots, y_k) \sim \distinct([n],k)}  \sum_{\tau:|\tau| \geq t} \sum_f \sum_{(g_1, \ldots, g_k) \in \overline{\mathsf{Global}_{\tau}(y)}}  \tau\big(\widehat{H}^{y_1, \ldots, y_k}_{g_1, \ldots, g_k}  B^{u}_f \big)\;,\label{eq:B-appears-out-of-thin-air}
\end{align}
where in the second line we used that~$B$ is a measurement. Next, we claim that
\begin{equation}\label{eq:add-in-indicator-for-f-and-g}
\eqref{eq:B-appears-out-of-thin-air}
\approx_{\frac{k^2}{n} + k \cdot\nu''_3} 
\E_{u} \E_{\substack{(y_1, \ldots, y_k) \\ \quad \sim \distinct([n],k)}}   \sum_{\tau:|\tau| \geq t} \sum_f \sum_{(g_1, \ldots, g_k) \in \overline{\mathsf{Global}_{\tau}(y)}} \tau\big(\widehat{H}^{y_1, \ldots, y_k}_{g_1, \ldots, g_k} B^{u}_f \big) \cdot \indicator[\forall i\in \tau, f(y_i) = g_i(u)].
\end{equation}
To show this, we note that $\eqref{eq:B-appears-out-of-thin-air} \geq \eqref{eq:add-in-indicator-for-f-and-g}$.
Thus, it suffices to upper bound their difference $\eqref{eq:B-appears-out-of-thin-air} - \eqref{eq:add-in-indicator-for-f-and-g}$. This is given by
\begin{equation}\label{eq:about-to-swap-x-for-y}
\E_{u} \E_{(y_1, \ldots, y_k) \sim \distinct([n],k)}  \sum_{\tau: |\tau|\geq t} \sum_f  \sum_{(g_1, \ldots, g_k) \in \overline{\mathsf{Global}_{\tau}(y)}}  \tau\big( \widehat{H}^{y_1, \ldots, y_k}_{g_1, \ldots, g_k}  B^{u}_f \big) \cdot \indicator[\exists i \in \tau, f(y_i) \neq g_i(u)]\;.
\end{equation}
Recall that $x_1, \ldots, x_k \sim [n]$ are sampled independently and uniformly at random.
Using~\Cref{prop:ld-dnoteq},~\Cref{eq:about-to-swap-x-for-y} is $(k^2/n)$-close to
\begin{align*}
 &\E_{u} \E_{x_1, \ldots, x_k}  \sum_{\tau: |\tau|\geq t} \sum_f  \sum_{(g_1, \ldots, g_k) \in \overline{\mathsf{Global}_{\tau}(x)}}  \tau\big( \widehat{H}^{x_1, \ldots, x_k}_{g_1, \ldots, g_k}  B^{u}_f\big) \cdot \indicator[\exists i \in \tau, f(x_i) \neq g_i(u)]\\
 \leq~& \E_{u} \E_{x_1, \ldots, x_k}   \sum_{\tau:|\tau| \geq t} \sum_f \sum_{(g_1, \ldots, g_k) \in \mathsf{Outcomes}_{\tau}}  \tau\big( \widehat{H}^{x_1, \ldots, x_k}_{g_1, \ldots, g_k}  B^{u}_f \big) \cdot \indicator[\exists i \in \tau, f(x_i) \neq g_i(u)]\\
 \leq~& \E_{u} \E_{x_1, \ldots, x_k}   \sum_{\tau} \sum_f \sum_{(g_1, \ldots, g_k) \in \mathsf{Outcomes}_{\tau}}  \tau\big(\widehat{H}^{x_1, \ldots, x_k}_{g_1, \ldots, g_k}  B^{u}_f \big) \cdot \indicator[\exists i \in \tau, f(x_i) \neq g_i(u)]\\
 \leq~&  \E_{u} \E_{x_1, \ldots, x_k}  \sum_{\tau} \sum_f \sum_{(g_1, \ldots, g_k) \in \mathsf{Outcomes}_{\tau}}  \tau\big( \widehat{H}^{x_1, \ldots, x_k}_{g_1, \ldots, g_k}  B^{u}_f \big)\cdot \Big(\sum_{i\in \tau} \indicator[f(x_i) \neq g_i(u)]\Big)\\
 =~& \sum_i  \E_{u} \E_{x_1, \ldots, x_k}  \sum_{g_1, \ldots, g_k : g_i \neq \bot} \sum_{f: f(x_i) \neq g_i(u)}  \tau\big( \widehat{H}^{x_1, \ldots, x_k}_{g_1, \ldots, g_k}  B^{u}_f \big)\\
 =~& \sum_i  \E_{u} \E_{x_1, \ldots, x_k}  \sum_{g_1, \ldots, g_k : g_i \neq \bot} \sum_{a \neq g_i(u)}  \tau\big( \widehat{H}^{x_1, \ldots, x_k}_{g_1, \ldots, g_k}  B^{u}_{[f(x_i)=a]} \big)\\
 \leq~& \sum_i \nu''_3 \tag{by \Cref{lem:ld-sandwich-line-one-point}}\\
 =~& k \cdot \nu''_3.
\end{align*}

Returning to \Cref{eq:add-in-indicator-for-f-and-g},
we introduce the  notation
\begin{equation*}
\mathsf{Consistent}_{\tau}(g, y, u)
= \left\{\begin{array}{rl}
	1 & \text{if $\exists f$ such that $f(y_i) = g_i(u)$ for all $i \in \tau$\;,}\\
	0 & \text{otherwise\;.}
	\end{array}\right.
\end{equation*}
Clearly, for any $f$,
\begin{equation*}
\indicator[\forall i\in \tau, f(y_i) = g_i(u)] \leq \indicator[\mathsf{Consistent}_{\tau}(g, y, u)]\;.
\end{equation*}
As a result,
\begin{align}
\eqref{eq:add-in-indicator-for-f-and-g}
\leq~&
\E_{u}  \E_{\substack{(y_1, \ldots, y_k) \\ \quad \sim \distinct([n],k)}} \sum_{\tau:|\tau|\geq t} \sum_f  \sum_{(g_1, \ldots, g_k) \in \overline{\mathsf{Global}_{\tau}(y)}}  \tau\big( \widehat{H}^{y_1, \ldots, y_k}_{g_1, \ldots, g_k}  B^{u}_f \big) \cdot \indicator[\mathsf{Consistent}_{\tau}(g,y,u)]\nonumber\\
=~&
\E_{u}  \E_{\substack{(y_1, \ldots, y_k) \\ \quad \sim \distinct([n],k)}}  \sum_{\tau:|\tau|\geq t} \sum_{(g_1, \ldots, g_k) \in \overline{\mathsf{Global}_{\tau}(y)}}  \tau\Big( \widehat{H}^{y_1, \ldots, y_k}_{g_1, \ldots, g_k}  \Big(\sum_f B^{u}_f\Big) \Big) \cdot \indicator[\mathsf{Consistent}_{\tau}(g,y,u)]\nonumber\\
=~&
\E_{u}  \E_{\substack{(y_1, \ldots, y_k) \\ \quad \sim \distinct([n],k)}}    \sum_{\tau:|\tau|\geq t} \sum_{(g_1, \ldots, g_k)  \in \overline{\mathsf{Global}_{\tau}(y)}}  \tau\big(\widehat{H}^{y_1, \ldots, y_k}_{g_1, \ldots, g_k} \big) \cdot \indicator[\mathsf{Consistent}_{\tau}(g,y,u)]\nonumber\\
=~&
  \E_{\substack{(y_1, \ldots, y_k) \\ \quad \sim \distinct([n],k)}}   \sum_{\tau:|\tau|\geq t} \sum_{(g_1, \ldots, g_k) \in \overline{\mathsf{Global}_{\tau}(y)}}  \tau\big( \widehat{H}^{y_1, \ldots, y_k}_{g_1, \ldots, g_k}\big) \cdot \Pr_{u} \Big( \mathsf{Consistent}_{\tau}(g,y,u) \Big).\label{eq:consistent-indicator}
\end{align}

Let us now fix $(y_1, \ldots, y_k) \in \distinct([n],k)$,
a type $\tau$ such that $|\tau| \geq t$,
and $(g_1, \ldots, g_k) \in \overline{\mathsf{Global}_{\tau}(y)}$.
Suppose without loss of generality that $\tau_1 = \cdots = \tau_{t} = 1$,
so that $g_1, \ldots, g_{t} \in \code^{\otimes m}$.
Then by~\Cref{prop:interpolate-tuple} there is a unique $h^* \in \code^{\otimes (m+1)}$ which interpolates $g_1, \ldots, g_{t}$.
In other words, for all $1 \leq i \leq t$ and for all $u \in [n]^m$,
\begin{equation*}
h^*(u, y_i) = g_i(u).
\end{equation*}
On the other hand, because $(g_1, \ldots, g_k) \in \overline{\mathsf{Global}_{\tau}(y)}$, 
there exists an $i^* \in \tau$ such that $g_{i^*} \neq (h^*)|_{y_{i^*}}$.
Hence,
\begin{equation*}
\Pr_{u} \Big(\mathsf{Consistent}_{\tau}(g, y,u)\Big)
\leq \Pr_{u}\Big(g_{i^*}(u) = h^*(u, y_{i^*}) \Big)
= \Pr_{u} \Big( g_{i^*}(u) = (h^*)|_{y_{i^*}}(u)\Big)
\leq \Big(1-\frac{d^m}{n^m}\Big)\,=\,\gamma_m\;,
\end{equation*}
by Proposition~\ref{prop:distance}.
As a result,
\begin{align*}
\eqref{eq:consistent-indicator}
& \leq 
 \E_{\substack{(y_1, \ldots, y_k) \\ \quad \sim \distinct([n],k)}}    \sum_{\tau:|\tau|\geq t} \sum_{(g_1, \ldots, g_k) \in \overline{\mathsf{Global}_{\tau}(y)}} \tau\big(\widehat{H}^{y_1, \ldots, y_k}_{g_1, \ldots, g_k} \big)\cdot \gamma_m\\
 &\leq
 \E_{\substack{(y_1, \ldots, y_k) \\ \quad \sim \distinct([n],k)}}    \sum_{g_1, \ldots, g_k} \tau\big( \widehat{H}^{y_1, \ldots, y_k}_{g_1, \ldots, g_k} \big) \cdot \gamma_m\\
  & = \gamma_m\;.
\end{align*}
This establishes \Cref{eq:remove-the-restriction}.
Finally,  \Cref{eq:remove-the-restriction} is $(k^2/n)$-close to
\begin{equation*}
\E_{x_1, \ldots, x_k}  \sum_{\tau:|\tau| \geq t} \sum_{(g_1, \ldots, g_k) \in \mathsf{Outcomes}_{\tau}} \tau\big( \widehat{H}^{x_1, \ldots, x_k}_{g_1, \ldots, g_k}\big)\;,
\end{equation*}
as desired.
\end{proof}

\begin{lemma}
\label{lem:from-H-to-G}
Recall the error parameter $\nu'_2$ from \Cref{cor:G-hat-facts}. Let $\nu'_6 = 2k^2\nu'_2$.
Let $x_1, \ldots, x_k \sim [n]$ be sampled uniformly at random.
Then
\begin{equation*}
\E_{x_1, \dots, x_k} \sum_{\tau:|\tau| \geq t} 
	\sum_{(g_1, \ldots, g_{k}) \in \outc_{\tau}} \tau\big(\widehat{H}^{x_1, \dots, x_k}_{g_1, \dots, g_k} \big)
\approx_{\nu'_6}
\sum_{i = t}^k \binom{k}{i} \tau\big(G^i (\id - G)^{k-i}\big)\;.
\end{equation*}
\end{lemma}
\begin{proof}
We begin by introducing some notation that we will use throughout the proof.
Let $\tau \in \{0, 1\}^k$ be a type.
Then we define
\begin{equation*}
\tau_{< \ell}  = (\tau_1, \ldots, \tau_{\ell-1}) \in \{0, 1\}^{\ell-1}\;,
\quad
\tau_{> \ell}  = (\tau_{\ell+1}, \ldots, \tau_{k}) \in \{0,1\}^{k-\ell}\;,
\end{equation*}
and we define $\tau_{\leq \ell}$ and $\tau_{\geq \ell}$ similarly.
In addition, given $(g_1, \ldots, g_k) \in (\code^+)^{k}$, we define
\begin{equation*}
g_{< \ell}  = (g_1, \ldots, g_{\ell-1})\;,
\qquad
g_{> \ell}  = (g_{\ell+1}, \ldots, g_{k})
\end{equation*}
and we define $g_{\leq \ell}$ and $g_{\geq \ell}$ similarly.
Using this notation, we can write
\begin{equation*}
 \wH^{x_{\geq \ell}}_{g_{\geq \ell}} = \wH^{x_{\ell}, \ldots, x_k}_{g_{\ell}, \ldots, g_k}\;.
\end{equation*}
Next, we introduce the notation
\begin{equation*}
\wG^{x_{\geq \ell}}_{g_{\geq \ell}} = \wG^{x_{\ell}}_{g_{\ell}} \cdots \wG^{x_k}_{g_k}\;.
\end{equation*}
This satisfies the recurrence relation
\begin{equation}\label{eq:G-recurrence}
\wG^{x_{\geq \ell}}_{g_{\geq \ell}} = \wG^{x_{\ell}}_{g_{\ell}} \cdot \wG^{x_{> \ell}}_{g_{> \ell}}\;.
\end{equation}
Furthermore, we can write
\begin{equation}\label{eq:split-H-into-two-Gs}
 \wH^{x_{\geq \ell}}_{g_{\geq \ell}} = (\wG^{x_{\geq \ell}}_{g_{\geq \ell}}) \cdot (\wG^{x_{\geq \ell}}_{g_{\geq \ell}})^\dagger.
\end{equation}

To prove the lemma we show that for each $1 \leq \ell \leq k$,
\begin{align}
&\E_{x_{\geq \ell}} \sum_{\tau:|\tau| \geq t}  \sum_{g_{\geq \ell} \in \outc_{\tau_{\geq \ell}}}
	\tau\big(\widehat{H}^{x_{\geq \ell}}_{g_{\geq \ell}} \cdot (G^{|\tau_{<\ell}|}\cdot (\id - G)^{(\ell-1)-|\tau_{<\ell}|}) \big)\nonumber\\
\approx_{2k\nu'_2}~&\E_{x_{>\ell}} \sum_{\tau:|\tau| \geq t} \sum_{g_{> \ell} \in \outc_{\tau_{> \ell}}}
	\tau\big(\widehat{H}^{x_{> \ell}}_{g_{> \ell}} \cdot (G^{|\tau_{\leq \ell}|} \cdot (\id - G)^{\ell-|\tau_{\leq\ell}|}) \big)\;.
	\label{eq:i-think-this-is-what-i'm-supposed-to-prove}
\end{align}
If we then repeatedly apply \Cref{eq:i-think-this-is-what-i'm-supposed-to-prove} for $\ell = 1, \ldots, k$,
we derive
\begin{align*}
&\E_{x_1, \dots, x_k} \sum_{\tau:|\tau| \geq t} 
	\sum_{(g_1, \ldots, g_{k}) \in \outc_{\tau}} \tau\big(\widehat{H}^{x_1, \dots, x_k}_{g_1, \dots, g_k} \big)\\
=~&\E_{x_{\geq 1}} \sum_{\tau:|\tau| \geq t} 
	\sum_{g_{\geq 1} \in \outc_{\tau}} \tau\big(\widehat{H}^{x_{\geq 1}}_{g_{\geq 1}} \big)\\
\approx_{2k\nu'_2}~&\E_{x_{\geq 2}} \sum_{\tau:|\tau| \geq t} 
	\sum_{g_{\geq 2} \in \outc_{\tau_{\geq 2}}} \tau\big(\widehat{H}^{x_{\geq 2}}_{g_{\geq 2}} \cdot (G^{|\tau_{\leq 1}|} \cdot (\id - G)^{1-|\tau_{\leq1}|}) \big)\\
\approx_{2k\nu'_2}~&\E_{x_{\geq 3}} \sum_{\tau:|\tau| \geq d+1} 
	\sum_{g_{\geq 3} \in \outc_{\tau_{\geq 3}}} \tau\big(\widehat{H}^{x_{\geq 3}}_{g_{\geq 3}} \cdot (G^{|\tau_{\leq 2}|} \cdot (\id - G)^{2-|\tau_{\leq2}|}) \big)\\
	\cdots&\\
\approx_{2k\nu'_2}~& \sum_{\tau:|\tau| \geq t} 
	 \tau \big(G^{|\tau|} \cdot (\id - G)^{k-|\tau|} \big)\\
=~&\sum_{i = t}^k \binom{k}{i} \tau\big(G^i (\id - G)^{k-i}  \big)\;.
\end{align*}
This proves the lemma.

We now prove \Cref{eq:i-think-this-is-what-i'm-supposed-to-prove}.
To begin, for each $1 \leq \ell \leq k+1$ and $\tau_{\geq \ell} \in \{0, 1\}^{k - \ell + 1}$,
we define the operator
\begin{equation}\label{eq:S-def}
S_{\tau_{\geq \ell}} = \sum_{\tau_{< \ell} : |\tau| \geq t} G^{|\tau_{<\ell}|}\cdot (\id - G)^{(\ell-1)-|\tau_{<\ell}|}\;.
\end{equation}
Then the statement in \Cref{eq:i-think-this-is-what-i'm-supposed-to-prove} can be rewritten as
\begin{align}
&\E_{x_{\geq \ell}} \sum_{\tau_{\geq \ell}}  \sum_{g_{\geq \ell} \in \outc_{\tau_{\geq \ell}}}
	\tau\big(\widehat{H}^{x_{\geq \ell}}_{g_{\geq \ell}} \cdot S_{\tau_{\geq \ell}}\big)\nonumber\\
\approx_{2k\nu'_2}~&\E_{x_{> \ell}} \sum_{\tau_{>\ell}} \sum_{g_{>\ell} \in \outc_{\tau_{> \ell}}}
	\tau\big(\widehat{H}^{x_{>\ell}}_{g_{>\ell}} \cdot S_{\tau_{>\ell}} \big)\;.
	\label{eq:i-think-this-is-what-i'm-supposed-to-prove-2}
\end{align}
To prove this, we will use several facts about $S_{\tau_{\geq \ell}}$.
First, $S$ is self-adjoint and positive semidefinite.
This is because each term in \Cref{eq:S-def}
is a product of~$G$ and $(\id - G)$.
These operators commute with each other,
and both are self-adjoint and positive semidefinite.
Next, $S$ is bounded:
\begin{align}
S_{\tau_{\geq \ell}}
&= \sum_{\tau_{< \ell} : |\tau| \geq t} G^{|\tau_{<\ell}|}\cdot (\id - G)^{(\ell-1)-|\tau_{<\ell}|}\nonumber\\
&\leq \sum_{\tau_{< \ell} } G^{|\tau_{<\ell}|}\cdot (\id - G)^{(\ell-1)-|\tau_{<\ell}|}\nonumber\\
& = (G + (\id - G))^{\ell-1}\nonumber\\
& = \id\;. \label{eq:S-bound}
\end{align}
In addition, for any $\tau_{\ell} \in \{0, 1\}$,
\begin{align}
\Big(\E_{x_{\ell}} \sum_{g_{\ell} \in \outc_{\tau_{\ell}}} \wG^{x_{\ell}}_{g_{\ell}}\Big)
& = \left\{ \begin{array}{cl}
		G & \text{if } \tau_{\ell} = 1\;,\\
		(\id - G) & \text{if } \tau_{\ell} = 0\;,
	\end{array}\right.\nonumber\\
&= G^{\tau_{\ell}} \cdot (\id - G)^{1 - \tau_{\ell}}\;.\label{eq:explicit-formula-for-G-expectation}
\end{align}
Thus, for any $\tau_{> \ell}$,
\begin{align}
\sum_{\tau_{\ell}} 
	S_{\tau_{\geq \ell}} \cdot \Big(\E_{x_{\ell}} \sum_{g_{\ell} \in \outc_{\tau_{\ell}}} \wG^{x_{\ell}}_{g_{\ell}}\Big)
&= \sum_{\tau_{\ell}}  S_{\tau_{\geq \ell}} \cdot(G^{\tau_{\ell}} \cdot (\id - G)^{1 - \tau_{\ell}})\nonumber\\
&= \sum_{\tau_{\ell}}  \sum_{\tau_{< \ell} : |\tau| \geq t} G^{|\tau_{<\ell}|}\cdot (\id - G)^{(\ell-1)-|\tau_{<\ell}|}
	\cdot (G^{\tau_{\ell}} \cdot (\id - G)^{1 - \tau_{\ell}})\nonumber\\
&= \sum_{\tau_{\ell}}  \sum_{\tau_{< \ell} : |\tau| \geq t} G^{|\tau_{\leq \ell}|}\cdot (\id - G)^{\ell-|\tau_{\leq\ell}|}\nonumber\\
&= \sum_{\tau_{\leq \ell} : |\tau| \geq t} G^{|\tau_{\leq \ell}|}\cdot (\id - G)^{\ell-|\tau_{\leq\ell}|}\nonumber\\
&= S_{\tau_{> \ell}}\;.\label{eq:S-recurrence}
\end{align}
Finally, for any $\tau_{\geq \ell}$,
\begin{align}
&S_{\tau_{\geq \ell}}
	\cdot \Big(\E_{x_{\ell}} \sum_{g_{\ell} \in \outc_{\tau_{\ell}}} \wG^{x_{\ell}}_{g_{\ell}}\Big)
	\cdot S_{\tau_{\geq \ell}}\nonumber\\
 =~& S_{\tau_{\geq \ell}}
	\cdot (G^{\tau_{\ell}} \cdot (\id - G)^{1 - \tau_{\ell}})
	\cdot S_{\tau_{\geq \ell}} \tag{by \Cref{eq:explicit-formula-for-G-expectation}}\\
 =~& \sqrt{G^{\tau_{\ell}} \cdot (\id - G)^{1 - \tau_{\ell}}} \cdot (S_{\tau_{\geq \ell}})^2
	\cdot \sqrt{G^{\tau_{\ell}} \cdot (\id - G)^{1 - \tau_{\ell}}}\tag{because $S_{\tau_{\geq \ell}}$ commutes with $G$ and $(\id - G)$}\\
\leq~& \sqrt{G^{\tau_{\ell}} \cdot (\id -  G)^{1 - \tau_{\ell}}} \cdot I
	\cdot \sqrt{G^{\tau_{\ell}} \cdot (\id -  G)^{1 - \tau_{\ell}}} \tag{by \Cref{eq:S-bound}}\\
=~& G^{\tau_{\ell}} \cdot (\id -  G)^{1 - \tau_{\ell}}\nonumber\\
=~& \Big(\E_{x_{\ell}} \sum_{g_{\ell} \in \outc_{\tau_{\ell}}} \wG^{x_{\ell}}_{g_{\ell}}\Big)\;, \label{eq:S-sandwich}
\end{align}
where the last step uses \Cref{eq:explicit-formula-for-G-expectation} again.
This concludes the set of facts we will need about $S_{\tau_{\geq \ell}}$.

Now we prove \Cref{eq:i-think-this-is-what-i'm-supposed-to-prove-2}.
To start, we write using the definition
\begin{align}
&\E_{x_{\geq \ell}} \sum_{\tau_{\geq \ell}}  \sum_{(g_{\ell}, \ldots,g_k) \in \outc_{\tau_{\geq \ell}}}
	\tau\big(\widehat{H}^{x_{\ell}, \dots, x_{k}}_{g_{\ell}, \dots, g_{k}} \cdot S_{\tau_{\geq \ell}}\big)\nonumber\\
=~&\E_{x_{\geq \ell}} \sum_{\tau_{\geq \ell}}  \sum_{(g_{\ell}, \ldots,g_k) \in \outc_{\tau_{\geq \ell}}}
	\tau\big(\wG^{x_{\ell}}_{g_{\ell}}\cdot (\wG^{x_{> \ell}}_{g_{> \ell}} \cdot (\wG^{x_{> \ell}}_{g_{> \ell}})^\dagger) \cdot  \wG^{x_{\ell}}_{g_{\ell}}\cdot S_{\tau_{\geq \ell}} \big)\;.\label{eq:move-g-over-there}
\end{align}
Next, we commute the leftmost $\wG^{x_{\ell}}_{g_{\ell}}$ to the right. Using \Cref{cor:G-hat-facts} and \Cref{lem:switcheroo} we get
\[ \wG^{x_{\ell}}_{g_{\ell}}\cdot \wG^{x_{> \ell}}_{g_{> \ell}} \approx_{k\nu'_2}  \wG^{x_{> \ell}}_{g_{> \ell}}\cdot \wG^{x_{\ell}}_{g_{\ell}}\]
and similarly 
\[ \wG^{x_{\ell}}_{g_{\ell}}\cdot (\wG^{x_{> \ell}}_{g_{> \ell}})^\dagger \approx_{k\nu'_2}  (\wG^{x_{> \ell}}_{g_{> \ell}})^\dagger \cdot \wG^{x_{\ell}}_{g_{\ell}}\;.\]
Combining with \Cref{lem:closeness-to-close-ips} and injecting back into~\eqref{eq:move-g-over-there} we get
\begin{equation}\label{eq:move-g-over-there-2}
\E_{x_{\geq \ell}} \sum_{\tau_{\geq \ell}}  \sum_{(g_{\ell}, \ldots,g_k) \in \outc_{\tau_{\geq \ell}}}
	\tau\big(\widehat{H}^{x_{\ell}, \dots, x_{k}}_{g_{\ell}, \dots, g_{k}} \cdot S_{\tau_{\geq \ell}}\big)
	\approx_{2k\nu'_2} \E_{x_{\geq \ell}} \sum_{\tau_{\geq \ell}}  \sum_{(g_{\ell}, \ldots,g_k) \in \outc_{\tau_{\geq \ell}}}
	\tau\big((\wG^{x_{> \ell}}_{g_{> \ell}} \cdot (\wG^{x_{> \ell}}_{g_{> \ell}})^\dagger) \cdot  \wG^{x_{\ell}}_{g_{\ell}}\cdot S_{\tau_{\geq \ell}} \big)\;.
\end{equation}
We end by noting that
\begin{align*}
\eqref{eq:move-g-over-there-2}
& =
\E_{x_{\geq \ell}} \sum_{\tau_{\geq \ell}}  \sum_{(g_{\ell}, g_{>\ell}) \in \outc_{\tau_{\geq \ell}}}
	\tau\big(\wH^{x_{> \ell}}_{g_{> \ell}} \cdot (S_{\tau_{\geq \ell}} \cdot \wG^{x_{\ell}}_{g_{\ell}})\big)\\
& =\E_{x_{> \ell}} \sum_{\tau_{> \ell}}  \sum_{g_{>\ell} \in \outc_{\tau_{> \ell}}}
	\tau\big(\wH^{x_{> \ell}}_{g_{> \ell}} \cdot \Big(\sum_{\tau_{\ell}}S_{\tau_{\geq \ell}} \cdot
		\Big(\E_{x_{\ell}} \sum_{g_\ell \in \outc_{\tau_{\ell}}}\wG^{x_{\ell}}_{g_{\ell}}\Big)\Big)\big)\\
& =\E_{x_{> \ell}} \sum_{\tau_{> \ell}}  \sum_{g_{>\ell} \in \outc_{\tau_{> \ell}}}
	\tau\big(\wH^{x_{> \ell}}_{g_{> \ell}} \cdot S_{\tau_{>\ell}}\big)\;. \tag{by \Cref{eq:S-recurrence}}
\end{align*}
This concludes the proof of
\Cref{eq:i-think-this-is-what-i'm-supposed-to-prove-2}
and therefore proves the lemma.
\end{proof}

The following lemma follows from the functional calculus for bounded self-adjoint operators. In the lemma, $\algebra$ may be any von Neumann subalgebra of $\Bounded(\hilb)$ and $\tau$ a normal tracial state on $\algebra$, but we state (and use it) for the algebra and tracial state associated with the strategy $\strategy$. 

\begin{lemma}
  \label{lem:chernoff-bernoulli-matrix}
  Let $0 < \theta < 1$ and let $k, t > 0$ be integers such that $k \geq 2t/\theta$.
  Define the function $F: \R \to \R$ by
  \[ F(x) = \sum_{r = t}^{k} \binom{k}{r} x^r (1 -  x)^{k - r}. \]
  Then for any self-adjoint $X\in \algebra$ such that $0 \leq X \leq \id$ and $\tau\big(X 
  \big) \geq 1 - \kappa$, it holds that $F(X)\in \algebra$ and moreover
  \[ \tau\big(F(X) \big) \geq 1 - \frac{\kappa}{1 - \theta}
    - e^{-\theta^2 k/2}. \]
  \end{lemma}
	
  \begin{proof}
	The statement that $F(X)\in \algebra$ follows from the functional calculus for bounded self-adjoint operators.  
Let $\mu$ be the projection-valued measure on the Borel $\sigma$-algebra in $[0,1]$ associated to $X$ by the spectral theorem for bounded self-adjoint operators (see e.g.~\cite[Theorem 7.12]{hall2013quantum}). Then $X = \int_{[0,1]} \lambda \, \mathrm{d}\mu(\lambda)$, and using that $\tau$ is normal, there exists a trace class operator $A \in \Bounded(\hilb)$ such that
\begin{align*}
 \tau(X) &= \tr\Big(\int_{[0,1]} \lambda  \, \mathrm{d}\mu(\lambda) A\Big)\\
&= \int_{[0,1]} \lambda \, \tr\big(A \, \mathrm{d}\mu(\lambda)\big)\\
&= \int_{[0,1]} \lambda \, \mathrm{d}\nu(\lambda)\;,
\end{align*}
where $\nu$ is the real-valued measure on $[0,1]$ defined by $\, \mathrm{d}\nu(E)=\int_{E}  \tr(A\, \mathrm{d}\mu(\lambda))$.  
Thus the 
    assumption $\tau(X)\geq 1-\kappa$  implies that
     $ \int_\lambda  \lambda \, \mathrm{d}\nu(\lambda)
       \geq 1 - \kappa$.
      By Markov's inequality, for any $0 < \theta < 1$, 
     \begin{equation}\label{eq:in-other-words}
     \int_{[0,1]} \indicator[1-\lambda \geq 1 - \theta]\, \mathrm{d}\nu(\lambda)
     \,\leq\, \frac{\int_{[0,1]}(1-\lambda) \, \mathrm{d}\nu}{1-\theta}
     \,\leq\,\frac{\kappa}{1-\theta}\;.
     \end{equation} 
    We now evaluate $\tau(F(X) )$. 
    Observe that for $0\leq p \leq 1$, $F(p)$ is precisely the 
    probability of observing at least $t$ successes out of $k$
    i.i.d.\ Bernoulli trials, each of which succeeds with probability
    $p$.
    In other words, it is the probability that $Y := Y_1 + \cdots + Y_k \geq t$,
    where $Y_1, \ldots, Y_k \sim \mathrm{Bernoulli}(p)$.
    We can bound this probability by the additive Chernoff bound (see the second additive bound in~\cite{Blu11}):
    \begin{align*}
    \Pr(Y < t)
    &= \Pr(Y < p k - (pk - t)) \\
    &=\Pr\Big(Y < pk - \Big(p - \frac{t}{k}\Big)\cdot k\Big)\\
    	&\leq  \exp\Big(-2 \Big(p - \frac{t}{k}\Big)^2 \cdot k\Big)\;.
    \end{align*}
    Thus,
    \begin{equation}\label{eq:by-chernoff}
    F(p) = \Pr(Y \geq t) = 1 - \Pr(Y < t) \geq 1 - \exp\Big(-2 \Big(p - \frac{t}{k}\Big)^2 \cdot k\Big)\;.
    \end{equation}
        Putting the pieces together, we compute $\tau(F(X) )$: 
    \begin{align}
      \tau\big(F(X) \big)         &= \tau\Big(\int_{[0,1]} F(\lambda) \, \mathrm{d}\mu(\lambda)\Big)
                                   \nonumber\\
																	&= \int_{[0,1]} F(\lambda) \, \tau\big(A \, \mathrm{d}\mu(\lambda)\big)
                                   \nonumber\\
                                   	&\geq \int_{[0,1]}  F(\lambda) \, \indicator[\lambda \geq \theta] \, \mathrm{d}\nu(\lambda) \tag{$F$ is nonnegative} \\
                                       &\geq F(\theta) \, \int_{[0,1]} \indicator[\lambda \geq \theta]  \, \mathrm{d}\nu(\lambda) \tag{$F$ is non-decreasing}\\
                                       &\geq \Big( 1 - \frac{\kappa}{1 - \theta} \Big) \cdot F(\theta)  \tag{by \Cref{eq:in-other-words}}\\
      &\geq \Big( 1 - \frac{\kappa}{1 - \theta} \Big) \cdot \Big(1 - \exp \Big(-2\Big(\theta -
        \frac{t}{k}\Big)^2 \cdot k \Big) \Big) \tag{\Cref{eq:by-chernoff}} \\
        &\geq  1 - \frac{\kappa}{1 - \theta} - \exp \Big(-2\Big(\theta -
        \frac{d}{k}\Big)^2 \cdot k \Big)\;. \label{eq:almost-done-with-this-giant-proof}
        \end{align}
    Finally, we note that because $k \geq 2t/\theta$,
    we have $\theta/2 \geq t/k$.
    This implies that
    $\theta-t/k \geq \theta - \theta/2 = \theta/2$,
    and so    
    $(\theta - t/k)^2 \geq (\theta/2)^2 = \theta^2/4$.
    As a result,
    \begin{equation*}
    \exp \Big(2\Big(\theta -
        \frac{t}{k}\Big)^2 \cdot k \Big)
        \geq 
        \exp \Big(\theta^2 k/2 \Big)\;.
    \end{equation*}
    Equivalently,
    \begin{equation*}
    \exp \Big(-2\Big(\theta -
        \frac{t}{k}\Big)^2 \cdot k \Big)
        \leq 
        \exp \Big(-\theta^2 k/2 \Big)\;.
    \end{equation*}
    Thus, we conclude
    \begin{equation*}
     \tau\big(F(X)  \big)
    \geq  1 - \frac{\kappa}{1 - \theta} - \exp \Big(-\theta^2 k/2 \Big)\;.\qedhere
    \end{equation*}
  \end{proof}
	
\begin{corollary}[Completeness of~$H$]
  \label{cor:ld-pasting-N-completeness}
Recall the error parameters $\nu_5',\nu_6'$ from \Cref{lem:over-all-outcomes} and \Cref{lem:from-H-to-G}, respectively. Let $\nu = \nu_5' + \nu_6'$ and let $k \geq 12mt$ be an integer. Then $\nu$ is equal to some polynomial function $\poly(m,t,k) \cdot \poly(\eps,\delta,\zeta,n^{-1})$ and furthermore
  \[ \tau\big(H \big)
\geq 1 - \kappa \cdot \left(1 + \frac{1}{3m}\right)
    - \nu - e^{- \frac{k}{72m^2}}\;. \]
\end{corollary}

\begin{proof}
We begin by approximating the completeness as follows.
\begin{align*}
\tau\big(H  \big)
& \approx_{\nu'_5} \E_{x_1, \ldots, x_k} \sum_{\tau:|\tau| \geq t} \sum_{(g_1, \ldots, g_k) \in \mathsf{Outcomes}_{\tau}} \tau\big(\widehat{H}^{x_1, \ldots, x_k}_{g_1, \ldots, g_k}\big) \tag{by \Cref{lem:over-all-outcomes}}\\
& \approx_{\nu'_6} \sum_{i = d+1}^k \binom{k}{i} \tau\big(G^i (\id - G)^{k-i}\big) \tag{by \Cref{lem:from-H-to-G}}\\
& = \tau\big(F(G)  \big)\
\end{align*}
where $F$ is the function from \Cref{lem:chernoff-bernoulli-matrix}. We unroll the error parameters in terms of $\eps,\delta,\zeta,m,n,t,k$:
\begin{gather*}
\gamma_m = 1 - \frac{d^m}{n^m} \leq \frac{mt}{n} \tag{\Cref{prop:distance}} \\
\nu_1 = 8(\sqrt{\zeta} + \sqrt{(m+1)\delta}) \tag{\Cref{lem:g-comm-after-eval}} \\
\nu_2 = 4(\gamma_m + \nu_1) \tag{\Cref{lem:g-comm}} \\
\nu_2' = 27 \nu_2^{1/4}\tag{\Cref{cor:G-hat-facts}} \\
\nu_3' = k\nu''_3 + k^2/n \tag{\Cref{lem:h-b-consistency}} \\
\nu_4' = \nu'_3 + \sqrt{2 (m+1) \eps} \tag{\Cref{cor:ha-cons}} \\
\nu_5' = 2\frac{k^2}{n}  + k\nu''_3 + \gamma_m \tag{\Cref{lem:over-all-outcomes}} \\
\nu_6' = 2k^2\nu'_2 \tag{\Cref{lem:from-H-to-G}}
\end{gather*}
One can see that each of the error parameters $\gamma_m,\nu_1,\ldots,\nu_6'$ can be expressed in the form $\poly(m,t,k) \cdot \poly(\eps,\delta,\zeta,n^{-1})$ where each error parameter is a different polynomial function. Thus letting $\nu = \nu_5' + \nu_6' = \poly(m,t,k) \cdot \poly(\eps,\delta,\zeta,n^{-1})$, we have that 
\[
	\tau(H) \geq \tau(F(G)) - \nu~.
\]
Define $\theta = \frac{1}{6m}$. Note that $k\geq 2t/\theta$ by assumption, and that $\frac{1}{1 - \theta} \leq 1 + \frac{1}{3m}$.
As a result, $k$ and~$\theta$ satisfy the hypotheses of \Cref{lem:chernoff-bernoulli-matrix}, which implies
\begin{align*}
\tau\big(F(G)  \big)
&\geq 1 - \frac{\kappa}{1 - \theta}
    - e^{-\theta^2 k/2}\\
&\geq 1 - \frac{\kappa}{1 - \theta}
    - e^{- \frac{k}{72m^2}}\\
&\geq 1 - \kappa \cdot \left(1 + \frac{1}{3m}\right)
    - e^{- \frac{k}{72m^2}}.
\end{align*}
In total, we have
\begin{equation*}
\tau\big(H \big)
\geq 1 - \kappa \cdot \left(1 + \frac{1}{3m}\right)
    - \nu - e^{- \frac{k}{72m^2}}.
\end{equation*}
This completes the proof.
\end{proof}

We conclude by observing that the conclusion of \Cref{lem:pasting} follows immediately by letting $k$ denote the parameter $r$ in the statement of \Cref{lem:pasting} \footnote{We use $r$ in the statement of \Cref{lem:pasting} in order to avoid confusion with the dimension of the code $\code$, which is also denoted by $k$.}, combining the bounds from \Cref{cor:ha-cons}, and \Cref{cor:ld-pasting-N-completeness} using the same ``completion'' procedure as described in Section~\ref{sec:together}.

\appendix

\section{Operators on expander graphs}

Let $G$ be an $n$-vertex graph with some distribution on the edges, such that the marginal distribution on the vertices is uniform. Let
\[
	K = \E_{(u,v) \sim G} \ketbra{u}{v}
\]
and let 
\[
	L = \frac{1}{n} \Id - K
\]
denote the normalized Laplacian of $G$. 

\begin{proposition}
$L = \frac{1}{2} \cdot \E_{(u,v) \sim G} ( \ket{u} - \ket{v})(\bra{u} - \bra{v})$.
\end{proposition}

Thus $L$ is positive semidefinite and has $0$-eigenvector
\[
	\ket{\varphi_0} = \frac{1}{\sqrt{n}}\sum_{u\in V(G)} \ket{u}
\]
where $V(G)$ denotes the vertex set of $G$.

Let $0 = \lambda_1 \leq \lambda_2 \leq \cdots \leq \lambda_n$ denote the eigenvalues of $L$ in non-decreasing order. Then we have that
\begin{equation}
\label{eq:expansion-lower-bound}
	L \geq \lambda_2 \, (\Id - \ketbra{\varphi_0}{\varphi_0}).
\end{equation}

The following lemma is a variant of the familiar Poincar\'e inequality for expander graphs. It states that ``global'' averages of an operator-valued function on an expander graph can be approximated by averages of the function on random edges of the graph. 

\begin{lemma}\label{lem:local-to-global-expander}
  Let $\hilb$ be a Hilbert space and let $\rho$ be a positive linear functional
  on $\Bounded(\hilb)$.
  Let $\{A^u\}_{u \in [n]}$ be bounded operators on $\hilb$.
  Suppose $G$ is an $n$-vertex graph associated with an edge distribution that
  has the uniform marginal distribution on the vertices.
  Then
	\[
		\E_{u,v \sim [n]} \, \rho( (A^u - A^v)^* (A^u - A^v)) \leq \frac{1}{n \lambda_2} \E_{(u,v) \sim G} \,  \rho( (A^u - A^v)^* (A^u - A^v)) 
	\]
	where the expectation on the left hand side is over uniformly random vertices $u,v \sim [n]$. 
\end{lemma}
\begin{proof}
	Define the following map $V: \hilb \to \C^n \otimes \hilb$:
	\[
		V = \sum_u \ket{u} \otimes A^u.
	\]
	The proof is based on estimating the Rayleigh quotient 
	\[ \frac{V^* (L \otimes \Id) V}{V^*V}\;.\]
	Observe that 
	\begin{equation}
	\label{eq:expansion-bound-1}
		V^* (L \otimes \Id) V = \frac{1}{2} \E_{(u,v) \sim G} (A^u - A^v)^* (A^u - A^v).
	\end{equation}
	On the other hand, we have
	\begin{align}
		V^* (L \otimes \Id) V &\geq \lambda_2 \cdot V^* \Big (\Id - \ketbra{\varphi_0}{\varphi_0} \Big ) V & \text{\Cref{eq:expansion-lower-bound}} \\
		&\geq \lambda_2 \cdot \Big ( \sum_u (A^u)^* A^u - \big(\sum_v n^{-1/2} A^v\big)^* \big(\sum_v n^{-1/2} A^v\big) \Big) \\
		&= n \lambda_2 \cdot  \Big ( \E_{u \sim [n]}  (A^u)^* A^u - \overline{A}^* \overline{A} \Big) \\
		&= n \lambda_2 \cdot \Big ( \E_{u \sim [n]}  (A^u)^* (A^u - \overline{A}) \Big) \\
		&= \frac{n \lambda_{2}}{2} \cdot \Big ( \E_{u,v \sim [n]}  (A^u - A^v)^* (A^u - A^v) \Big) \label{eq:expansion-bound-2} \;.
	\end{align}
	where in the second line we computed $V^* V = \sum_u (A^u)^* A^u$ and $V^*\ket{\varphi_0} = \sum_u n^{-1/2} (A^u)^*$, and in the third line we defined $\overline{A} = \E_{v \sim [n]} A^v$. Applying the positive linear functional $\rho$ on \Cref{eq:expansion-bound-1} and \Cref{eq:expansion-bound-2} yields the statement of the Lemma.
\end{proof}

\newcommand{\etalchar}[1]{$^{#1}$}


\begin{dajauthors}
\begin{authorinfo}[zji]
  Zhengfeng Ji\\
  Tsinghua University\\
	jizhengfeng\imageat{}tsinghua\imagedot{}edu\imagedot{}cn
\end{authorinfo}
\begin{authorinfo}[anand]
  Anand Natarajan\\
  Massachusetts Institute of Technology\\
  anandn\imageat{}mit\imagedot{}edu \\
\end{authorinfo}
\begin{authorinfo}[thomas]
  Thomas Vidick\\
  California Institute of Technology\\
  vidick\imageat{}caltech\imagedot{}edu
\end{authorinfo}
\begin{authorinfo}[john]
  John Wright\\
  University of Texas at Austin\\
  wright\imageat{}cs\imagedot{}utexas\imagedot{}edu
\end{authorinfo}
\begin{authorinfo}[henry]
  Henry Yuen\\
  Columbia University\\
  hyuen\imageat{}cs\imagedot{}columbia\imagedot{}edu
\end{authorinfo}
\end{dajauthors}

\end{document}